\documentclass[sigplan,screen]{acmart}
\copyrightyear{2020} 
\acmYear{2020} 
\setcopyright{acmlicensed}
\acmConference[LICS '20]{Proceedings of the 35th Annual ACM/IEEE Symposium on Logic
  in Computer Science (LICS)}{July 8--11, 2020}{Saarbr\"ucken, Germany}
\acmPrice{15.00}
\acmDOI{10.1145/3373718.3394806}
\acmISBN{978-1-4503-7104-9/20/07}
\usepackage{booktabs}   %% For formal tables:
\usepackage[T1]{fontenc} % correct output rendering as first packages
\usepackage[utf8]{inputenc} % correct input typing as second package
\usepackage[english]{babel} 
\usepackage{geometry} %for  margins
\usepackage{csquotes} %virgolette
\usepackage[toc,page]{appendix}
\usepackage{amsfonts}
\usepackage{amssymb}
\usepackage{amsmath}
\usepackage{amsthm}
\usepackage{stmaryrd} 
\usepackage{mathrsfs} % for  mathscr 
\usepackage{mathpartir} % display proof systems and proofs
\usepackage{bussproofs} %display proofs
\usepackage{mathtools}
\usepackage{cmll} %linear logic
\usepackage[size=small,captionskip=20pt]{subfig} %replaced by subcaption
\usepackage[size=normalsize,skip=20pt]{caption}
\usepackage{framed} % to give a frame for figures
\usepackage{graphicx} %to include graphics
\usepackage{adjustbox} % to adjust graphics
\usepackage{lscape}%a full horizontal page
\usepackage{afterpage}%eliminates white spaces before lscape 
\usepackage{longtable} % break tables (to break arrays use \allowdisplaybreaks)
\usepackage[shortlabels]{enumitem} % allows to reference enumerate's items and to design items
\usepackage{hyperref} %hyperreference
\usepackage{xcolor} %coloring text and math expressions
\definecolor{mygreen}{rgb}{0,0.6,0}  %to set comments 
\definecolor{mygray}{rgb}{0.5,0.5,0.5}
\definecolor{mymauve}{rgb}{0.58,0,0.82}
\hypersetup{
    colorlinks,
    linkcolor={red!50!black},
    citecolor={blue!50!black},
    urlcolor={blue!80!black}  }
\usepackage{cancel}  %% \bcancel{expression}
\usepackage{tikz}
\usetikzlibrary{arrows.meta}
\usetikzlibrary{decorations.markings}
\usetikzlibrary{matrix}
\usetikzlibrary{arrows}
\usetikzlibrary{decorations.pathmorphing}
\usetikzlibrary{shapes}
\usetikzlibrary{arrows}
\usetikzlibrary{positioning} % draw simple graphs
\usepackage{circuitikz} %circuits
\usepackage{tikz-cd} %categories
\theoremstyle{plain}
\newtheorem{thm}{Theorem}
\newtheorem{lem}[thm]{Lemma}
\newtheorem{prop}[thm]{Proposition}
\newtheorem{cor}[thm]{Corollary}

\newtheorem{fact}[thm]{Fact}

\theoremstyle{definition}
\newtheorem{defn}{Definition}
\newtheorem{exmp}{Example}

\theoremstyle{remark}
\newtheorem{rem}{Remark}

\newcommand{\plam}{\Lambda_{\oplus}}
\newcommand{\elam}{\Lambda_{\oplus}^\emptyset}
\newcommand{\glam}{\Lambda_{\oplus}^\Gamma}

\newcommand{\clam}{\mathsf{C}\plam}
\newcommand{\alam}{\mathsf{A}\plam}
\newcommand{\genlam}{\mathsf{G}\plam}
\newcommand{\inter}[1]{\llbracket  {#1} \rrbracket}

\newcommand{\R}[2]{  {#1} \ \mathcal{R} \ {#2}}

\newcommand{\wt}[1]{\widetilde{#1}}
\newcommand{\h}{\mathrm{HNF}}
\newcommand{\he}{\mathrm{HNF}^\emptyset}
\newcommand{\hg}{\mathrm{HNF}^\Gamma}
\newcommand{\neu}{\mathrm{NEUT}}
\newcommand{\dhe}{\mathrm{\widetilde{HNF}}}
\newcommand{\ww}[1]{\widetilde{#1}}
\newcommand{\M}[1]{\mathcal{P}_{\oplus}({#1})}
\newcommand{\dist}{\mathfrak{D}}
\newcommand{\cont}[1]{\mathcal{E}[ {#1}]}
\newcommand{\contrat}[1]{\lambda \vec{x}. {#1}\vec{L}}
\newcommand{\mes}[1]{\Vert {#1}\Vert}

\begin{document}

%% Title information
\title[]{The Benefit of Being Non-Lazy\\ in Probabilistic $\lambda$-calculus}
%{The Discriminating Power of the Head Reduction Strategy in the Probabilistic $\lambda$-calculus}         %% [Short Title] is optional;
                                        %% when present, will be used in
                                        %% header instead of Full Title.
%\titlenote{with title note}             %% \titlenote is optional;
                                        %% can be repeated if necessary;
                                        %% contents suppressed with 'anonymous'
\subtitle{Applicative Bisimulation is Fully Abstract for Non-Lazy Probabilistic Call-by-Name}                     %% \subtitle is optional
%\subtitlenote{with subtitle note}       %% \subtitlenote is optional;
                                        %% can be repeated if necessary;
                                        %% contents suppressed with 'anonymous'

%% Author information
%% Contents and number of authors suppressed with 'anonymous'.
%% Each author should be introduced by \author, followed by
%% \authornote (optional), \orcid (optional), \affiliation, and
%% \email.
%% An author may have multiple affiliations and/or emails; repeat the
%% appropriate command.
%% Many elements are not rendered, but should be provided for metadata
%% extraction tools.

%% Author with single affiliation.
\author{Gianluca Curzi}
%\authornote{with author1 note}          %% \authornote is optional;
                                        %% can be repeated if necessary
%\orcid{nnnn-nnnn-nnnn-nnnn}             %% \orcid is optional
\affiliation{
%  \position{Position1}
%  \department{Department1}              %% \department is recommended
  \institution{University of Turin}            %% \institution is required
%  \streetaddress{Street1 Address1}
  \city{Turin}
%  \state{State1}
%  \postcode{Post-Code1}
  \country{Italy}                    %% \country is recommended
}
\email{gcurzi@unito.it}          %% \email is recommended

%% Author with two affiliations and emails.
\author{Michele Pagani}
%\authornote{with author2 note}          %% \authornote is optional;
                                        %% can be repeated if necessary
%\orcid{nnnn-nnnn-nnnn-nnnn}             %% \orcid is optional
\affiliation{
%  \position{Position2a}
%  \department{Department2a}             %% \department is recommended
  \institution{IRIF UMR 8243, Universit\'e de Paris, CNRS}           %% \institution is required
%  \streetaddress{Street2a Address2a}
  \city{Paris}
%  \state{State2a}
%  \postcode{Post-Code2a}
  \country{France}                   %% \country is recommended
}
\email{pagani@irif.fr}         %% \email is recommended
%\affiliation{
%  \position{Position2b}
%  \department{Department2b}             %% \department is recommended
%  \institution{Institution2b}           %% \institution is required
%  \streetaddress{Street3b Address2b}
%  \city{City2b}
%  \state{State2b}
%  \postcode{Post-Code2b}
%  \country{Country2b}                   %% \country is recommended
%}
%\email{first2.last2@inst2b.org}         %% \email is recommended

%% Abstract
%% Note: \begin{abstract}...\end{abstract} environment must come
%% before \maketitle command
\begin{abstract}
We consider the probabilistic applicative bisimilarity (PAB) --- a coinductive relation comparing the applicative behaviour of probabilistic untyped $\lambda$-terms according to a specific operational semantics. This notion has been studied by Dal Lago et al. with respect to the two standard parameter passing policies, call-by-value (cbv) and call-by-name (cbn), using a lazy reduction strategy not reducing within the body of a function. In particular, PAB has been proven to be fully abstract with respect to the contextual equivalence in cbv \cite{crubille2014probabilistic} but not in lazy cbn \cite{dal2014coinductive}. 

We overcome this issue of cbn by relaxing the laziness constraint: we prove that PAB is fully abstract with respect to the standard head reduction contextual equivalence. Our proof is based on Leventis' Separation Theorem \cite{leventis2018probabilistic}, using probabilistic Nakajima trees as a tree-like representation of the contextual equivalence classes.

Finally, we prove also that the inequality full abstraction fails, showing that the probabilistic applicative similarity is strictly contained in the contextual preorder. 
\end{abstract}

%% 2012 ACM Computing Classification System (CSS) concepts
%% Generate at 'http://dl.acm.org/ccs/ccs.cfm'.

%% End of generated code

%% Keywords
%% comma separated list
  %% \keywords are mandatory in final camera-ready submission

%% \maketitle
%% Note: \maketitle command must come after title commands, author
%% commands, abstract environment, Computing Classification System
%% environment and commands, and keywords command.
\maketitle

\section{Introduction}

The probabilistic $\lambda$-calculus $\plam$ extends the pure untyped $\lambda$-calculus with a sum $M\oplus N$, evaluating to $M$ or $N$ with equal probability $0.5$. 
The operational semantics gives then a function mapping a term $M$ to a probability distribution $\inter M$ of values. Exactly as in standard $\lambda$-calculus, different design choices may affect the meaning $\inter M$ of a term. 

%in a small-step semantics as Markov chain, or in a big-step semantics as a function associating with a term $M$ a probability distributions $\inter M$ of values. In this paper, we mainly use this latter definition (see Definition \ref{}), but two presentation are equivalent \cite{}. 

First, one has to decide how to evaluate a \emph{$\beta$-redex}, i.e.~the application of a function $\lambda x.M$ to an argument $N$. There are two main evaluation mechanisms: the \emph{call-by-value policy} (cbv) consists first in evaluating $N$ to some value $V$ and then replacing the parameter $x$ in $M$ with $V$, while the \emph{call-by-name policy} (cbn) replaces $x$ with $N$ as it is, before any evaluation. It is well-known that the two policies give rise to different results, especially in a probabilistic setting. Consider for example the term $(\lambda vz. vv)(\mathbf{T} \oplus \mathbf{F})$, where $\mathbf{T}= \lambda xy.x$ and $\mathbf{F}= \lambda xy.y$. In cbv, we first evaluate  $\mathbf{T} \oplus \mathbf{F}$, yielding either  $\mathbf{T}$ or $\mathbf{F}$ with equal probability, and then we pass the result to the function $\lambda vz. vv$, producing either $\lambda z.\mathbf{T}\mathbf{T}$ or $\lambda z.\mathbf{F}\mathbf{F}$, both with probability $0.5$.  By contrast, in cbn we pass the whole term $\mathbf{T} \oplus \mathbf{F}$ to the function before evaluating it,  obtaining $\lambda z.(\mathbf{T} \oplus \mathbf{F})(\mathbf{T} \oplus \mathbf{F})$ with probability $1$.   

Second, one has to define which redexes to evaluate in a term, if any. Also in this case, there are various choices in $\lambda$-calculus: the \emph{lazy strategy}, forbidding any reduction in the body of a function, so that $\lambda x.M$ is a value whatever $M$ is, or the \emph{head reduction}, consisting in reducing the redex in head position, which is at the left of any application. Again, the choice gives rise to different meanings, the meaning of a term w.r.t. the head reduction is a distribution of head normal forms. %In the above example, the head reduction evaluations further the terms under $\lambda z$, eventually obtaining after some steps $\inter M$ giving $\lambda zy.\mathbf T$ or $\lambda zy.y$ with equal probability $\frac 12$, while $\inter N$ will yield $\lambda zy.y$ with probability $\frac 12$, and $\lambda zy.\mathbf F$ or $\lambda zy.\mathbf T$ each with probability $\frac 14$.

By the way, let us remark here that some variants of the standard head reduction have been considered in the literature, as for example the \emph{head spine reduction} that, given a $\beta$-redex $(\lambda x.M)N$, first evaluates the body of $M$ and then evaluates the outermost redex according to cbn. A side result of our paper is that the head and head spine strategies are actually equivalent, even in a probabilistic setting (Theorem~\ref{thm: equivalence head e head spine nel paper}). 

Comparing terms by their operational semantics is too narrow,  as higher-order normal forms differ often by syntactical details that are inessential with respect to their computational behaviour. Contextual equivalence is usually considered: two terms $M,N$ are contextually equivalent  ($M=_{\mathrm{cxt}} N$ in symbols) whenever they ``behave'' the same in any  possible ``programming context''. This definition depends on the notion of \emph{context} and on that of \emph{observable behavior}. In $\plam$, a context $\mathcal{C}$ is a term with a special variable $[\cdot]$, the \emph{hole}, and what we observe is the total mass of the distribution $\inter{\mathcal{C}[M]}$, i.e.~the total probability of getting a result from the evaluation of the term $\mathcal{C}[M]$ obtained by replacing the hole with $M$. The definition of $=_{\mathrm{cxt}}$ depends therefore on the chosen operational semantics but it is more canonical than the latter. 

Proving that two terms are contextually equivalent is rather difficult since we have to consider \emph{all} contexts, hence the quest for more tractable equivalences comparable with $=_{\mathrm{cxt}}$. We say in particular that an equivalence $\equiv$ over $\lambda$-terms  is \emph{sound} with respect to $=_{\mathrm{cxt}}$  whenever the former implies the latter (i.e.~$\equiv \, \subseteq\,  =_{\mathrm{cxt}}$), it is \emph{complete} if the converse holds (i.e.~$=_{\mathrm{cxt}}\, \subseteq\, \equiv$) and it is \emph{fully abstract} if it is both sound and complete, i.e.~the two relations coincide. 

In probabilistic $\lambda$-calculus, the first results in this line of research have been achieved in the setting of the denotational semantics of the $\plam$ head reduction. In particular, Ehrhard et al. prove that the equivalence $\equiv_{\mathcal D^\infty}$ induced by the reflexive object $\mathcal D^\infty$ of the cartesian closed category of probabilistic coherence spaces \cite{EhrPagTas11} (as well as of the weighted relations \cite{LairdMMP13}) is sound. More recently, Leventis proves a fundamental separation theorem, giving as a consequence that the probabilistic Nakajima tree equality is complete \cite{leventis2018probabilistic}. From the  latter result, Clairambault and Paquet derive a fully abstract game model of $\plam$ and as a corollary also the full abstraction of $\mathcal D^\infty$ \cite{ClairambaultP18}. The latter result has been also achieved independently by Leventis and Pagani \cite{leventis2019strong}. 

All the above results deal with the head reduction, i.e.~ a non-lazy cbn operational semantics. For lazy strategies, a different approach is available, based on the notion of 
%In a series of work by Dal Lago et al. inaugurate a different approach for studying the lazy semantics of $\plam$, based on 
 \emph{applicative bisimulation}, which is the true object of this paper. The idea dates back to \cite{abramsky1990lazy} and consists in looking at the operational semantics as a transition system having $\lambda$-terms as states and transitions given by the evaluation of the application between $\lambda$-terms. The benefit of this setting is to transport into $\lambda$-calculus the whole theory of bisimilarity and its associated coinductive reasoning, which is a fundamental tool for comparing processes in concurrency theory. Basically, two terms $M$ and $N$ are applicative bisimilar (in symbols $M\sim N$) whenever their applications $MP$ and $NP$ reduce to applicative bisimilar values for any argument $P$. 

This approach has been lifted to the probabilistic $\lambda$-calculus in a series of works by Dal Lago et al. \cite{dal2014coinductive,crubille2014probabilistic,crubille2015applicative}, introducing the notion of \emph{probabilistic applicative bisimilarity} (PAB) for lazy semantics. In particular, PAB  is proven to be sound with the contextual equivalence in both cbv and cbn, but only cbv PAB is fully abstract. In case of lazy cbn, we have terms like: 
\begin{align}\label{ex:no_cbn_fa}
M&\triangleq  \lambda xy.(x\oplus y)&
N&\triangleq (\lambda xy.x)\oplus (\lambda xy.y)
\end{align}
such that $M =_{\mathrm{cxt}} N$ but $M\not\sim N$.
 In fact, lazy PAB is able to discriminate between a term where a choice can be performed \emph{before} any interaction, like $N$, and a term that needs to interact in order to trigger a choice, like $M$.
% In fact, lazy PAB allows for discriminating between a process giving a value \emph{allowing} two choices and a process giving two values \emph{after} a choice.
Notice that this difference is caught also by cbv contextual semantics, as these two terms are distinguished by the context $\mathcal{C}=(\lambda v. (v\mathbf I \mathbf\Omega)(v\mathbf I \mathbf\Omega))[\cdot]$ in cbv (the total mass of $\inter{\mathcal{C}[M]}_{\mathrm{cbv}}$ is $0.25$, while that of $\inter{\mathcal{C}[N]}_{\mathrm{cbv}}$ is $0.5$), but not in cbn (namely, $\inter{\mathcal{C}[M]}_{\mathrm{cbn}}=\inter{\mathcal{C}[N]}_{\mathrm{cbn}}$ has mass $0.25$).

In \cite{dal2014coinductive} the authors analyse this example remarking that the cbn policy misses the ``capability to copy a term \emph{after} having evaluated it''. This is indeed a fundamental primitive in probabilistic programming: when implementing a probabilistic algorithm we need often to toss a coin and then to pass \emph{the result} of this tossing to several subroutines. It is so common to extend a probabilistic language with a \texttt{let-in} constructor, often called \emph{sampling}, evaluating a choice \emph{before} passing it to a function even in a cbn semantics. As expected, it is shown~\cite{kasterovic2019discriminating} that such an extension recovers cbn PAB full abstraction, as terms like \eqref{ex:no_cbn_fa} become contextually different. 

Let us remark that we are here in front of two disconcerting facts. First,  it has been proven that in simply typed languages the presence of the \texttt{let-in} constructor does not affect the discriminating power of the contextual equivalence, for example in probabilistic PCF the lazy cbn contextual equivalence coincides with the equality in the model of probabilistic coherence spaces \cite{EhrPagTas14,EhrhardPT18}, with or without a sampling primitive. Why this striking difference with an untyped framework? Second, we have already mentioned several denotational models of $\plam$ which are fully abstract with respect to a pure cbn contextual equivalence, so without this ``capability to copy a term \emph{after} having evaluated it''. Is it really so necessary for getting a fully abstract PAB?

The first question can be easily answered by focussing on the laziness constraint of the operational semantics. Every $\lambda$-abstraction is a value for a lazy semantics. This does not affect the set of observables in a simply typed setting (as PCF), because this is defined on ground types (booleans, numerals, etc). By contrast, every term is a function in an untyped setting, so the laziness radically changes what we can observe in the behaviour of a term. The goal of this paper is to show that also the second question deals with laziness: we prove that PAB is fully abstract for the head reduction (Theorem~\ref{thm: completeness}). This is unexpected: non-lazy semantics seems not needing the sampling primitives in order to have fully abstract PAB, even with a cbn policy and an untyped setting.% so restoring the perfect correspondence for non-lazy cbn that we have with the numerous denotational models \cite{}. 

On a more technical side, we stress that our proofs of soundness and completeness follow a different reasoning than the one used in probabilistic lazy semantics \cite{crubille2014probabilistic,crubille2015applicative,kasterovic2019discriminating}. First, the soundness ($\sim\,  \subseteq \, =_{\mathrm{cxt}}$) does not need an Howe lifting \cite{howe1996proving}, as we prove a Context Lemma (Lemma~\ref{lem: context lemma}) for $=_{\mathrm{cxt}}$ and an applicative property of $\sim$ (Lemma~\ref{lem: context lemma 4}), the latter using the notion of  probabilistic assignments as in \cite{dal2014coinductive}. Second, and more fundamental, the proof of completeness ($=_{\mathrm{cxt}}\, \subseteq \, \sim$) is not achieved by transforming PAB into a testing equivalence using a theorem by van Breugel et al. \cite{VANBREUGEL2005}. Rather, we use Leventis' Separation property \cite{leventis2018probabilistic} to prove that  the contextual equivalence  is a probabilistic applicative bisimulation and so contained in PAB by definition (Theorem~\ref{thm: completeness}). 

%This in fact requires a major endeavour in order to represent duplicating tests as contexts. Notice that in this point, again the problem of having context allowing to duplicates values and not computations, is arides. So our solution is to circumvent this problem by 

%Let us also stress that along the way to this result, we also prove the equivalence between two different reduction strategy: the standard head-reduction, firing the leftmost outermost redex, and the so-called \emph{spine}-reduction (\cite{}), which, given an head-redex, evaluation the body of its $\lambda$ into an head-normal form, before evaluting the head-redex. In fact, the applicative bisimilation is more natural to define with respect to the spine-reduction, while Leventi's theorem is stated for standard head-reduction. Theorem~\ref{} proves that the two strategies give the same normal forms with the same probability. 

What about inequalities? All equivalences so far introduced have an asymmetric version: the contextual preorder and the probabilistic applicative similarity (PAS). We prove also that PAS is sound but not  complete with respect to the contextual inequality. A counterexample to the full abstraction in the asymmetric case is given in Section~\ref{sec4} and it is further discussed in the Conclusion.
 % it is based on the two terms in \eqref{ex:no_cbn_infa}. %This point will be discussed in the Conclusion. 

Many proofs are postponed in the Appendix.

\paragraph{Notation.}
We write $\mathbb{N}$ for the set of natural numbers, $\mathbb{R} $ for the set of real numbers and $[0,1]$ for the unit interval of $\mathbb{R} $.

 A \textit{subprobability distribution over a countable set $X$} is a function $f: X\to [0,1]$ such that $\sum_{x \in X}f(x)\leq 1$.  Distributions are ranged over by $\mathscr{D}, \mathscr{E}, \mathscr{F}, \ldots$ and $\mathcal{D}(X)$ denotes the set of all subprobability distributions over $X$.  Given a distribution $\mathscr{D}\in \mathcal{D}(X)$, its \textit{support} $\mathrm{supp}(\mathscr{D})$ is the subset of all elements in $X$ such that $\mathscr{D}(x)>0$, its \textit{mass} $\sum \mathscr{D}$ is simply $\sum_{x \in X} \mathscr{D}(x)$.     Given $x_1, \ldots, x_n\in X$, the expression $p_1 x_1+ \ldots + p_nx_n $  is used to denote the distribution $\mathscr{D}\in \mathcal{D}(X)$ with finite support $\lbrace x_1, \ldots, x_n\rbrace$ such that $\mathscr{D}(x_i)=p_i$, for every $i \leq n$. Notice that, in this case,  $\sum \mathscr{D} =\sum^n_{i=1}p_i$. The symbol $\perp$ denotes the empty distribution and $x$ can denote both an element in $X$ and the distribution having all its mass on $x$.
Given a (possibly infinite) index set $I$, a family $\lbrace r_i \rbrace_{i \in I}$ of positive real numbers such that $\sum_{i\in I}r_i\leq 1$ and a family $\lbrace \mathscr{D}_i\rbrace_{i \in I}$ of distributions,  the distribution $\sum_{i\in I} r_i \cdot \mathscr{D}_i$  is defined,  for all $x \in X$, by $(\sum_{i\in I} r_i \cdot \mathscr{D}_i)(x)= \sum_{i\in I} r_i \cdot \mathscr{D}_i(x)$.%  When $\mathscr{D}_i=\bot$ for all $i > n$, we simply write $\sum^n_{i=0} r_i \cdot \mathscr{D}_i$.

 A relation  $\mathcal{R}$ over a set $X$ is a subset of $X \times X$. Given a relation $\mathcal{R}$ over a set $X$ and $Y \subseteq X$, $\mathcal{R}(Y)$ denotes the image of  $Y$ under $\mathcal{R}$, i.e.~the set $ \lbrace x \ \vert \ \exists y \in Y \ (y, x) \in \mathcal{R} \rbrace$,  $\mathcal{R}^{op}$ represents the converse of $\mathcal{R}$, i.e.~$ \lbrace (x,y) \ \vert \ (y, x)\in \mathcal{R} \rbrace$, and $\mathcal{R}^*$ the reflexive and transitive closure of $\mathcal{R}$.
% ,  and  $\mathcal{R}^+$ (resp.  $\mathcal{R}^*$)  denotes the transitive (resp. reflexive and transitive) closure of $\mathcal{R}$
Moreover, if $\mathcal{R}$ is an equivalence relation, then $X/\mathcal{R}$ stands for the set of all equivalence classes of $X$ modulo $\mathcal{R}$.

\section{Preliminaries} \label{sec2}
This section introduces the fundamental notions of the paper. We first present   the syntax and the operational semantics of the probabilistic $\lambda$-calculus $\Lambda_\oplus$, on top of which we  shall consider  the   contextual equivalence and the   contextual preorder relations.  Then, we recall Larsen and Skou's  probabilistic (bi)similarity on  labelled Markov chains~\cite{larsen1991bisimulation} and,  in the spirit of Abramsky's work on applicative (bi)similarity~\cite{abramsky1990lazy} and following~\cite{dal2014coinductive,crubille2014probabilistic,crubille2015applicative,kasterovic2019discriminating}, we apply it to the operational semantics of $\Lambda_\oplus$, getting the probabilistic applicative (bi)similarity.

\subsection{The Probabilistic  $\lambda$-calculus $\plam$}\label{chap 4 sec 1 subsec 1}
The set $\Lambda_\oplus$ of  \emph{probabilistic $\lambda$-terms} over a given set $\mathcal{V}$  of variables is generated by the following grammar:
\begin{equation}
M,N := x  \ \vert \ \lambda x. M \ \vert \ (MN) \ \vert \ M\oplus N
\end{equation}
where $x \in\mathcal{V}$.  We consider the usual conventions as in~\cite{barendregt1984lambda}, so for example application is left-associative and has higher precedence than $\lambda$-abstraction. Parenthesis can be omitted when clear from the context.  A term is in (or is a) \emph{head normal form}, or \emph{hnf} for short, if it is of the form $\lambda x_1 \ldots x_n. yN_1 \ldots N_m$, for some $n,m\in \mathbb{N}$.   If $n=0$ then the term  is also called  \emph{neutral}.  Head normal forms are ranged over by metavariables like $H$.  The set of all hnfs will be denoted by $\h$, the set  of  all neutral terms will be denoted by $\neu$. 

 Terms are considered modulo renaming of bound variables.  The set $FV(M)$ of the free variables of a term $M$  and the capture-free substitution $M[N/x]$ of $N$  for the free occurrences of $x$ in $M$ are defined in the standard way. Finite subsets of $\mathcal{V}$  are ranged over by $\Gamma$. Given  $\Gamma$,  the set of terms (resp.~head normal forms) whose free variables are within $\Gamma$ is denoted $\glam$ (resp.~$\hg$).  
\begin{exmp}\label{exmp: examples of probabilistic terms} 
Useful terms are the identity $\mathbf I\triangleq \lambda x.x$, the boolean values $\mathbf T \triangleq \lambda xy.x $ and $\mathbf F\triangleq \lambda xy.y$, the duplicator $\mathbf \Delta\triangleq \lambda x.xx$, the Turing fixed-point combinator $\mathbf \Theta\triangleq (\lambda x. \lambda y. (y(xxy))) (\lambda x. \lambda y. (y(xxy)))$ and the ever looping  term $\mathbf \Omega\triangleq \mathbf \Delta\mathbf \Delta$. An example of probabilistic $\lambda$-term that does not belong to the standard $\lambda$-calculus is   $\mathsf{hid} \triangleq  \mathbf  I \oplus\mathbf  \Omega$.
\end{exmp}

Let $\dist(\h)$ be the set of all subprobability distributions over $\h$, called \emph{head distributions}. Let $\mathscr{D}\in\dist(\h)$, we define  $\lambda x. \mathscr{D}$ as  $(\lambda x. \mathscr{D})(H) \triangleq \mathscr{D}(H')$, if $H= \lambda x. H'$, for some $H' \in \h$, otherwise $(\lambda x. \mathscr{D})(H) \triangleq 0$.
If $X \subseteq \h$, we let $\mathscr{D}(X)\triangleq\sum_{H \in X}\mathscr{D}(H)$. We may also write $\mathscr{D}(X)$ for a generic subset $X\subseteq \plam$ of terms, meaning in fact $\mathscr{D}(X\cap\h)$.

%\begin{equation*}
% (\lambda x. \mathscr{D})(H) \triangleq \begin{cases} \mathscr{D}(H') & \text{if }H= \lambda x. H' \text{, for some }H' \in \h,  \\ 0 &\text{otherwise}   . \end{cases}
%\end{equation*}

Subprobability distributions allow us to model divergence and to look at some distributions as \enquote{approximations} of others. To  formally define this, we lift the canonical order on $\mathbb{R}$ pointwise: we set $\mathscr{D} \leq_{\dist} \mathscr{E}$ if and only if $\forall H \in \h$, $\mathscr{D}(H)\leq \mathscr{E}(H)$. Notice that $\leq_\dist$ is a directed-complete partial order over $\dist(\h)$, $\bot$ being the least element.

\begin{figure*}[ht]
\vspace{-.4cm}
\centering
\begin{framed}
\begin{mathpar} 
\inferrule*[Right=$s1$]{\\ }{M \Downarrow \bot}\and
\inferrule*[Right= $s2$]{\\ }{x \Downarrow x} \and
\inferrule*[Right=$s3$]{M \Downarrow \mathscr{D}}{\lambda x. M \Downarrow \lambda x. \mathscr{D}}\\  
\mprset{vskip=0.5ex}
\inferrule*[Right=$s4$]{M \Downarrow \mathscr{D}\\ \lbrace H[N/x]\Downarrow \mathscr{E}_{H, N} \rbrace_{\lambda x. H\, \in\,  \mathrm{ supp}(\mathscr{D})}}{MN \Downarrow \sum_{\lambda x. H\,  \in \, \mathrm{ supp}(\mathscr{D})} \mathscr{D}(\lambda x. H)\cdot \mathscr{E}_{H, N} + \sum_{  H \, \in \,\mathrm{ supp}(\mathscr{D})\, \cap \, \neu} \mathscr{D}(H)\cdot H N}\and 
\mprset{vskip=0.5ex}
\inferrule*[Right=$s5$]{M \Downarrow \mathscr{D} \\ N \Downarrow \mathscr{E}}{M \oplus N \Downarrow \frac{1}{2} \cdot \mathscr{D}+ \frac{1}{2} \cdot \mathscr{E}}
\end{mathpar}
\vspace{-.6cm}
\caption{Big-step approximation.}
\label{fig: big-step approximation}
\end{framed}
\end{figure*}
We now endow $\Lambda_\oplus$ with a big-step probabilistic operational semantics  in two stages, following  Dal Lago and Zorzi~\cite{dal2012probabilistic}. First, the rules of Figure~\ref{fig: big-step approximation} define a  big-step approximation relation $M \Downarrow \mathscr{D}$ between a  term $M$ and a head distribution $\mathscr{D}$. This relation is not a function: many different head distributions can be put in correspondence with the same term $M$, because of the rule $s1$ that allows one to \enquote{give up} while looking for a distribution of a term. The big-step semantics is then the supremum of all such finite approximations:
\begin{align}\label{eq: big-step semantics}%was defn: big-step semantics
\inter M&\triangleq \sup\{\mathscr{D}\,\vert\, M \Downarrow \mathscr{D}\}
\end{align}
Observe that this supremum is guaranteed to exist since $\lbrace \mathscr{D} \in \dist(\h) \ \vert \ M \Downarrow \mathscr{D} \rbrace$ is a directed set, as can be proved by induction on $M$.

\begin{exmp}\label{exmp: lambda xx T+F} 
Consider the term $M\triangleq\mathbf  \Delta (\mathbf T \oplus\mathbf  F)$.  One can easily check that the rules in Figure~\ref{fig: big-step approximation}  allow us to derive $M \Downarrow \mathscr{D}$ for any $\mathscr{D}$ in the following set  $\Big\{  \bot,\ \frac{1}{4}\cdot  \lambda y. \mathbf T,\ \frac{1}{4}\cdot \lambda y. \mathbf F,\  \frac{1}{2}\cdot\mathbf  I,\ \frac{1}{4}\cdot \lambda y. \mathbf T+ \frac{1}{4}\cdot \lambda y.\mathbf  F, \ \frac{1}{4}\cdot \lambda y.\mathbf  T+\frac{1}{2}\cdot\mathbf  I,\  \frac{1}{4}\cdot \lambda y.\mathbf  F +\frac{1}{2}\cdot \mathbf I,\ \frac{1}{4}\cdot \lambda y.\mathbf  T + \frac{1}{4}\cdot \lambda y.\mathbf  F+ \frac{1}{2}\cdot\mathbf  I \Big\} $. The latter head distribution is the supremum of this set and so it defines the semantics of $M$.
\end{exmp}
Example~\ref{exmp: lambda xx T+F} is about normalizing terms, which means here terms $M$ with semantics of total mass $\sum \llbracket M \rrbracket =1$ and such that there exists a unique finite derivation giving $M \Downarrow \llbracket M \rrbracket$. Standard non-converging terms gives partiality:
\begin{exmp}\label{exmp: partial} 
By inspection on the rule s4 in Figure~\ref{fig: big-step approximation}, one can check that $\mathbf \Omega \Downarrow \mathscr{D}$ only if $\mathscr{D}=\bot$, so $\llbracket\mathbf  \Omega \rrbracket=\bot$. As a consequence we also have, e.g.~$\llbracket\mathbf  \Omega \oplus\mathbf  I \rrbracket = \frac{1}{2}\cdot \mathbf I$. 
\end{exmp}
The probabilistic $\lambda$-calculus allows us also for \emph{almost sure terminating} terms, namely  terms $M$ such that $\sum \llbracket M \rrbracket=1$ but \textit{without} finite derivations  of  $M \Downarrow \llbracket M \rrbracket$:
\begin{exmp}\label{exmp: almost sure termination} Consider the derivation of $ MM\Downarrow \sum ^n_{i=1}\frac{1}{2^i}\cdot y$  depicted in Figure~\ref{fig: example quasi termination}, where $M\triangleq \lambda x. (y \oplus xx)$. Any such  finite approximation of $\inter{MM}$ gives a head distribution of the form $\sum ^n_{i=1}\frac{1}{2^i}\cdot y$, for some $n \geq 1$, but only the limit sum $\sup^n_{i=1} \sum \frac{1}{2^i}\cdot y$  is equal to $y$, thus yielding $\llbracket MM \rrbracket=y$.
\end{exmp}
%\afterpage{
%\begin{landscape}
\begin{figure*}[t]
\vspace{-.4cm}
\centering
\begin{framed}
%\scalebox{0.9}{$
\def\defaultHypSeparation{\hskip .2cm}
\def\ScoreOverhang{3pt}
\AxiomC{}
\RightLabel{$s2$}
\UnaryInfC{$y \Downarrow y$}
\AxiomC{}
\RightLabel{$s2$}
\UnaryInfC{$x \Downarrow x $}
\RightLabel{$s4$}
\UnaryInfC{$xx \Downarrow  xx$}
\RightLabel{$s5$}
\BinaryInfC{$y \oplus xx \Downarrow  \frac{1}{2} \cdot y + \frac{1}{2}\cdot  xx$}
\RightLabel{$s3$}
\UnaryInfC{$M  \Downarrow   \frac{1}{2}\cdot  \lambda x. y + \frac{1}{2} \cdot \mathbf \Delta$}
\AxiomC{}
\RightLabel{$s2$}
\UnaryInfC{$y \Downarrow y$}
\AxiomC{\vdots}
\noLine
\UnaryInfC{$M  \Downarrow   \frac{1}{2}\cdot  \lambda x. y + \frac{1}{2} \cdot \mathbf \Delta$}
\AxiomC{}
\RightLabel{$s2$}
\UnaryInfC{$y \Downarrow y$}
\AxiomC{\vdots}
\noLine
\UnaryInfC{$M  \Downarrow   \frac{1}{2}\cdot  \lambda x. y + \frac{1}{2} \cdot \mathbf \Delta$}
\AxiomC{}
\RightLabel{$s2$}
\UnaryInfC{$y \Downarrow y$}
\AxiomC{}
\RightLabel{$s1$}
\UnaryInfC{$MM \Downarrow \bot$}
\RightLabel{$s4$}
\TrinaryInfC{$MM \Downarrow \frac{1}{2}\cdot y$}
\noLine
\UnaryInfC{$\vdots$}
\RightLabel{$s4$}
\TrinaryInfC{$MM\Downarrow \sum^{n-1}_{i=1} \frac{1}{2^i}\cdot y$}
\RightLabel{$s4$}
\TrinaryInfC{$MM\Downarrow \sum^n_{i=1}\frac{1}{2^i} \cdot y $}
\DisplayProof
\vspace{-.2cm}
%$}
\caption{A derivation in the big-step semantics of $MM\Downarrow\sum^n_{i=1}\frac{1}{2^i} \cdot y$, where  $M\triangleq \lambda x. (y \oplus xx)$ and $\mathbf \Delta= \lambda x.xx$.}
\label{fig: example quasi termination}
\end{framed}
\end{figure*}
%\end{landscape}
%}
The operational semantics can be defined inductively as follows:
\begin{prop} \label{prop: the semantics is invariant under reduction} For every $M, N \in \plam$ and $H \in \h$:
\begin{enumerate}[(1)]
\item  \label{enum: invariance beta}$\llbracket (\lambda x. H)N \rrbracket = \llbracket H[N/x] \rrbracket$.
\item \label{enum: invariance abs} $\inter{\lambda x. M}= \lambda x. \inter{M}$.
\item  \label{lem: invariance beta general case}  $\inter{MN}$ is equal to the following distribution:
 \begin{equation*}
\begin{split}
& \sum_{\lambda x. H\,  \in\, \mathrm{ supp}(\inter{M})} \inter{M}(\lambda x.H)\cdot \inter{H[N/x]} \\
& + \sum_{ H \, \in\, \mathrm{ supp}(\inter{M})\,\cap\, \neu} \inter{M}(H)\cdot H N .
\end{split}
\end{equation*}
\item \label{enum: invariance sum} $\llbracket M \oplus N \rrbracket = \frac{1}{2}\llbracket M \rrbracket + \frac{1}{2} \llbracket N \rrbracket$.
\end{enumerate}
Moreover, for every $H \in\h$,  $\inter{H}=H$.
\end{prop}
Note that,  if $M$ is deterministic, i.e.~a term without the probabilistic sum $\oplus$, then either $M$ has a unique head normal form  $H$ and $\inter{M}(H)=1$, or $M$ is a diverging term and $\inter{M}=\bot$. So $\inter{\cdot}$ generalises the usual deterministic semantics.

\subsection{The Head Spine Reduction is Equivalent to the Head Reduction}

The rules  in Figure \ref{fig: big-step approximation} do not correspond to the standard head reduction of the $\lambda$-calculus, but implement a variant of it, called \emph{head spine} reduction in \cite{sestoft2002demonstrating}. Let us see the difference on a deterministic $\lambda$-term, e.g. $M\triangleq (\lambda x.(\lambda y.x)y) z$. The (small-step) head reduction first evaluates the outermost redex of $M$, getting $(\lambda y.z)y$, and then the latter term, terminating in the hnf $z$. The small-step reduction relation associated with Figure \ref{fig: big-step approximation} is given in Appendix~\ref{appequiv}, but just the inspection of the rule s$4$ may convince the reader that this reduction will first evaluate the body of $\lambda x.(\lambda y.x)y$ to an hnf, so getting the term $\lambda x.x$ and then it fires the application of the latter to the variable $z$, getting $z$. The two reduction sequences are different but they give the same result (and actually with the same number of reduction steps). We prove in Theorem \ref{thm: equivalence head e head spine nel paper} that this is always the case, even in a probabilistic setting\footnote{To the best of our knowledge, this result does not appear in the earlier literature, even in the deterministic case.}. 

We decided to consider the head spine reduction as it has a compact big-step presentation and it fits perfectly into the $\plam$-Markov chain definition  (see Remark~\ref{rk:spine_markov}). Also, it allows us for a simpler proof of the soundness property  (Theorem~\ref{thm: soundness new}). On the other side, the equivalence with the head reduction makes available the separation property (here Theorem \ref{thm: probabilistic separation}) that Leventis proved for the head reduction strategy \cite{leventis2018probabilistic} and that will play a crucial role for completeness. 

In order to state Theorem \ref{thm: equivalence head e head spine nel paper} let us define precisely the probabilistic head reduction operational semantics $\mathcal H^\infty$. Following~\cite{danos2011probabilistic, EhrPagTas11}, we define it as the limit of the small-step transition matrix $\mathcal{H}$ over $\plam$. For $M,N\in\plam$ we set:
\begin{equation*}
\mathcal{H}(M, N)\triangleq  \begin{cases}1 &\text{if }M=\cont{(\lambda y. P)Q} \text{ and } N=\cont{P[Q/y]}\\
\frac{1}{2}&\text{if } M=\cont{P_1\oplus P_2}, \, P_1 \not = P_2\, \text{and}\,  N=\cont{P_i}\\
1&\text{if } M=\cont{P\oplus P},\text{ and } N=\cont{P}\\
1&\text{if } M=N\in\h\\
0&\text{otherwise}
 \end{cases}
 \end{equation*}
 where $\mathcal{E}$ is a \textit{head context}, i.e. a special one-hole context of the form $\lambda x_1\ldots x_n.[\cdot]L_1\ldots L_m$, with $n, m \geq 0$ and $L_i \in \Lambda_\oplus$ (we slightly anticipate from Subsection \ref{sect:context}).  The matrix $\mathcal{H}$ is stochastic, i.e.~for any $M$, $\sum_{N \in \plam} \mathcal{H}(M, N)=1$. 
 
Intuitively, the entry  $\mathcal{H}^n(M, N)$ of the $n$-th power  $\mathcal{H}^n$ of the matrix $\mathcal{H}$ describes the probability that $M$ reduces to $N$ after at most $n$ steps of head reduction. Notice that the head normal forms are absorbing states of the process, so for $M \in \plam$ and $H \in \h$, the sequence $(\mathcal{H}^n(M, H))_{n \in \mathbb{N}}$ is monotone increasing and bounded by $1$, so it converges. We define its limit by: 
\begin{align}\label{eq:def_hnf}
\mathcal{H}^\infty(M, H)&\triangleq \sup_{n \in \mathbb{N}} \mathcal{H}^n(M, H)&\forall M \in \plam, \forall H \in \h.
\end{align}
 This quantity gives the total probability of $M$ to reduce to the hnf $H$ in an arbitrary number of head reduction steps.
%In order to understand the difference between the two strategies, let us consider the term $M=(\lambda x. x \oplus I) y $. By evaluating $M$ according to the head reduction strategy we obtain $y \oplus I$, that reduces  in one step to   $y$ and  $I$, both with probability $\frac{1}{2}$. Instead, according to the head     spine reduction strategy,  the function $\lambda x. x \oplus I$ is evaluated to a head normal form \emph{before} being called, so that we obtain the  terms   $(\lambda  x.x) y$ and $(\lambda  x. I) y$, both with probability $\frac{1}{2}$, which reduce in a single step to   $y$ and  $I$, respectively.  Diagrammatically, we have the following situation:
%\begin{equation*}
%\begin{tikzcd}
%& (\lambda  x. x) y\arrow[rrd,start anchor=south, end anchor=north west] &  (\lambda  x. I) y\arrow[r]& I\\
%(\lambda  x. x \oplus I) y \arrow[r, dashed] \arrow[ru, start anchor=east]\arrow[rru, start anchor=east]& y \oplus I \arrow[rr, dashed]\arrow[rru,start anchor=east, end anchor = south west, dashed]&&y
%\end{tikzcd}
%\end{equation*}  
%In the above example, $M$ reduces to the same head normal forms with the same probability, no matter which strategy is considered. Actually,  the head and the head spine reductions share the \textit{same} observational behaviour:
\begin{thm}\label{thm: equivalence head e head spine nel paper}
Let $M \in \plam$, $H \in \h$, we have:
\[
	\inter{M}(H)=\mathcal{H}^\infty(M, H).
\]
\end{thm}
Hence, our definition of $\inter{\cdot}$ is just another way of presenting  the operational semantics generated by the head reduction and discussed, for example, in~\cite{EhrPagTas11, leventis2018probabilistic, leventis2019strong}

\subsection{Contextual Equivalence}\label{sect:context}

%A standard way of comparing terms is by observing their behaviours within programming contexts. Intuitively,  two terms $M$ and $N$ are considered as  equivalent if any occurrence of $M$ in a program $L$ can be replaced with $N$ without changing the observable behaviour of $L$. The notion of context allows us to formalize this idea. 
A \textit{context} of $\Lambda_{\oplus}$ is a term containing a unique occurrence of a special variable  $[\cdot ]$, called the hole. This is generated by:
\begin{equation}\label{eqn: grammar context}
\mathcal{C}:= [\cdot] \ \vert  \   \lambda x. \mathcal{C}  \ \vert \  \mathcal{C}M  \ \vert \ M \mathcal{C} \ \vert \ \mathcal{C}\oplus M \ \vert \ M \oplus \mathcal{C} \enspace.
\end{equation}
We denote by $\clam$ the set of all contexts. Given $\mathcal{C}\in\clam$ and $M\in\plam$, then $\mathcal{C}[M]$ denotes a term obtained by substituting the unique hole in $\mathcal{C}$ with $M$ allowing the possible capture of free variables of $M$.

The typical observation in $\plam$ is the probability of converging to a value. Since values are hnfs,  \emph{contextual preorder} $ \leq_{\mathrm{cxt}}$ and \emph{contextual equivalence} $=_{\mathrm{cxt}}$ can be defined as follows:
\begin{align}
\label{eqn: contextual preorder}M \leq_{\mathrm{cxt}} N&\text{ iff } \forall\mathcal{C} \in \clam,  \sum \llbracket \mathcal{C}[M] \rrbracket \leq \sum \llbracket \mathcal{C}[N] \rrbracket\enspace,\\
\label{eqn: contextual equivalence}M =_{\mathrm{cxt}} N&\text{ iff } \forall\mathcal{C} \in \clam,   \sum \llbracket \mathcal{C}[M] \rrbracket = \sum \llbracket \mathcal{C}[N]\rrbracket\enspace.
\end{align}
%\begin{defn} \label{defn: contextual preorder and equivalence} For every $M, N \in \plam$:
%\begin{enumerate}[(1)]
%\item \textit{context preorder}: 
%\item \textit{context equivalence}:  $M =_{\mathrm{cxt}} N$ if and only if, for all $\mathcal{C} \in \clam$,   $\sum \llbracket \mathcal{C}[M] \rrbracket = \sum \llbracket \mathcal{C}[N]\rrbracket $. \label{eqn: contextual equivalence}
%\end{enumerate}
%\end{defn}
Note that $M=_{\mathrm{cxt}} N$ if and only if $M \leq_{\mathrm{cxt}} N$ and $N \leq_{\mathrm{cxt}} M$. 
\begin{exmp} \label{exmp: zx +zy and z(x+y)} Consider the terms $M \triangleq \lambda xyz. z(x \oplus y) $ and $N \triangleq  \lambda xyz. (zx \oplus zy)$. They can be discriminated by the context $\mathcal{C}\triangleq [\cdot]\mathbf \Omega\mathbf  I \mathbf   \Delta$, where $\mathbf \Omega$, $\mathbf I$, and $\mathbf \Delta$ are as in Example~\ref{exmp: examples of probabilistic terms}. In Figure~\ref{fig: example context inequivalence} we show that  $\sum \llbracket \mathcal{C}[M]\rrbracket=\frac{1}{4}$ and $\sum \llbracket \mathcal{C}[N]\rrbracket= \frac{1}{2}$.
\end{exmp}
%\afterpage{
%\begin{landscape}
\begin{figure*}[t]
\centering
\begin{framed}
%\scalebox{0.9}{
$
\begin{matrix}
\def\defaultHypSeparation{\hskip .2cm}
\def\ScoreOverhang{2pt}
\AxiomC{\vdots}
\noLine
\UnaryInfC{$M \Downarrow M$}
\AxiomC{\vdots}
\noLine
\UnaryInfC{$\lambda yz. z(\mathbf \Omega \oplus y)\Downarrow \lambda yz. z(\mathbf \Omega\oplus y)$}
\AxiomC{\vdots}
\noLine
\UnaryInfC{$\lambda z. z\, \mathsf{hid} \Downarrow \lambda z.z\, \mathsf{hid} $}
\AxiomC{\vdots}
\noLine
\UnaryInfC{$\mathbf \Delta \Downarrow\mathbf  \Delta$}
\AxiomC{}
\RightLabel{$s1$}
\UnaryInfC{$\mathbf \Omega \Downarrow \bot$}
\AxiomC{\vdots}
\noLine
\UnaryInfC{$\mathbf I \Downarrow\mathbf  I$}
\RightLabel{$s5$}
\BinaryInfC{$\mathsf{hid}  \Downarrow \frac{1}{2}\cdot\mathbf  I$}
\AxiomC{\vdots}
\noLine
\UnaryInfC{$\mathsf{hid}  \Downarrow \frac{1}{2}\cdot\mathbf  I$}
\RightLabel{$s4$}
\BinaryInfC{$\mathsf{hid} \, \mathsf{hid} \Downarrow \frac{1}{4}\cdot\mathbf  I$}
\RightLabel{$s4$}
\BinaryInfC{$\mathbf \Delta \, \mathsf{hid} \Downarrow \frac{1}{4}\cdot \mathbf I$}
\RightLabel{$s4$}
\BinaryInfC{$(\lambda z.z\, \mathsf{hid} )\mathbf  \Delta\Downarrow \frac{1}{4}\cdot\mathbf  I$}
\RightLabel{$s4$}
\BinaryInfC{$(\lambda yz. z(\mathbf \Omega \oplus y)) \mathbf I\mathbf  \Delta\Downarrow \frac{1}{4}\cdot\mathbf  I$}
\RightLabel{$s4$}
\BinaryInfC{$M\mathbf \Omega\mathbf  I\mathbf  \Delta \Downarrow \frac{1}{4}\cdot\mathbf  I$}
\DisplayProof
\\  \qquad \\ 

\def\defaultHypSeparation{\hskip .2cm}
\def\ScoreOverhang{1pt}
\AxiomC{}
\RightLabel{$s2$}
\UnaryInfC{$ z \Downarrow  z $}
\RightLabel{$s4$}
\UnaryInfC{$ zx \Downarrow  zx $}
\AxiomC{}
\RightLabel{$s2$}
\UnaryInfC{$ z \Downarrow  z $}
\RightLabel{$s4$}
\UnaryInfC{$zy   \Downarrow  zy $}
\RightLabel{$s5$}
\BinaryInfC{$zx \oplus zy \Downarrow \frac{1}{2}\cdot zx + \frac{1}{2}\cdot zy$}
\doubleLine
%\RightLabel{$s3$}
%\UnaryInfC{$\lambda z. zx \oplus zy \Downarrow \frac{1}{2}\cdot \lambda z. zx + \frac{1}{2}\cdot\lambda z. zy$}
%\RightLabel{$s3$}
%\UnaryInfC{$\lambda yz. zx \oplus zy \Downarrow \frac{1}{2}\cdot \lambda yz. zx + \frac{1}{2}\cdot\lambda yz. zy$}
\RightLabel{$s3$}
\UnaryInfC{$N \Downarrow \frac{1}{2}\cdot \lambda xyz. zx + \frac{1}{2}\cdot \lambda xyz. zy$}

%\AxiomC{\vdots}
%\noLine
%\UnaryInfC{$\lambda yz. z\mathbf \Omega \Downarrow \lambda yz. z\mathbf \Omega$}
%\AxiomC{\vdots}
%\noLine
%\UnaryInfC{$\lambda z. z\mathbf \Omega \Downarrow \lambda z.z\mathbf \Omega$}
%\AxiomC{\vdots}
%\noLine
%\UnaryInfC{$\mathbf \Delta \Downarrow \mathbf \Delta$}
%\AxiomC{}
%\RightLabel{$s1$}
%\UnaryInfC{$\mathbf \Omega\mathbf  \Omega \Downarrow \bot$}
%\RightLabel{$s4$}
%\BinaryInfC{$\mathbf \Delta\mathbf  \Omega  \Downarrow \bot$}
%\RightLabel{$s4$}
%\BinaryInfC{$(\lambda yz. z\mathbf \Omega)  \mathbf \Delta\Downarrow \bot$}
%\RightLabel{$s4$}

\AxiomC{}
\RightLabel{$s1$}
\UnaryInfC{$(\lambda yz. z \mathbf \Omega )\mathbf  I \Delta\Downarrow  \bot$}
\AxiomC{\vdots}
\noLine
\UnaryInfC{$\lambda yz. z y \Downarrow \lambda yz. zy$}
\AxiomC{\vdots}
\noLine
\UnaryInfC{$\lambda z. z\mathbf I \Downarrow \lambda z.z\mathbf I$}
\AxiomC{\vdots}
\noLine
\UnaryInfC{$\mathbf \Delta \Downarrow \mathbf \Delta$}
\AxiomC{\vdots}
\noLine
\UnaryInfC{$\mathbf{I}\Downarrow \mathbf{I}$}
\AxiomC{\vdots}
\noLine
\UnaryInfC{$\mathbf{I}\Downarrow \mathbf{I}$}
\RightLabel{$s4$}
\BinaryInfC{$\mathbf I \mathbf  I \Downarrow  \mathbf I $}
\RightLabel{$s4$}
\BinaryInfC{$\mathbf \Delta\mathbf  I  \Downarrow  \mathbf I$}
\RightLabel{$s4$}
\BinaryInfC{$(\lambda z. z \mathbf I )  \mathbf \Delta\Downarrow  \mathbf I$}
\RightLabel{$s4$}
\BinaryInfC{$(\lambda yz. zy)\mathbf I\mathbf  \Delta \Downarrow  \mathbf I$}
\RightLabel{$s4$}
\TrinaryInfC{$N\mathbf \Omega\mathbf  I \Delta \Downarrow \frac{1}{2}\cdot \mathbf I$}
\DisplayProof

%\AxiomC{}
%\RightLabel{$s2$}
%\UnaryInfC{$N \Downarrow N$}
%\AxiomC{}
%\RightLabel{$s2$}
%\UnaryInfC{$\lambda yz. (z\Omega \oplus zy)\Downarrow \lambda yz. (z\Omega\oplus zy)$}
%\AxiomC{}
%\RightLabel{$s2$}
%\UnaryInfC{$\lambda z. (z\Omega \oplus zI)\Downarrow \lambda z.(z\Omega \oplus zI)$}
%
%
%\AxiomC{}
%\RightLabel{$s2$}
%\UnaryInfC{$\Delta \Downarrow \Delta$}
%
%\AxiomC{}
%\RightLabel{s1}
%\UnaryInfC{$\Omega \Omega \Downarrow \bot$}
%\RightLabel{$s4$}
%\BinaryInfC{$\Delta \Omega \Downarrow \bot$}
%
%\AxiomC{}
%\RightLabel{$s2$}
%\UnaryInfC{$\Delta \Downarrow \Delta$}
%
%\AxiomC{$\vdots$}
%\noLine
%\UnaryInfC{$II \Rightarrow I$}
%\RightLabel{$s4$}
%\BinaryInfC{$\Delta I \Rightarrow I$}
%\RightLabel{$s5$}
%\BinaryInfC{$\Delta \Omega \oplus  \Delta I \Downarrow \frac{1}{2}\cdot I$}
%\RightLabel{$s4$}
%\BinaryInfC{$(\lambda z.(z\Omega \oplus zI)) \Delta \Downarrow \frac{1}{2}\cdot I$}
%\RightLabel{$s4$}
%\BinaryInfC{$(\lambda yz. (z\Omega \oplus zy)) I \Delta \Downarrow \frac{1}{2}\cdot I$}
%\RightLabel{$s4$}
%\BinaryInfC{$N\Omega I \Delta \Downarrow \frac{1}{2}\cdot I$}
%\DisplayProof
\end{matrix}
$
%}
\caption{The derivations in the big-step semantics of $M\mathbf \Omega\mathbf  I\mathbf  \Delta \Downarrow \frac{1}{4}\cdot\mathbf  I$ and $N\mathbf \Omega\mathbf  I\mathbf  \Delta \Downarrow \frac{1}{2}\cdot\mathbf  I$, where $M \triangleq \lambda xyz. z(x \oplus y)$,  $N \triangleq\lambda xyz. (zx \oplus zy)$, $\mathbf \Delta= \lambda x.xx$, and  $\mathsf{hid}= \mathbf \Omega \oplus\mathbf  I$. The double inference line means multiple applications of the same rule.} 
\label{fig: example context inequivalence}
\end{framed}
\end{figure*}
%\end{landscape}
%}
Contexts enjoy the following monotonicity property:
\begin{lem}\label{lem: operational semantics monotonicity contexts} Let $M, N \in \Lambda_\oplus$. If $\inter{M}\leq_{\dist}\inter{N}$ then  $\forall \mathcal{C}\in \mathsf{C}\Lambda_\oplus$ $\inter{\mathcal{C}[M]}\leq_{\dist}\inter{\mathcal{C}[N]}$.
\end{lem}
An immediate consequence of Lemma~\ref{lem: operational semantics monotonicity contexts} is the soundness of the operational semantics:
\begin{prop} \label{prop: same small-step implies context equivalent} 
Let $M, N \in \plam$: if $\llbracket M \rrbracket \leq_\dist \llbracket N \rrbracket$ (resp. $\inter{M}= \inter{N}$) then $M \leq_{\mathrm{cxt}} N$ (resp. $M =_{\mathrm{cxt}} N$).
%\begin{enumerate}[(1)]
%\item \label{enum: small-step preorder implies contextual preorder} If $\llbracket M \rrbracket \leq_\dist \llbracket N \rrbracket$ then $M \leq_{\mathrm{cxt}} N$.
%\item \label{enum: small-step equality implies contextual equivalence}If $\inter{M}= \inter{N}$ then $M =_{\mathrm{cxt}} N$. 
%\end{enumerate}
\end{prop}
Thanks to Proposition~\ref{prop: same small-step implies context equivalent}, one can prove that quite different terms are indeed contextually equivalent, as the following example shows:
\begin{exmp}  
The term $MM$ in Example~\ref{exmp: almost sure termination} and $y$ are contextually  equivalent, i.e.~$ MM =_{\mathrm{cxt}} y $, since     $\inter{ MM }=y$.

However, not all contextually equivalent terms have the same semantics: the term $\lambda x.x$ and its $\eta$-expansion $\lambda xy. xy$ are contextually equivalent but $\inter{\lambda x.x}= \lambda x.x \neq \lambda xy. xy= \inter{\lambda xy. xy} $.
\end{exmp}
Proving contextual equivalence might be rather difficult since its definition quantifies over the set of \textit{all} contexts. Fortunately, various other tools can be deployed to show the equivalence of terms. An example is bisimilarity,  we shall discuss in the next subsection. Checking that two terms are bisimilar requires the  \textit{existence} of a particular relation, called \enquote{bisimulation}. Proving that bisimilarity and contextual  equivalence actually coincide would imply that the latter can be established using the much more tractable operational techniques coming from bisimilarity. 

\subsection{Probabilistic Applicative (Bi)Similarity}\label{chap 4 sec 3 subsec 1}
We recall here the main definitions and basic properties given in \cite{dal2014coinductive}, as these do not depend on a specific operational semantics.  First, we introduce labelled Markov chains and its associated probabilistic (bi)similarity \cite{larsen1991bisimulation}. Then, we apply these notions to the operational semantics of $\plam$, getting the probabilistic applicative (bi)similarity.

\medskip
A \emph{labelled Markov chain} is a triple $\mathcal{M}= (\mathcal{S}, \mathcal{L}, \mathcal{P})$, where $\mathcal{S}$ is  a countable set of states, $\mathcal{L}$ is a set of labels (actions) and $\mathcal{P}$ is a transition probability matrix, i.e.~a function $\mathcal{P}: \mathcal{S}\times \mathcal{L}\times \mathcal{S}\longrightarrow [0, 1]$ satisfying the following condition: 
\begin{equation*}
\forall s \in \mathcal{S},\, \forall l \in \mathcal{L}: \qquad \sum_{t \in \mathcal{S}}\mathcal{P}(s, l, t )\leq 1\enspace.
\end{equation*}
If $X \subseteq \mathcal{S}$, we let  $\mathcal{P}(s, l, X)$  denote $\sum_{t \in X}\mathcal{P}(s, l, t)$.   
% The notions of probabilistic simulation and probabilistic bisimulation can be defined as follows:

A  \textit{probabilistic simulation} $\mathcal{R}$ in $\mathcal M$ is a preorder over $\mathcal{S}$ s.t.:
 \begin{equation}\label{eq:sim_gen}
    \forall (s, t)\in \mathcal{R}, \forall X \subseteq \mathcal{S}, \forall l \in \mathcal{L},\, \mathcal{P}(s, l, X)\leq \mathcal{P}(t, l, \mathcal{R}(X))
\end{equation} 
A \textit{probabilistic bisimulation} $\mathcal{R}$ is an equivalence over $\mathcal{S}$ s.t.:
\begin{equation}\label{eq:bisim_gen}
\forall (s, t) \in \mathcal{R},  \forall E \in \mathcal{S}/ \mathcal{R}, \forall l \in \mathcal{L},\, \mathcal{P}(s, l , E)= \mathcal{P}(t, l, E)
\end{equation}

%Let $(\mathcal{S}, \mathcal{L}, \mathcal{P})$ be a labelled Markov chain. 
The \emph{probabilistic similarity $\precsim$} (resp.~\emph{probabilistic bisimilarity $\sim$})  is the union of all probabilistic simulations (resp.~bisimulations). For all $ s, t\in \mathcal{S}$: 
\begin{align}
 s \preceq t &\Leftrightarrow \exists \mathcal{R}  \text{ probabilistic simulation s.t. }  s\ \mathcal{R}\ t\label{enum: similarity}, \\
 s \sim t &\Leftrightarrow \exists \mathcal{R}  \text{ probabilistic bisimulation s.t. }  s\ \mathcal{R}\ t \label{enum: bisimilarity}     .
\end{align}
\begin{prop}[e.g. \cite{dal2014coinductive}] \label{prop: properties dal lago bisimilarity}  The relation $\precsim$  (resp.~$\sim$) is a probabilistic simulation (resp.~bisimulation). Moreover,   it holds that ${\sim} =  {\precsim} \cap {\precsim}^{op}$.
\end{prop}

In order to apply these notions to $\plam$, we need to preset its operational semantics as a labelled Markov chain (Definition~\ref{def:plam_markov}). Intuitively,  terms are seen as states, while labels are of two kinds: one can either \textit{evaluate} a term (this kind of transition will be labelled by $\tau$), obtaining a distribution of hnfs, or \textit{apply} a hnf to a term $M$ (this kind of transition will be labelled by $M$).  
%This idea has been first developed in the standard (deterministic) $\lambda$-calculus by Abramsky~\cite{abramsky1990lazy}, who called the  corresponding notion of bisimilarity \enquote{applicative}. Applicative bisimilarity has been then studied in the probabilistic $\lambda$-calculus for several reduction strategies like, for example,  (lazy) call-by-name (Dal Lago et al.~\cite{dal2014coinductive}) and  call-by-value (Crubillé\&Dal Lago~\cite{crubille2014probabilistic}).  The benefit of this approach is  to check program equivalence via an \textit{existential} quantifier (as in~\eqref{enum: bisimilarity}) rather than a \textit{universal} one, as in the case of   context equivalence  (see Definition~\ref{defn: contextual preorder and equivalence}.\ref{eqn: contextual equivalence}). 
For technical reasons, it is useful to consider only closed terms and to consider for each closed hnf $H=\lambda x. H'$ two distinct representations, depending on the way we consider it: either as a term  or properly as a normal form, and in the latter case we indicate it as  $\ww{H}\triangleq \nu x. H' $ to stress the difference. Consequently, we define $\dhe$ as the set of all  \enquote{distinguished} closed hnfs,  namely $\lbrace \ww{H} \ \vert \ H \in \he  \rbrace$.  More in general, if $X \subseteq \he$, we define $\ww{X}\triangleq \lbrace \ww{H} \ \vert \ H \in X\rbrace$.

\begin{defn}\label{def:plam_markov}
The  \textit{$\plam$-Markov chain} is the triple $(\elam \uplus \dhe,\, \elam \uplus \lbrace \tau \rbrace,\, \mathcal{P}_\oplus )$, where the set of states is the disjoint union of the set of closed terms and  the set of \enquote{distinguished} closed hnfs, labels (actions) are either closed terms or the $\tau$ action, and the transition probability matrix $\mathcal{P}_\oplus$ is defined in the following way:
 \begin{enumerate}[(i)]
 \item \label{enum: plam markov 1}for every closed term $M$ and distinguished hnf $\nu x. H$:
 \begin{equation*}
 \mathcal{P}_\oplus(M, \tau , \nu x. H)\triangleq  \inter{M}(\lambda x. H) \enspace,
 \end{equation*}
 \item \label{enum: plam markov 2} for every closed term $M$ and distinguished hnf $\nu x. H$:
 \begin{equation*}
 \mathcal{P}_\oplus(\nu x. H, M , H[M/x])\triangleq1 \enspace,
 \end{equation*}
 \item in all other cases, $\mathcal{P}_\oplus$ returns $0$.
 \end{enumerate}
 
 A \textit{probabilistic applicative (bi)simulation} is a probabilistic (bi)simulation of the $\plam$-Markov chain.  The \textit{probabilistic applicative similarity}, $\mathrm{PAS}$ for short, and the \textit{probabilistic applicative bisimilarity}, $\textrm{PAB}$ for short, are defined as in~\eqref{enum: similarity} and~\eqref{enum: bisimilarity}. From now on, with  $\precsim$ (resp. $\sim$) we mean probabilistic \textit{applicative} similarity (resp.~bisimilarity).
 \end{defn}
 
\begin{rem}\label{rk:spine_markov}
%The definition of $\plam$-Markov chain makes reference to the    operational semantics based on the head spine reduction (Definition~\ref{defn: big-step semantics}).   By Theorem~\ref{thm: equivalence head e head spine nel paper}, switching to an operational semantics implementing the head spine reduction (like the one in Remark~\ref{rem: head vs head spine}) would not affect the above definition. Nevertheless, the $\plam$-Markov chain  suggests an evaluation strategy that is  much closer to the head \textit{spine} reduction. 
In the  $\plam$-Markov chain, a term $M$ can be  thought at the head of a (potentially infinite) stack of applications, where at each time we first evaluate  the head of the stack until we reach a head normal form $H$  (point~\ref{enum: plam markov 1}),  and then we apply $H$ to the next term of the stack (point~\ref{enum: plam markov 2}). This is exactly the behaviour of  the head spine reduction on an application $MN_1\ldots N_n$.
 Lemma~\ref{lem: context lemma 4} formalizes these intuitions.
%The definition of $\plam$-Markov chain 
% makes reference to the    operational semantics based on the head spine reduction (Definition~\ref{defn: big-step semantics}).   By Theorem~\ref{thm: equivalence head e head spine nel paper}, switching to an operational semantics implementing the head spine reduction (like the one in Remark~\ref{rem: head vs head spine}) would not affect the above definition. Nevertheless, the $\plam$-Markov chain 
%suggests an evaluation strategy that is  much closer to the head \textit{spine} reduction. In the  $\plam$-Markov chain, a term $M$ can be  thought at the head of a (potentially infinite) stack of applications, where at each time we first evaluate  the head of the stack until we reach a head normal form $H$  (point~\ref{enum: plam markov 1}),  and then we apply $H$ to the next term of the stack (point~\ref{enum: plam markov 2}). This is exactly the behaviour of  the head spine reduction on an application $MN_1\ldots N_n$.
% Lemma~\ref{lem: context lemma 4} formalizes these intuitions.
\end{rem}
%Since $\plam$ can be seen as a labelled Markov chain, simulation and bisimulation can be defined as well.  

The notions of $\textrm{PAS}$ and $\textrm{PAB}$ are defined on closed terms. We extend them to open terms $M, N \in \Lambda_\oplus^{\lbrace x_1, \ldots, x_n \rbrace}$, by:
\begin{align}
M \precsim N &\Leftrightarrow  \lambda x_1 \ldots x_n.  M \precsim \lambda x_1\ldots x_n. N\label{eqn: open term simil}\enspace , \\
 M \sim N &\Leftrightarrow \lambda x_1\ldots x_n.  M \sim \lambda x_1\ldots x_n. N \label{eqn: open term bisimil}\enspace .
\end{align}
One can notice that the order of the abstractions in the term closure does not affect the obtained relation. 

The following proposition is  analogous to Proposition~\ref{prop: same small-step implies context equivalent}, stating the soundness of  the operational semantics with respect to both  $\mathrm{PAS}$ and $\mathrm{PAB}$.
\begin{prop}\label{prop: operational semantics is sound w.r.t. bisimilarity} Let $M, N \in \plam$: if $\inter{M}\leq_\dist \inter{N}$ (resp.~$\inter{M}= \inter{N}$) then $M \precsim N$ (resp.~$M \sim N$).
%\item \label{enum: operational equality implies bisimilarity} If $\inter{M}= \inter{N}$ then $M \sim N$.
%\begin{enumerate}[(1)]
%\item \label{enum: operational preorder implies similarity} If $\inter{M}\leq_\dist \inter{N}$ then $M \precsim N$.
%\item \label{enum: operational equality implies bisimilarity} If $\inter{M}= \inter{N}$ then $M \sim N$.
%\end{enumerate}
\end{prop}
\begin{proof}
We prove only the inequality soundness, as the equality one is an immediate consequence by Proposition~\ref{prop: properties dal lago bisimilarity}. Moreover, the proof is for closed terms, as the case of open terms follows from Proposition~\ref{prop: the semantics is invariant under reduction}.\ref{enum: invariance abs}.

Let $M, N \in \elam$ be such that $\inter{M}\leq_{\dist} \inter{N}$, and consider the relation $\mathcal{R}= \lbrace(P, Q) \in \elam \times \elam  \ \vert \ \inter{P}\leq_\dist \inter{Q} \rbrace \cup \lbrace (\nu x.H, \nu x.H) \in \dhe \times \dhe \rbrace$.  If we show that $\mathcal{R}$ is a PAS, then $\mathcal{R}\subseteq \, \precsim$, and hence $M \precsim N$.  Clearly, $\mathcal{R}$ is a preorder. Now, let $(P,Q), (\nu x.H,\nu x.H) \in \mathcal{R}$, and let $X \subseteq \elam \cup  \dhe$.  It is straightforward that $\mathcal{P}_\oplus(\nu x.H, l, X)\leq \mathcal{P}_\oplus(\nu x.H, l, \mathcal{R}(X))$, for all $l \in \elam \cup \lbrace \tau \rbrace$.  Moreover, for all $F \in \elam$ we have $0=\mathcal{P}_\oplus(P, F, X)\leq \mathcal{P}_\oplus(Q, F, \mathcal{R}(X))$. Last:
\allowdisplaybreaks
\begin{align*}
 \mathcal{P}_\oplus(P, \tau, X)&= \sum_{\nu x.H \in X} \mathcal{P}_\oplus(P, \tau , \nu x.H)=\inter{P}(X\cap\h)
 % \sum_{\nu x. H \in X} \inter{P}(\lambda x. H)
\\
 &\leq\inter{Q}(X\cap\h)= \mathcal{P}_\oplus(Q, \tau , \mathcal{R}(X)).
 % \leq \sum_{\nu x. H \in X} \inter{Q}(\lambda x.H)=\sum_{\nu x. H \in X} \mathcal{P}_\oplus(Q, \tau ,\nu x. H)\\
% &=\sum_{\nu x.H \in X} \mathcal{P}_\oplus(Q, \tau , \mathcal{R}(\nu x. H))= \mathcal{P}_\oplus(Q, \tau , \mathcal{R}(X))\ .
\end{align*}
Hence, for all $l \in \elam \cup \lbrace \tau \rbrace$ and $X \subseteq \elam \cup \dhe$,  we have $\mathcal{P}_\oplus(P, l, X)\leq \mathcal{P}_\oplus(Q, l, \mathcal{R}(X))$.
%Now, let $M, N \in \plam^{\lbrace x_1, \ldots, x_n\rbrace}$ be such that  $\inter{M}\leq_{\dist} \inter{N}$. This means that $\lambda x_1\ldots x_n.\inter{M}\leq_{\dist} \lambda x_1\ldots x_n.  \inter{N}$, and hence $\inter{\lambda x_1\ldots x_n.M}\leq_{\dist}  \inter{\lambda x_1\ldots x_n. N}$ by  Proposition~\ref{prop: the semantics is invariant under reduction}.\ref{enum: invariance abs}. Since these terms are closed, we have $\lambda x_1\ldots x_n.M \precsim \lambda x_1\ldots x_n.N$. From~\eqref{eqn: open term simil}, we conclude  $M \precsim N$.
\end{proof}

\begin{exmp} 
Let us show that $\mathbf I \sim \lambda xy.xy$ so that, from the soundness (Theorem~\ref{thm: soundness new}), one can infer $\mathbf I=_{\mathrm{cxt}} \lambda xy.xy$. 
%Checking that  $\mathbf I=_{\mathrm{cxt}} \lambda xy.xy$ is not immediate. However, one can easily show that $\mathbf I \sim \lambda xy.xy$ so that, from the soundness (Theorem~\ref{thm: soundness new}), one can infer   contextual equivalence. 

Let us  define $\mathcal{R}_1\triangleq\big\{ (\mathbf I,\lambda xy.xy), ( \lambda xy.xy, \mathbf I) \big\}$, as well as  $\mathcal{R}_2\triangleq 
  \big\{ ( \ww{\mathbf I}, \nu x.\lambda y.xy), (\nu x.\lambda y.xy,\ww{\mathbf I}) \big\}
 $ and $\mathcal{R}_3 \triangleq {\sim}$. Let $\mathcal{R}\triangleq {(\mathcal{R}_1\cup \mathcal{R}_2 \cup \mathcal{R}_3)^*}$. Since $\mathcal{R}_1\cup \mathcal{R}_2 \cup \mathcal{R}_3$ is a symmetric relation, then its reflexive and transitive closure $\mathcal{R}\triangleq {(\mathcal{R}_1\cup \mathcal{R}_2 \cup \mathcal{R}_3)^*}$ is an equivalence. Let us prove that it is a probabilistic bisimulation.  
%Since $\sim$ is an equivalence relation by Proposition~\ref{prop: properties dal lago bisimilarity}, then so is its reflexive and transitive closure $\mathcal{R}\triangleq {(\mathcal{R}_1\cup \mathcal{R}_2 \cup \mathcal{R}_3)^*}$. 

We have to prove that $\mathcal{P}_\oplus(M, l, E)= \mathcal{P}_\oplus(N, l, E)$,  $\forall(M, N ) \in \mathcal{R}$, $\forall E \in (\elam \cup \dhe)/ \mathcal{R}$, $\forall l \in \elam \cup \lbrace \tau \rbrace$. Notice that, if this holds for $(M, N)\in {(\mathcal{R}_1\cup \mathcal{R}_2 \cup \mathcal{R}_3)}$, then we are done. Indeed,  suppose $(M, N)\in \mathcal{R}$. Then there exists $n \geq 0$ and $P_0, \ldots, P_n \in \elam \cup \dhe$ such that $P_0=M$, $P_n=N$ and $P_{i-1}\mathcal{R}_{j_i}P_{i}$ for every $1 \leq i \leq n$, where $1 \leq j_i \leq 3$. Hence,   we have $\mathcal{P}_\oplus(M, l, E)=\mathcal{P}_\oplus(P_0, l, E)= \ldots = \mathcal{P}_\oplus(P_n, l, E)=\mathcal{P}_\oplus(N, l, E)$, $\forall E \in (\elam \cup \dhe)/ \mathcal{R}$, $\forall l \in \elam \cup \lbrace \tau \rbrace$.

Let us now show the case $(M, N)\in {(\mathcal{R}_1\cup \mathcal{R}_2 \cup \mathcal{R}_3)}$. If $(M, N)\in \mathcal{R}_3$ we just apply Proposition~\ref{prop: properties dal lago bisimilarity}. Otherwise,   it suffices to consider  $(\mathbf{I}, \lambda xy.xy)$ and $(\ww{\mathbf I}, \nu x.\lambda y.xy)$.  Recall that, by Definition~\ref{def:plam_markov},   $\mathcal{P}_\oplus(M, N , E)= 0$  and $\mathcal{P}_\oplus(\ww{H}, \tau , E)= 0$,   for all $M, N \in \elam$, $\ww{H}\in \dhe$ and $E \in (\elam \cup \dhe)/ \mathcal{R}$.    On the one hand, since  $ (\ww{\mathbf I}, \nu x.\lambda y.xy)\in \mathcal{R}$, we have $\ww{\mathbf I} \in E$ if and only if $\nu x.\lambda y. xy \in E$, for all $E \in (\elam \cup \dhe)/ \mathcal{R}$. This implies   $\mathcal{P}_\oplus(\mathbf I, \tau, E) = \mathcal{P}_\oplus(\lambda xy.xy, \tau, E)$, for all $E \in (\elam \cup \dhe)/ \mathcal{R}$.  On the other hand, since  terms are considered modulo renaming of bound variables, by Proposition~\ref{prop: the semantics is invariant under reduction} we have $\inter{N}= \inter{\lambda y.Ny}$, for all $N \in \elam$ (notice that this equality may fail if $N$ has free variables).  By Proposition~\ref{prop: operational semantics is sound w.r.t. bisimilarity}, $N \sim \lambda y. N y$, and hence $N \in E$ if and only if $\lambda y. N y\in E$, for all $E \in (\elam \cup \dhe)/ \mathcal{R}$.  This implies   $\mathcal{P}_\oplus(\ww{\mathbf I}, N, E)=\mathcal{P}_\oplus(\nu x.\lambda y. xy, N, E)$, for all $N \in  \elam$ and for all $E \in (\elam \cup \dhe)/ \mathcal{R}$. 
\end{exmp}
\begin{exmp} We show that the terms $M\triangleq \lambda xyz. z(x\oplus y)$ and $N \triangleq \lambda xyz. (zx\oplus zy)$ in Example~\ref{exmp: zx +zy and z(x+y)} are not bisimilar. Indeed, suppose for the sake of contradiction that a probabilistic bisimulation $\mathcal{R}$ such that $(M, N ) \in \mathcal{R}$ exists. By definition $\mathcal{R}$ is an equivalence relation.  Let $E \in (\elam \cup \dhe)/ \mathcal{R}$ be such that $ \nu x. \lambda yz. z(x\oplus y) \in E$. Then it must be that $\mathcal{P}_\oplus(M, \tau, E)=1= \mathcal{P}_\oplus(N, \tau, E)$, and it follows that both   $\nu x. \lambda yz. zx$ and $ \nu x. \lambda yz. zy$ are in $E$, so that $(\nu x.\lambda yz. z(x\oplus y), \nu x. \lambda yz.zx)\in \mathcal{R}$.  Then it must be that $\mathcal{P}_\oplus(\nu x.\lambda yz. z(x\oplus y),\mathbf  \Omega, E_1 )=1= \mathcal{P}_\oplus(\nu x. \lambda yz.zx,\mathbf  \Omega, E_1)$, for some $E_1\in (\elam \cup \dhe)/ \mathcal{R}$ containing both $\lambda yz. z(\mathbf \Omega \oplus y)$ and $ \lambda yz. z\mathbf \Omega \in E_1$, which implies $(\lambda yz. z(\mathbf \Omega \oplus y), \lambda yz. z\mathbf \Omega)\in \mathcal{R}$. By a similar reasoning, we get  that $\mathcal{R}$ contains the pairs $(\nu y.\lambda z. z(\mathbf \Omega \oplus y),\nu y. \lambda z. z\mathbf \Omega)$, $(\lambda z.z(\mathbf \Omega \oplus\mathbf  I), \lambda z. z\mathbf \Omega )$,  and $(\nu z.z(\mathbf \Omega \oplus\mathbf  I), \nu z. z\mathbf \Omega )$. Now, let $E_2$ be an equivalence class containing $\mathbf \Omega \oplus\mathbf  I$. From $ \mathcal{P}_\oplus(\nu z. z(\mathbf \Omega \oplus\mathbf  I), \mathbf I, E_2)=1= \mathcal{P}_\oplus(\nu z.z\mathbf \Omega,\mathbf  I, E_2)$ we get that $\mathbf \Omega \in  E_2$, i.e.~$(\mathbf \Omega \oplus\mathbf  I, \mathbf\Omega)\in\mathcal{R}$. Finally, if $E_3$ is an equivalence class such that $\nu x.x \in E_3$, then $\mathcal{P}_\oplus(\mathbf \Omega \oplus\mathbf  I, \tau, E_3)=\frac{1}{2}= \mathcal{P}_\oplus(\mathbf \Omega, \tau, E_3)$. This is a contradiction, since $\mathcal{P}_\oplus(\mathbf \Omega, \tau, E_3)=0$. Therefore, the terms $M$ and $N$ are not bisimilar.
\end{exmp}

\section{Soundness}\label{sec3}
A fundamental technique to establish the soundness of applicative (bi)similarity is based on \textit{Howe's lifting} \cite{howe1996proving}. This  method shows that applicative bisimilarity is a \textit{congruence}, i.e.~an equivalence relation that respects the structure of terms, which is the hard part  in the soundness proof.  This technique has been used in e.g.~\cite{dal2014coinductive,crubille2014probabilistic} for, respectively, the lazy cbn and cbv semantics of $\plam$. We consider here a different approach.  Following the reasoning by Abramsky and Ong~\cite{abramsky1993full},  we shall first prove that $\precsim$  is included in $ \leq_{\mathrm{app}}$ (Lemma~\ref{lem: context lemma 4}), which requires a technical Key Lemma (Lemma~\ref{lem: context lemmna 3}) specific to the probabilistic framework and then we conclude by applying a Context Lemma (Lemma~\ref{lem: context lemma}). The latter result says that the computational behaviour of the contextual semantics is \textit{functional}. This property has also been called \textit{operational extensionality} in Bloom~\cite{bloom1990can}. Milner~\cite{milner1977fully} proved a similar result in the case of simply typed combinatory algebra. To the best of our knowledge, the Context Lemma lacks a corresponding formulation in the probabilistic $\lambda$-calculus $\plam$, so we prove it in the following subsection. %This allows us to infer that  $\mathrm{PAS}$ is sound with respect to context preorder.  The soundness of $\mathrm{PAB}$ for  context equivalence will be a straightforward consequence of this result.

\subsection{Context Lemma}\label{sect:context_lemma}
%Context equivalence captures the intuitive idea that two programs are indistinguishable in all possible programming contexts. As already remarked, though context preorder and context equivalence are clearly important, it is hard to reason about them \textit{directly}.  Fortunately,   
The Context Lemma states that only the subset of applicative contexts \enquote{really matter} in establishing contextual equivalence. We define  an \textit{applicative context}  as a context $\mathcal{E}\in \clam$ of the form $(\lambda x_1\ldots x_n.[\cdot ])P_1\ldots P_m$, where $n, m \in \mathbb{N}$ and $P_1\ldots P_m \in \elam$.  We denote by $\alam$ the set of all applicative contexts. 

The \emph{applicative contextual preorder} $\leq_{\mathrm{app}}$ (resp.~\emph{applicative contextual equivalence} $=_{\mathrm{app}}$) is defined by restricting the quantifier $\forall\mathcal{C}$ to the subset $\alam$ of $\clam$ in the contextual  preorder (resp.~equivalence) definition \eqref{eqn: contextual preorder}  (resp.~\eqref{eqn: contextual equivalence}).
%For every $M, N \in \plam$:
%\begin{enumerate}[(1)]
%\item \textit{Applicative context preorder}: $M \leq_{\mathrm{app}} N$ if and only if, for all  $\mathcal{C} \in \alam$, $\sum \llbracket \mathcal{C}[M] \rrbracket \leq \sum \llbracket \mathcal{C}[N] \rrbracket$.   \label{eqn: capplicative ontext preorder}
%\item \textit{Applicative context equivalence}:  $M =_{\mathrm{app}} N$ if and only if, for all $\mathcal{C} \in \alam$,  $\sum \llbracket \mathcal{C}[M] \rrbracket = \sum \llbracket \mathcal{C}[N]\rrbracket $. \label{eqn: applicative context equivalence}
%\end{enumerate}
%Notice that $M=_{\mathrm{app}} N$ if and only if $M \leq_{\mathrm{app}} N$ and $N \leq_{\mathrm{app}} M$, for all $M, N \in \Lambda_\oplus$. 
\begin{lem} \label{lem: abstraction congruence for obs and app} Let $M, N \in \plam^{\Gamma \cup \lbrace x \rbrace}$. Then:
\begin{enumerate}[(1)]
\item \label{eqn: abstraction congruence for obs} If $M \leq _{\mathrm{app}}N$ then  $\lambda x. M \leq_{\mathrm{app}}\lambda x. N$.
\item  \label{eqn: abstraction congruence for app} If $\lambda x.M \leq _{\mathrm{cxt}}\lambda x.N$ then $M \leq_{\mathrm{cxt}}N$.
\item \label{eqn: application congruence for obs} If $M \leq _{\mathrm{cxt}}N$ then, for all $L \in \plam$,  $ML\leq _{\mathrm{cxt}}NL$. 
\end{enumerate}
\end{lem}
In order to prove the Context Lemma more easily, we shall adopt a slightly more general notion of context, allowing multiple holes.  A \textit{generalized context} of $\Lambda_{\oplus}$ is a term containing holes $[\cdot ]$, generated by the following grammar:
\begin{equation}\label{eqn: generalized grammar context}
\mathcal{C}:= x  \ \vert  \ [\cdot] \ \vert \   \lambda x. \mathcal{C}  \ \vert \  \mathcal{C}\mathcal{C}  \ \vert \ \mathcal{C}\oplus \mathcal{C}  \enspace .
\end{equation}
We denote by $\genlam$ the set of all generalized contexts. If $\mathcal{C}\in\genlam$ and $M\in\plam$, then $\mathcal{C}[M]$ denotes the term obtained by substituting every hole in $\mathcal{C}$ with $M$ allowing the possible capture of free variables of $M$.%  We denote by $\genlam$ the set of all contexts generated by the grammar in~\eqref{eqn: generalized grammar context}. %Notice that the observable behaviour of a generalised context \mathcal{C}$ can be simulated by a one-hole context: let $x_1,\dots,x_n$ be the gen

%Context lemma says that if two programs are indistinguishable by some context then there is some \textit{applicative context} that distinguish them. We establish this result following a standard approach (see for example~\cite{abramsky1993full}).
\begin{lem}\label{lem: fundamental step toward context lemma} Let $M ,N \in \elam$ be such that  $M \leq_{\mathrm{app}} N$. Then  $\sum \inter{\mathcal{C}[M]}\leq \sum \inter{\mathcal{C}[N]}$, for all $ \mathcal{C} \in \genlam$. 
\end{lem}
\begin{proof}[Proof (sketch)] 
By  Theorem~\ref{thm: equivalence head e head spine nel paper}  it is enough to show that, for all $n \in \mathbb{N}$ and for all generalized contexts $\mathcal{C}\in \genlam$:
\begin{equation}\label{eqn: equantion to prove for context lemma}
 \sum_{H \in \h} \mathcal{H}^n (\mathcal{C}[M], H)\leq   \sum_{H \in \h}\mathcal{H}^\infty(\mathcal{C}[N], H) \enspace .
\end{equation}
The proof is by induction on  $(n, \vert \mathcal{C}\vert )$, where $\vert \mathcal{C}\vert$ is the size  of $\mathcal{C}\in \genlam$, i.e.~the number of nodes in the syntax tree of $\mathcal{C}$. Since  $\mathcal{C}$ must be of the form  $\mathcal{C}_0\mathcal{C}_1\ldots \mathcal{C}_k$, for some $k \in \mathbb{N}$, we  proceed by case analysis, looking at the structure of $\mathcal{C}_0$.
\end{proof}
\begin{lem}[Context Lemma]\label{lem: context lemma} Let $M, N \in \plam$. Then:
\begin{enumerate}[(1)]
\item \label{eqn: context lemma leq} $M \leq_{\mathrm{cxt}} N$ if and only if $M \leq_{\mathrm{app}} N$.
\item \label{eqn: context lemma equal} $M =_{\mathrm{cxt}} N$ if and only if $M =_{\mathrm{app}} N$.
\end{enumerate}
\end{lem}
\begin{proof}
Point~\ref{eqn: context lemma equal}  follows directly from point~\ref{eqn: context lemma leq}. Lemma~\ref{lem: fundamental step toward context lemma} gives us point~\ref{eqn: context lemma leq} for $M, N \in \elam$. We extend it to open terms by applying Lemma~\ref{lem: abstraction congruence for obs and app}.\ref{eqn: abstraction congruence for obs} and Lemma~\ref{lem: abstraction congruence for obs and app}.\ref{eqn: abstraction congruence for app}. 
\end{proof}	

\subsection{The Soundness Theorem} 
We start with some preliminary lemmas.
  \begin{lem} \label{lem: context lemma 2} Let $H, H'\in \h^{\lbrace x \rbrace}$. Then, the following are equivalent statements:
 \begin{enumerate}[(1)]
 \item \label{enum: context lemma 2 1}$ \lambda x. H \precsim \lambda x. H' ,$
 \item \label{enum: context lemma 2 2}$\nu x. H \precsim \nu x. H',$
 \item \label{enum: context lemma 2 3}$\forall P \in \elam, \ H[P/x]\precsim H'[P/x]\enspace .$
 \end{enumerate}
 \end{lem}
 \begin{proof}[Proof (sketch)]  The implication \ref{enum: context lemma 2 1} $\Rightarrow$ \ref{enum: context lemma 2 2} $\Rightarrow$ \ref{enum: context lemma 2 3} is  by definition and by Proposition~\ref{prop: properties dal lago bisimilarity}. To prove \ref{enum: context lemma 2 3}  $\Rightarrow$ \ref{enum: context lemma 2 2}, it suffices to show that  the  relation $\mathcal{R}\triangleq \lbrace (\nu x. H, \nu x. H' )\in  \dhe \times \dhe \ \vert \ \forall P \in \elam, \, H[P/x]\precsim H'[P/x] \rbrace  \cup  {\precsim} $ is a probabilistic applicative simulation. Similarly, \ref{enum: context lemma 2 2} $\Rightarrow$  \ref{enum: context lemma 2 1} holds by showing that the   relation
 $ \mathcal{R}\triangleq \lbrace (\lambda  x. H, \lambda x. H' )\in \h \times 	\h \ \vert \ \nu x.H \precsim \nu x. H' \rbrace  \cup  {\precsim}$ is a probabilistic applicative simulation. 
  \end{proof}
 
Let us recall that, given $X\subseteq \h$, ${\precsim} (X)$ denotes the image of $X$ under $\precsim$. Moreover,  given $X\subseteq \h^{\lbrace x \rbrace}$, $\nu x. {\precsim} (X)$ denotes the set of distinguished hnfs $\lbrace \nu x. H \ \vert \ H \in {\precsim} (X) \rbrace$, while $\lambda  x. {\precsim} (X)$ denotes the set of terms $\lbrace \lambda x. M \ \vert \   M \in  {\precsim} (X) \rbrace$.

%Given $X\subseteq \h^{\lbrace x \rbrace}$, we denote by $\nu x. \precsim (X)$ the set of distinguished hnfs $\lbrace \nu x. H \ \vert \ \exists  H' \in X, \, H' \precsim H \rbrace$, while $\lambda  x. \precsim (X)$ denotes the set of terms $\lbrace \lambda x. M \ \vert \ \exists  N \in X, \, N \precsim M \rbrace$.
 \begin{lem}\label{lem: commutation abstraction precsim}  
 Let $X \subseteq \h^{\lbrace x \rbrace}$. We have:
 \begin{align*}
{\precsim}(\lambda x.X)\cap\he&=\lambda x. {\precsim} (X)\cap\he,\\
{\precsim}(\nu x.X)&= \nu x. {\precsim} (X)\enspace.
 \end{align*}
%  \begin{enumerate}[(1)]
% \item \label{enum: lambda precsim} For all $H \in \he$, $H \in \,\precsim(\lambda x.X)$ if and only if $H \in \lambda x. \precsim (X)$.
%  \item \label{enum: nu precsim} $\precsim(\nu x.X)= \nu x. \precsim (X)$.
% \end{enumerate}
 \end{lem} 
 \begin{lem}\label{lem: context lemmna 1}
 Let $M, N \in \elam$. For all $X \subseteq \he$, $\inter{M}(X) \leq \inter{N}({\precsim} (X))$ if and only if  $M \precsim N$.
 \end{lem}
The forthcoming Lemma~\ref{lem: context lemmna 3} describes the applicative behaviour of $\precsim$ and it  requires an auxiliary result about the so-called \enquote{probabilistic assignments}. Probabilistic   assignments were first  introduced in this setting by~\cite{dal2014coinductive} to prove the soundness of PAS in the lazy cbn. 
\begin{defn}[Probabilistic assignments] \label{defn: probabilistic assignment} A \textit{probabilistic assignment} is defined as a pair $(\lbrace p_i \rbrace_{1 \leq i \leq n}, \lbrace r_I \rbrace_{I \subseteq \lbrace 1, \ldots, n \rbrace})$, with all $p_i$, $r_I$ in $[0,1]$, such that, for all $I \subseteq \lbrace 1, \ldots, n \rbrace$:
\begin{equation}\label{eqn: condition probabilistic assignment}
\sum_{i \in I}p_i \leq \sum_{\substack{J\subseteq \lbrace 1, \ldots, n \rbrace  \\ \text{s.t. }J \cap I \not = \emptyset}} r_J\enspace .  
\end{equation}
\end{defn} 
\begin{lem}[\cite{dal2014coinductive}] \label{lem:  probabilistic assignment entanglement} Let $(\lbrace p_i \rbrace_{1 \leq i \leq n}, \lbrace r_I \rbrace_{I \subseteq \lbrace 1, \ldots, n \rbrace})$  be a probabilistic assignment. Then for every $I \subseteq \lbrace 1, \ldots, n \rbrace  $ and for every $k \in I$ there is $s_{k, I}\in[0, 1]$ such that:
\begin{enumerate}[(1)]
\item \label{enum: condition 2 entanglement} $\forall j \in \lbrace 1, \ldots, n \rbrace$,  $p_j \leq \underset{\substack{J\subseteq \lbrace 1, \ldots, n \rbrace   \\ \text{s.t. }j \in J}}{\sum}s_{j, J}\cdot r_{J}$.
\item \label{enum: condition 1 entanglement}  $\forall J\subseteq \lbrace 1, \ldots, n \rbrace $, $\underset{\substack{j \in \lbrace 1, \ldots, n \rbrace \\\text{s.t. }j \in J}}{\sum}s_{j,J}\leq 1$.
\end{enumerate}
\end{lem}
Following essentially the same ideas of~\cite{dal2014coinductive}, we shall use the above property to decompose and recombine distributions in the proof of the following lemma.
\begin{lem}[Key Lemma]\label{lem: context lemmna 3} Let $M, N \in \elam$. If $M \precsim N$ then, for all $P \in \elam$, $MP \precsim NP$.
\end{lem}
\begin{proof}[Proof (sketch)]
By Lemma~\ref{lem: context lemmna 1} it suffices to  prove that, for all $X \subseteq \he$  and for all  $\mathscr{D}\in \dist{(\h)}$ such that $MP \Downarrow \mathscr{D}$,  it holds that $\mathscr{D}(X) \leq \inter{NP}({\precsim}( X))$. The non-trivial case is when  the last rule of $MP \Downarrow \mathscr{D}$ is $s4$, i.e.~when:
\begin{align}\label{eq:d}
\mathscr{D}(X)&=\sum_{\lambda x. H\,  \in\, \mathrm{ supp}(\mathscr{E})} \mathscr{E}(\lambda x.H)\cdot \mathscr{F}_{H, P}(X)
\end{align}
for $M\Downarrow\mathscr{E}$ and $H[P/x]\Downarrow\mathscr{F}_{H, P}$. Notice that $\mathrm{supp}(\mathscr{E})$ is finite, say $\mathrm{supp}(\mathscr{E})=\lbrace \lambda z. H_1, \ldots, \lambda z. H_n\rbrace$.  

Proposition~\ref{prop: the semantics is invariant under reduction} gives us:
\begin{align}\label{eq:NP}
\inter{NP}({\precsim} (X))=\!\!\sum_{\lambda x. H} \inter{N}(\lambda x. H)\cdot \inter{H[P/x]}({\precsim} (X))
\end{align}
One would be then tempted to compare the sums \eqref{eq:d} and \eqref{eq:NP} term by term. In fact, by hypothesis we know that for every $\lambda x. H$, $\mathscr{E}(\lambda x. H)\leq \inter{N}({\precsim} \{\lambda x. H\})$. This gives that every term $\mathscr{E}(\lambda x.H)\cdot \mathscr{F}_{H, P}({\precsim}( X))$ of \eqref{eq:d} is smaller than $\sum_{\lambda x.H'\in{\precsim}(\lambda x.H)}\inter N(\lambda x.H')\cdot \inter{H'[P/x]}({\precsim}( X))$. Unfortunately we cannot conclude, as different hnfs $\lambda x.H$ do not always generate disjoint ${\precsim}(\lambda x.H)$ (e.g. think about $\eta$-equivalent hnfs), so that we cannot factor \eqref{eq:NP} according to ${\precsim} (\lambda x.H_1)$,\dots, ${\precsim}(\lambda x.H_n)$. Here is where Lemma~\ref{lem:  probabilistic assignment entanglement} on probabilistic assignments plays a role, permitting to ``disentangle'' the different quantities $\inter N({\precsim}\{\lambda x.H_1\}),\dots,\inter N({\precsim}\{\lambda x.H_n\})$. In fact, one can prove that for all $\lambda z.H' \in  \bigcup_{i \in I} {\precsim} \{\lambda z. H_i\}$ (notice that, since $N \in \elam$, $\inter{N}(\bigcup_{i \in I} {\precsim}\{\lambda z. H_i\})= \inter{N}(\bigcup_{i \in I} {\precsim}\{\lambda z. H_i\}\cap \he)$), we can apply Lemma~\ref{lem:  probabilistic assignment entanglement} and get $s^{H'}_{1}$, \ldots, $s_{n}^{H'}\in[0,1]$ such that:
\allowdisplaybreaks
\begin{enumerate}[(1)]
\item $\forall i\leq n$, $\mathscr{E}(\lambda z. H_i)\leq \underset{\lambda z.H' \in {\precsim}( \lambda z.H _i)}{\sum} s_{i}^{H'}$,

\item $\forall \lambda z.H' \in   \bigcup_{i \in I} {\precsim}(\lambda z. H_i)$, $\inter{N}(\lambda z.H')\geq \sum_{i=1}^n s_{i}^{H'}$.
\end{enumerate}
From this, we have:%a simple computation will give the statement (see Appendix for more details).
\begin{align*}
\mathscr{D}(X)
&\leq \sum^n_{i=1}\Bigg( \sum_{\lambda z.H' \in {\precsim}(\lambda z. H_i)} s_{i}^{H'} \Bigg) \cdot \mathscr{F}_{H_i, P} (X)\\
&\leq  \sum^n_{i=1}\sum_{\substack{ \lambda z.H' \in   {\precsim}(\lambda z. H_i)}} s_{i}^{H'}\cdot \inter{H'[P/z]}({\precsim}  (X))\\
%&\leq   \sum^n_{i=1}\sum_{\substack{H' \in \, \bigcup^n_{i=1}{\precsim}(\lambda z.H_i)}} s_{i}^{H'}\cdot \inter{H'[P/z]} ({\precsim}(X)) \\
&\leq  \sum_{\substack{ \lambda z.H' \in\,\bigcup^n_{i=1}  {\precsim}(\lambda z.H_i)}} \bigg(\sum^n_{i=1} s_{i}^{H'}\bigg)\cdot   \inter{H'[P/z]} ({\precsim}(X)) \\
%&\leq  \sum_{\substack{ H' \in\,\bigcup^n_{i=1}  {\precsim}(\lambda z.H_i)}}    \inter{N}(H') \cdot   \inter{H'[P/z]} ({\precsim}(X))&&\text{by}~\eqref{eqn: s t leq u} \\
&\leq  \sum_{\lambda z.H'}\inter{N}(\lambda z.H') \cdot   \inter{H'[P/z]} ({\precsim}(X))
=\inter{NP}({\precsim}(X))
\end{align*}
and hence $\mathscr{D}(X)\leq \inter{NP}({\precsim}(X))$.
%The existence of $s^{H'}_{1}$, \ldots, $s_{n}^{H'}$ relies on Lemma~\ref{lem:  probabilistic assignment entanglement}.
\end{proof}
\begin{lem}\label{lem: context lemma 4} Let $M, N \in \elam$. If $M \precsim N$ then $M \leq_{\mathrm{app}}N$.
\end{lem}
\begin{proof}
 We have to show that $M \precsim N$ implies $\sum \inter{MP_1 \ldots P_n}\leq \sum \inter{N P_1\ldots P_n}$,  for any sequence $P_1, \ldots, P_n \in \elam$. The proof is by  induction on $n$, using Lemma~\ref{lem: context lemmna 1} for the base case and Lemma~\ref{lem: context lemmna 3} for the induction step. %  If $n=0$ then, from $M \precsim N$ and by Lemma~\ref{lem: context lemmna 1}, we have:
%\begin{equation*}
%\begin{split}
%\sum \inter{M}&= \inter{M}(\he)\leq \inter{N}(\precsim(\he))\\
%&= \inter{N}(\he)= \sum \inter{N}\enspace .
%\end{split}
%\end{equation*}
%If $n>0$ then $MP_1\precsim NP_1$ by Lemma~\ref{lem: context lemmna 3}. We conclude by applying the induction hypothesis on $MP_1$ and $NP_1$. 
\end{proof}

%\begin{lem} \label{lem: application distributions} The following statements hold:
%\begin{enumerate}[(1)]
%\item \label{enum: application distribution 1}$\pi: (\lambda x. P)QL_1\ldots L_n \Downarrow \mathscr{D}$  if and only if there exists $\pi': P[Q/x]L_1\ldots L_n\Downarrow \mathscr{D}$ and $\vert \pi \vert = \vert \pi' \vert+1$.
%\item \label{enum: application distribution 2}$\pi: (P_1\oplus P_2)L_1\ldots L_n \Downarrow \mathscr{D}$  if and only if there exist $\pi': P_1L_1\ldots L_n\Downarrow \mathscr{D}_1$ and $\pi'': P_2L_1\ldots L_n\Downarrow \mathscr{D}_2$ such that $\mathscr{D}= \frac{1}{2}\cdot \mathscr{D}_1+ \frac{1}{2}\cdot \mathscr{D}_2$ and $\vert \pi \vert = \vert \pi' \vert +1$.
%\end{enumerate}
%\end{lem}
%\begin{proof}
%
%\end{proof}
%We are now able to prove that $\mathrm{PAS}$ (resp.~$\mathrm{PAB}$) is sound with respect to context preorder (resp.~context equivalence), a first step toward full abstraction.
\begin{thm}[Soundness] \label{thm: soundness new}Let $M, N \in \plam$. Then:
\begin{enumerate}[(1)]
\item \label{eqn: soundness leq} $M \precsim N$ implies $M \leq_{\mathrm{cxt}}N$.
\item  \label{eqn: soundness equal} $M \sim N$ implies $M =_{\mathrm{cxt}}N$.
\end{enumerate}
\end{thm}
\begin{proof}
Point~\ref{eqn: soundness equal}  follows  from point~\ref{eqn: soundness leq} since it holds that   ${\sim}=  {\precsim} \cap {\precsim}^{op}$ (Proposition~\ref{prop: properties dal lago bisimilarity}) and $=_{\mathrm{cxt}}$  is ${\leq_{\mathrm{cxt}}} \cap  {(\leq_{\mathrm{cxt}})}^{op}$. Concerning point~\ref{eqn: soundness leq}, we first prove it  for closed terms. So, let   $M, N \in \elam$ be such that  $M \precsim N$. By Lemma~\ref{lem: context lemma 4}, it holds that $M \leq_{\mathrm{app}}N$. By Lemma~\ref{lem: context lemma}, this implies $M \leq_{\mathrm{cxt}}N$. Now, let $M, N \in \plam^{\lbrace x_1, \ldots, x_n \rbrace}$ be such that  $M \precsim N$. From~\eqref{eqn: open term simil}, we have that  $\lambda x_1\ldots x_n. M \precsim \lambda x_1\ldots x_n. N$.  Because these are closed terms, we obtain  $\lambda x_1\ldots x_n. M \leq_{\mathrm{cxt}} \lambda x_1\ldots x_n. N$.   By repeatedly applying Lemma~\ref{lem: abstraction congruence for obs and app}.\ref{eqn: abstraction congruence for app}, we conclude $M \leq _{\mathrm{cxt}}N$.
\end{proof}

  \section{Full Abstraction}\label{sec4}
We prove that PAB is complete, hence fully abstract (Theorem~\ref{thm: completeness}), while PAS is not, giving a countexemple to PAS completeness in Section~\ref{chap 4 sec 5 subsec 2}.%  with respect to context equivalence. On the contrast, PAS is not complete, fic

As mentioned in the Introduction, the completeness property is usually achieved by transforming PAB into a testing semantics defined by Larsen and Skou \cite{larsen1991bisimulation}, proven equivalent to probabilistic bisimulation by van Breugel et al. \cite{VANBREUGEL2005}, and then showing that every test is definable by a context in the language, see e.g.~\cite{crubille2014probabilistic,kasterovic2019discriminating}.
%In our case, a test is generated by A test is either a state of the considered Markov chain (in our case, either $\tau$ or a closed $\lambda$-term) or a finite nple $(t_1,\dots,t_n)$ 
This reasoning is not so simple to implement in our setting, as the testing definability needs a kind of sampling primitive, which is not clear if representable in a call-by-name semantics (see the discussion in the Introduction). 

Fortunately, we succeed in following a different path, based on Leventis' Separation Theorem \cite{leventis2018probabilistic}. The idea is to prove that (a trivial extension of) the contextual equivalence is a probabilistic applicative bisimulation, hence contained in $\sim$ by definition (Eq.~\eqref{enum: bisimilarity}). Basically, this amounts to check that for any contextual equivalence class $E$ of hnfs and any $M=_{\mathrm{cxt}} N$, we have $\inter{M}(E)=\inter{N}(E)$ (see Eq.~\eqref{eq:bisim_gen}). How to prove it? We associate terms with a kind of infinitary, extensional normal forms, the so-called probabilistic Nakajima trees (Section \ref{chap 4 sec 2 subsec 1}). The Separation Theorem states that two terms $M$ and $N$ share the same Nakajima tree whenever they are contextually equivalent (Theorem \ref{thm: probabilistic separation}), so that we can use such trees as representatives of the contextual equivalence classes. Lemma \ref{lem: inf limit} shows that the quantity $\inter{M}(E)$ depends only on the Nakajima tree of $M$ and that of $E$, so we can conclude with Lemma~\ref{lem: simeq implies same classes with same probability} giving $\inter{M}(E)=\inter{N}(E)$ and hence the full abstraction result Theorem~\ref{thm: completeness}.

% Theorem~\ref{thm: soundness new} tells us that, in order to prove that two programs are  context equivalent, it suffices to check that they are bisimilar, i.e.~that a probabilistic applicative bisimulation exists for them. In this section we show that $\mathrm{PAB}$ and   context equivalence actually coincide, namely that the former is fully abstract with respect to the latter (Theorem~\ref{thm: completeness}). Since the proof of full abstraction is based on the Separation theorem  in Leventis~\cite{leventis2018probabilistic}, this section starts with a short introduction about   probabilistic Nakajima trees,  i.e.~a probabilistic version of  infinitely $\eta$-expanded B\"{o}hm trees.

On the other hand, the counterexample to the completeness of PAS (Eq.~\eqref{eqn: counterexample similarity}) uses the Context Lemma.

\subsection{Probabilistic Nakajima Trees} \label{chap 4 sec 2 subsec 1}
A \emph{B\"{o}hm tree} \cite{barendregt1984lambda} is a labelled tree describing a kind of infinitary normal form of a deterministic $\lambda$-term.  In more details, the B\"{o}hm tree $BT(M)$ of a $\lambda$-term $M$ can be given co-inductively as follows:
\begin{itemize}
\item If the head reduction of $M$ terminates into the hnf $\lambda x_1\ldots x_n. yM_1\ldots M_m $, then:
\[
\begin{tikzpicture}
%[baseline=-1ex]
\node[](root){$\lambda x_1\ldots x_n\textbf{.} y$};
\draw($(root)+(-3, -.5)$) node(e){$BT(M)\triangleq $};
\node[below =of root](c){};
\node[left =of c](1){$BT(M_1)$};
\node[ right =of c](m){$BT(M_m)$};
\draw[-](root)to node[below, right]{$\ \ \  \ \ \ldots$} (1);
\draw[-](root)to (m);
\end{tikzpicture}
\]
where $BT(M_1)$, \ldots, $BT(M_m)$ are the   B\"{o}hm trees of the subterms $M_1, \ldots, M_m$ of the hnf of $M$.
\item Otherwise, the tree is a node labelled by $\mathbf \Omega$. 
\end{itemize}

The  notion of B\"{o}hm tree is not sufficient to characterize    contextual equivalence because it lacks extensionality: the terms $y$ and $\lambda z.yz$ have different B\"{o}hm trees and yet $y=_{\mathrm{cxt}}\lambda z.yz$ holds. To recover extensionality, we need the so-called \emph{Nakajima trees} \cite{Nakajima}, which are infinitely $\eta$-expanded representations of the B\"{o}hm trees. The Nakajima tree $BT^\eta(H)$ of a hnf $H= \lambda x_1\ldots x_n\textbf{.}yM_1\ldots M_m$ is the infinitely branching tree:
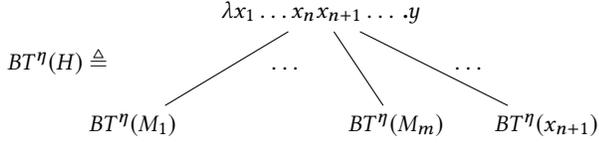
\begin{equation*}
\begin{tikzpicture}
\node[](root){$ \lambda x_1\ldots x_n   x_{n+1}\ldots \,\textbf{.}y$};
\draw($(root)+(-2.5,-1.5)$) node(1){$BT^\eta(M_1)$};
\draw($(root)+(3,-1.5)$) node(x){$BT^\eta(x_{n+1})$};
\draw($(x)+(-2, 0)$) node (m){$BT^\eta(M_m)$};
\draw[-](root) to node[below,right]{$\ \ \ldots$}(x);
\draw[-](root)to node[below, right]{$\ \ \ \ \ \  \ldots$} (1);
\draw[-](root)to (m);
\draw($(root)+(-3.5, -.65)$) node(e){$BT^\eta(H)\triangleq$};
\end{tikzpicture}
\end{equation*}
where  $x_1\ldots x_n   x_{n+1}\ldots$ is an infinite sequence of  pairwise distinct  variables  and, for $i>n$, the $x_i$'s are fresh. 

Nakajima trees represent infinitary $\eta$-long hnfs. Every hnf $H= \lambda x_1\ldots x_n. yM_1\ldots M_m$ $\eta$-expands into the head normal form $\lambda x_1\ldots x_{n+k}.yM_1\ldots M_m x_{n+1}\ldots x_{n+k}$ for any $k \in \mathbb{N}$ and $x_{n+1}\ldots x_{n+k}$ fresh: Nakajima trees are, intuitively, the asymptotical representations of  these $\eta$-expansions. 

To generalize such a construction to probabilistic terms we define by mutual recursion the tree associated with a hnf and the tree of an arbitrary term $M$  as a subprobability distribution over the trees of the hnfs $M$ reduces to. Hence, strictly speaking, a probabilistic Nakajima tree is not properly a tree. 

Following Leventis~\cite{leventis2018probabilistic} we shall  give an inductive, \enquote{level-by-level} definition of the probabilistic Nakajima trees.

The set $\mathcal{PT}^\eta_{\ell}$ of \textit{probabilistic Nakajima trees} with level at most $\ell\in \mathbb{N}$ is the set of subprobability distributions over \textit{value Nakajima trees} $\mathcal{VT}^\eta_\ell$. These sets are defined by mutual recursion as follows:
\allowdisplaybreaks
\begin{align*}
\mathcal{VT}^\eta_0&\triangleq \emptyset\\
  \mathcal{VT}^\eta_{\ell+1}&\triangleq \lbrace \lambda x_1 x_2 \ldots\, \textbf{.}y\, T_1, T_2, \ldots\ \vert \ T_i \in \mathcal{\mathcal{PT}^\eta_\ell} , \ \forall i \geq 1 \rbrace,\\
\mathcal{PT}^\eta_0&\triangleq \lbrace \perp \rbrace, \\
 \mathcal{PT}^\eta_{\ell+1}&\triangleq\lbrace T : \mathcal{VT}^\eta _{\ell+1} \to [0, 1]\ \vert \ \sum_{t \in\, \mathcal{VT}^\eta_{\ell+1}} T(t)\leq 1 \rbrace.
\end{align*}
where $\perp$ represents the zero distribution. Value Nakajima trees are ranged  over by $t$, and probabilistic Nakajima trees are ranged over by $T$.

%We now give an inductive definition of both the probabilistic Nakajima tree associated with a term and the  value Nakajima tree associated with a head normal form. 

 Let $\ell \in \mathbb{N}$. By mutual recursion we define a function  $VT^\eta_{\ell+1}$ associating  with each $H\in \h$  its value Nakajima tree $VT^\eta_{\ell+1}(H)$ of level $\ell+1$, and a function $PT^\eta _\ell$ associating with each $M \in \Lambda_{\oplus}$ its   probabilistic Nakajima tree $PT^\eta _\ell(M) $ of level $\ell$:
% Let $\ell \in \mathbb{N}$. By mutual recursion we define the value Nakajima tree $VT^\eta_{\ell+1}(H)$ of level $\ell+1$ of a $H\in \h$, and the probabilistic Nakajima tree $PT^\eta _\ell(M) $ of level $\ell$ of a $M \in \Lambda_{\oplus}$:
 \begin{itemize}
 \item If $H=\lambda x_1\ldots x_n.yM_1\ldots M_m$, then  $VT^\eta_{\ell+1}(H)$ is:\\
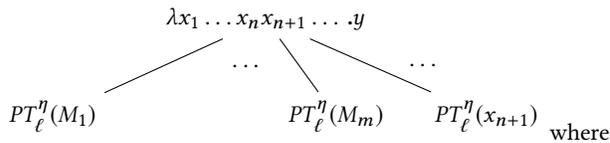
\begin{tikzpicture}[node distance=1cm]
\node[](root){$ \lambda x_1\ldots x_n   x_{n+1}\ldots \,\textbf{.}y$};
\node[ below left=of root](1){$PT^\eta_{\ell} (M_1)$};
\node[below right =of root](x){$PT^\eta _\ell (x_{n+1})$};
\draw($(x)+(-2, 0)$) node (m){$PT^\eta _\ell(M_m)$};
\draw[-](root) to node[below,right]{$\ \ \ \ \   \ldots$}(x);
\draw[-](root)to node[below, right]{$\ \ \  \ \ \ \ \  \ \ \ldots$} (1);
\draw[-](root)to (m);
\end{tikzpicture}
where  $x_1\ldots x_n   x_{n+1}\ldots$ is an infinite sequence of  pairwise distinct  variables  and, for $i>n$, the $x_i$'s are fresh;
\item $PT^\eta_{\ell} (M)\triangleq  
	\begin{cases}   
		t \mapsto \sum_{H \in (VT^\eta_{\ell})^{-1}(t)}  \inter{M}(H)  &\text{if }\ell>0 \\
		\bot &\text{otherwise}. 
	\end{cases}$
 \end{itemize}

We say that $M$ and $N$ \textit{have the same Nakajima tree}, and we write $M =_{\mathrm{PT}^\eta} N$, if $PT^\eta_\ell(M)= PT^\eta_\ell(N)$ holds  for all $\ell \in \mathbb{N}$. 

Theorem~\ref{thm: equivalence head e head spine nel paper} assures that the above definition based on the operational semantics $\inter{\cdot}$ given in \eqref{eq: big-step semantics} is equivalent to the one given by Leventis in~\cite{leventis2018probabilistic}, based on the head reduction.
 \begin{exmp} Figure~\ref{fig: exmp bohm trees} depicts the Nakajima trees of  level, respectively,  $1$ and $2$ associated with  term $\mathbf \Theta(\lambda f. (y\oplus yf))$, where $\mathbf \Theta$ is the Turing fixed-point combinator (Example~\ref{exmp: examples of probabilistic terms}).  Distributions are represented by barycentric sums, depicted as $\oplus$ nodes whose outgoing edges are weighted by probabilities. Notice that the more the level $\ell$ increases, the more the top-level distribution's support grows. 
 \end{exmp}
 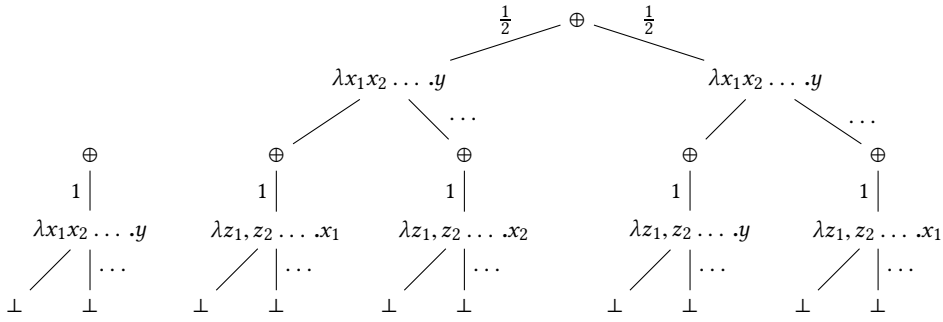
\begin{figure*}[t]
 \centering
 \begin{framed}
  \begin{tikzpicture}[node distance=0.5cm]
 \draw(0, 0) node (center){$\oplus$};
 \draw($(center)+(0, -1)$) node (a){$\lambda x_1 x_2\ldots\,\textbf{.} y$};
 \draw($(a)+(-1, -1)$) node (b1){$\bot$};
 \draw($(a)+(0, -1)$) node (b2){$\bot$};
 
 \draw[-] (center) to node[left]{$1$}(a);
 \draw[-] (a) to (b1);
 \draw[-] (a) to node[right]{$\ldots$}(b2);
  \end{tikzpicture}
\quad 
  \begin{tikzpicture}[node distance=0.5cm]
 \draw(0, 0) node (center){$\oplus$};
 \draw($(center)+(-2.5, -.8)$) node (la){$\lambda x_1 x_2\ldots\,\textbf{.} y$};
  
 \draw($(la)+(-1.5 ,-1)$ )node (llb){$\oplus$};
 \draw($(llb)+(0, -1)$ )node (llc){$\lambda z_1, z_2\ldots \, \textbf{.}x_1$};
 \draw($(llc)+(-1, -1)$ )node (lld1){$\bot$};
 \draw($(llc)+(0, -1)$ )node (lld2){$\bot$};
  
 \draw($(la)+(1 ,-1)$ )node (lb){$\oplus$};
 \draw($(lb)+(0, -1)$ )node (lc){$\lambda z_1, z_2\ldots \, \textbf{.}x_2$};
 \draw($(lc)+(-1, -1)$ )node (ld1){$\bot$};
 \draw($(lc)+(0, -1)$ )node (ld2){$\bot$};

 \draw($(center)+(+2.5, -.8)$ )node (ra){$\lambda x_1 x_2\ldots\,\textbf{.} y$};
  
 \draw($(ra)+(-1 ,-1)$ )node (rb){$\oplus$};
 \draw($(rb)+(0, -1)$ )node (rc){$\lambda z_1, z_2\ldots \, \textbf{.}y$};
 \draw($(rc)+(-1, -1)$ )node (rd1){$\bot$};
 \draw($(rc)+(0, -1)$ )node (rd2){$\bot$};
  
 \draw($(ra)+(1.5 ,-1)$ )node (rrb){$\oplus$};
 \draw($(rrb)+(0, -1)$ )node (rrc){$\lambda z_1, z_2\ldots \, \textbf{.}x_1$};
 \draw($(rrc)+(-1, -1)$ )node (rrd1){$\bot$};
 \draw($(rrc)+(0, -1)$ )node (rrd2){$\bot$};

 \draw[-] (center) to node[left, above]{$\frac{1}{2}$}(la); 
 
 \draw[-] (la) to (llb);
 \draw[-] (llb) to node[left]{$1$}(llc);
 \draw[-] (llc) to (lld1);
 \draw[-] (llc) to node[right]{$\ldots$}(lld2);
 
 \draw[-] (la) to node[right]{$\ \ \ldots$}(lb);
 \draw[-] (lb) to node[left]{$1$}(lc);
 \draw[-] (lc) to (ld1);
 \draw[-] (lc) to node[right]{$\ldots$} (ld2);

 \draw[-] (center) to node[right, above]{$\frac{1}{2}$}(ra); 

\draw[-] (ra) to node[right]{$\ \ \ldots$}(rrb);
 \draw[-] (rrb) to node[left]{$1$}(rrc);
 \draw[-] (rrc) to (rrd1);
 \draw[-] (rrc) to node[right]{$\ldots$} (rrd2);
 
 \draw[-] (ra) to (rb);
 \draw[-] (rb) to node[left]{$1$}(rc);
 \draw[-] (rc) to (rd1);
 \draw[-] (rc) to node[right]{$\ldots$}(rd2);
  \end{tikzpicture}
  
 \caption{From left,  the  Nakajima trees $PT^\eta_{1}(\mathbf \Theta(\lambda f. (y\oplus yf)))$ and $PT^\eta_{2}(\mathbf \Theta(\lambda f. (y\oplus yf)))$.}
 \label{fig: exmp bohm trees}
 \end{framed}
 \end{figure*}
\begin{prop}[\cite{leventis2018probabilistic}]\label{prop: equal bohm tree at depth d implies equal at less depth} Let  $M,N \in \Lambda_\oplus$.  If $PT^\eta_\ell(M)=PT^\eta_\ell(N)$ for some $\ell\in \mathbb{N}$, then $PT^\eta_{\ell'}(M)=PT^\eta_{\ell'}(N)$ for all $\ell' \leq \ell$.
\end{prop}
%\begin{proof}
%It suffices to prove that  $PT^\eta_\ell(M)\not =PT^\eta_\ell(N)$ implies $PT^\eta_{\ell+1}(M)\not = PT^\eta_{\ell+1}(N)$  by induction on $\ell\in \mathbb{N}$.
%\end{proof} 
% The fundamental result for probabilistic Nakajima trees is the following separation property:
 \begin{thm}[Separation~\cite{leventis2018probabilistic}] \label{thm: probabilistic separation} $\!$Let $M, N \in \plam$. If  $M =_{\mathrm{cxt}} N$ then $M =_{\mathrm{PT}^\eta} N$. 
 \end{thm}

 \subsection{The Completeness Theorem}\label{chap 4 sec 5 subsec 1}
% In what follows, we prove that $\mathrm{PAB}$ is fully abstract with respect to  context equivalence. The proof consists on showing the existence of a probabilistic applicative bisimulation  for every pair of  context equivalent terms, and   depends on a key property about probabilistic Nakajima trees  (Lemma~\ref{lem: simeq implies same classes with same probability})  that can be inferred from Theorem~\ref{thm: probabilistic separation}:  given an equivalence class $E \in \elam/=_{\mathrm{cxt}}$ and two  context equivalent terms $M,N\in \elam$, the probability that $M$ reduces to some head normal form in  $E$  is equal to the probability that $N$ does.  
%We now prove that, given an equivalence class $E \in \elam/=_{\mathrm{cxt}}$, and a term $M$, the probability $\inter M(E)$ that $M$ reduces to some hnf $E$ depends only on the Nakajima trees of $M$ and that associated with (any term in) $E$. In order to prove this, we need a notion of approximantion and some auxiliary lemmas. 

In the previous subsection probabilistic Nakajima trees have been inductively presented by introducing \enquote{level-by-level} their finite representations. To recover the full quantitative information of a Nakajima tree we shall need a notion of approximation  together with some general properties.
\begin{defn} Let $r, r' \in \mathbb{R}$ and $\epsilon>0$. We say that $r$ \textit{$\epsilon$-approximates} $r'$, and we write $r \approx_\epsilon r'$, if $\vert r-r' \vert < \epsilon$.
\end{defn} 
\begin{fact} \label{prop: 2epsilon} Let $r, r' , r'' \in \mathbb{R}$ and $\epsilon, \epsilon' >0$. If $r \approx_\epsilon r'$ and $r'\approx_{\epsilon'}r''$ then $r  \approx_{\epsilon +\epsilon'}r''$.
\end{fact}
%\begin{proof}
% We have $ \vert r- r'' \vert= \vert r-r'+r'-r'' \vert \leq \vert r-r' \vert + \vert r'-r'' \vert < \epsilon+\epsilon'$.
%\end{proof}
\begin{lem} \label{fact: descending chain} Let $\lbrace A_n \rbrace_{n \in \mathbb{N}}$ be a  descending chain of countable sets of positive real numbers satisfying $\sum_{r \in A_n}r< \infty$, for all $n \in \mathbb{N}$.  Then:
\begin{equation}\label{eqn: infinite descending chain}
\sum_{r\,  \in \, \bigcap_{n\in \mathbb{N}} A_n }r = \inf_{n \in \mathbb{N}} \bigg( \sum_{r \, \in\, A_n} r \bigg)\enspace .
\end{equation}
\end{lem}

A consequence of Theorem~\ref{thm: probabilistic separation} is that for every contextual  equivalence class $ E \in \elam/ =_{\mathrm{cxt}}$ and for every level  $\ell \in \mathbb{N}$ there exists a \textit{unique} value Nakajima tree $t$ of that level such that $VT^\eta_\ell(H)=t$ for all $H \in E$. Let  $t_{E, \ell}$ denote such a tree.
%says in particular that $VT^\eta_\ell(H)= VT^\eta_\ell(H')$ holds at any level $\ell \in \mathbb{N}$ for every pair of  context equivalent head normal forms $H, H'\in \he$. This means that for every equivalence class $ E \in \elam/ =_{\mathrm{cxt}}$ and for every level  $\ell \in \mathbb{N}$ there exists a \textit{unique} value Nakajima tree $t$ of that level such that $VT^\eta_\ell(H)=t$ for all $H \in E$. Let  $t_{E, \ell}$ denote such a unique tree.
%
\begin{lem}\label{lem: inf limit}  Let $M \in \elam$ and $E \in \elam/ =_{\mathrm{cxt}}$. We have:
\begin{enumerate}[(1)]
\item  \label{enum: inf} $\inter{M}(E)= \inf_{\ell \in \mathbb{N}} \, (PT^\eta_\ell(M)(t_{E, \ell}))$.
\item\label{enum: limit}  $\forall \epsilon>0$ $\exists \ell \in \mathbb{N}$  $\forall \ell' \geq \ell$:  $\inter{M}(E)\approx_\epsilon PT^\eta_{\ell'}(M)(t_{E, \ell'})$.
\end{enumerate}
\end{lem} 
\begin{proof}
Let $E_{\mathsf{V}} \triangleq E\cap\he$, notice that $\inter{M}(E)=\inter{M}(E_{\mathsf{V}})$. As for point~\ref{enum: inf}, we have  $H \in E_{\mathsf{V}}$ if and only if  $\forall \ell \in \mathbb{N}$ $ VT^\eta_\ell(H)=t_{E, \ell}$  if and only if $\forall \ell \in \mathbb{N}$ $H \in (VT^\eta_\ell)^{-1}(t_{E, \ell})$, so that  $E_{\mathsf{V}}  = \bigcap _{\ell \in \mathbb{N}} (VT^\eta_\ell)^{-1}(t_{E, \ell})$. Moreover, by Proposition~\ref{prop: equal bohm tree at depth d implies equal at less depth}, for all $\ell \in \mathbb{N}$ it holds that:
\begin{equation}\label{eqn: monotone decreasing}
\begin{split}
 (VT^\eta_{\ell+1})^{-1}(t_{E, \ell+1})&= \lbrace H \in \he \ \vert \ VT^\eta_{\ell+1}(H)=t_{E, \ell+1} \rbrace \\
 &\subseteq \lbrace H \in \he \ \vert\ VT^\eta_{\ell }(H)=t_{E, \ell}\rbrace\\
 &= (VT^\eta_\ell)^{-1}(t_{E, \ell}) \enspace . 
 \end{split}
\end{equation}
Therefore,  $((VT^\eta_\ell)^{-1}(t_{E, \ell}))_{\ell \in \mathbb{N}}$ is a descending chain, so that $\lbrace \inter{M}(H)\ \vert \ H \in (VT^\eta_\ell)^{-1}(t_{E, \ell}) \rbrace_{\ell \in \mathbb{N}}$ is. Moreover, by definition we have $\sum_{H \in(VT^\eta_\ell)^{-1}(t_{E, \ell}) }\inter{M}(H )\leq \sum \inter{M}\leq 1$, for all $\ell \in \mathbb{N}$.  Hence, by applying  Lemma~\ref{fact: descending chain} and by definition of Nakajima tree equality, we have:
\allowdisplaybreaks
\begin{align*}
\inter{M}(E)=\sum_{H \in E_{\mathsf{V}}}\inter{M}(H)&= \sum_{H \in \, \bigcap _{\ell \in \mathbb{N}} ((VT^\eta_\ell)^{-1}(t_{E, \ell}))} \inter{M}(H)\\
&= \inf _{\ell \in \mathbb{N}} \sum_{H \in (VT^\eta_\ell)^{-1}(t_{E, \ell})} \inter{M}(H)
\\
&= \inf_{\ell \in \mathbb{N}} \, (PT^\eta_\ell(M)(t_{E, \ell})).
\end{align*}
 Let us prove point~\ref{enum: limit}. On the one hand,    $(PT^\eta_{\ell}(M)(t_{E, \ell}))_{\ell \in \mathbb{N}}$   is clearly a bounded below sequence. On the other hand, from~\eqref{eqn: monotone decreasing} it is also monotone decreasing. Indeed, for all $\ell \in\mathbb{N}$:
\allowdisplaybreaks
\begin{align*}
PT^\eta_{\ell+1}(M)(t_{E, \ell+1})&= \sum_{H \in (VT^\eta_{\ell+1})^{-1}(t_{E, \ell+1})}\inter{M}(H)\\
& \leq \sum_{H \in (VT^\eta_\ell)^{-1}(t_{E, \ell})}\inter{M}(H)=PT^\eta_\ell(M)(t_{E, \ell}).
\end{align*}
 Thus,  $\lim_{\ell \rightarrow \infty} (PT^\eta_\ell(M)(t_{E, \ell}))_{\ell \in \mathbb{N}}=  \inf_{\ell \in \mathbb{N}} \, (PT^\eta_\ell(M)(t_{E, \ell}))=\inter{M}(E)$, and  point~\ref{enum: limit} follows by definition of limit. 
\end{proof} 
\begin{lem}\label{lem: simeq implies same classes with same probability} Let $M, N \in \elam$. If  $M =_{\mathrm{cxt}} N$ then  $\inter{M}(E)= \inter{N}(E)$,  for all $E \in \elam/ =_{\mathrm{cxt}}$.
\end{lem}
\begin{proof}
Suppose toward contradiction that $\inter{M}(E)\neq \inter{N}(E)$ and consider $\epsilon>0$ such that $2\epsilon \leq  \vert \inter{M}(E)- \inter{N}(E)\vert$. By Lemma~\ref{lem: inf limit}.\ref{enum: limit}  there exist $\ell\in \mathbb{N} $ such that:
\begin{align*}
 &\inter{M}(E)  \approx_\epsilon  PT^\eta_\ell(M)(t_{E, \ell})& &\inter{N}(E) \approx _\epsilon  PT^\eta_\ell(N)(t_{E, \ell})\enspace. 
\end{align*}
By Theorem~\ref{thm: probabilistic separation}, from $M =_{\mathrm{cxt}} N$ we obtain  $M =_{\mathrm{PT}^\eta}N$, and hence $PT^\eta_{\ell}(M)= PT^\eta_\ell(N)$.  By Fact~\ref{prop: 2epsilon}, $\inter{M}(E) \approx_{2\epsilon}\inter{N}(E)$, i.e.~$\vert \inter{M}(E) -\inter{N}(E)\vert < 2\epsilon$. A contradiction.
\end{proof}

\begin{rem}
Observe that the statement of Lemma~\ref{lem: simeq implies same classes with same probability} may fail when $\Lambda_\oplus$ is endowed with a different operational semantics than head reduction. As an example, recall the terms  $M\triangleq\lambda xy.(x\oplus y) $ and $N\triangleq (\lambda xy. x)\oplus (\lambda xy.y)$ discussed in the Introduction (Eq. \eqref{ex:no_cbn_fa}). In the lazy cbn, $M$ and $N$ are contextually  equivalent \cite{dal2014coinductive}.   Moreover,  $M$  is a value for lazy cbn, while $N$ reduces with equal probability $\frac{1}{2}$ to $\mathbf T = \lambda xy.x $  and $\mathbf F = \lambda xy.y$. However, $M$, $\mathbf T$ and $\mathbf F$ are pairwise contextually inequivalent since, by setting $\mathcal{C}=[\cdot]\mathbf I\mathbf  \Omega$, we have that  $\mathcal{C}[M]$, $\mathcal{C}[\mathbf T]$,  and $\mathcal{C}[\mathbf F]$ converge with probability  $\frac{1}{2}$, $1$, and $0$, respectively. Therefore, by setting $E$ as the lazy cbn contextual equivalence class containing $M$, we have $\inter M (E)=1$, while $\inter N(E)=0$.
\end{rem}

%We can  now state and prove the fundamental result of this paper:
\begin{thm}[Full abstraction] \label{thm: completeness} For all $M, N \in \plam$:
\begin{equation*}
M =_{\mathrm{cxt}} N\Leftrightarrow M \sim N \enspace . 
\end{equation*}
\end{thm}
\begin{proof}
The right-to-left direction is Theorem~\ref{thm: soundness new}.\ref{eqn: soundness equal}. Concerning  the converse,  we first consider the case of closed terms. So, let $M, N \in \elam$ be such that $M =_{\mathrm{cxt}}N$. We prove that there exists probabilistic applicative bisimulation $\mathcal{R}$ containing $=_{\mathrm{cxt}}$. We define $\mathcal{R}$ as follows:
\begin{multline*}
\lbrace (P,Q) \in \elam \times \elam \ \vert \  P =_{\mathrm{cxt}} Q \rbrace\\
\cup  \lbrace (\nu x.H,\nu x. H')\in \dhe \times \dhe   \ \vert \ \lambda x.H =_{\mathrm{cxt}}\lambda x.H' \rbrace 
 \enspace .
\end{multline*}
Let us prove that $\mathcal{R}$ is a probabilistic applicative bisimulation. Since   $=_{\mathrm{cxt}}$ is an equivalence relation, then $\mathcal{R}$ is. Now, let $(\nu x. H, \nu x. H'), (P, Q) \in \mathcal{R}$, $E \in (\elam \cup \dhe )/ \mathcal{R}$, and let $l \in \elam \cup \lbrace \tau \rbrace$. We have to show that:
\begin{enumerate}[(1)]
\item  \label{enum: term bisimulation}$\mathcal{P}_\oplus(P, l, E)=\mathcal{P}_\oplus(Q,l, E)$,
\item \label{enum: value bisimulation} $\mathcal{P}_\oplus(\nu x. H, l, E)=\mathcal{P}_\oplus(\nu x. H',l, E)$.
\end{enumerate} 
Let us  prove point~\ref{enum: term bisimulation}.  If  $l \in \elam$ then $\mathcal{P}_\oplus(P, l, E)=0=\mathcal{P}_\oplus(Q, l, E)$. If $l=\tau$ we define $\widehat{E}\triangleq \lbrace \lambda  x.H \in \he \ \vert \ \nu x.H \in E \rbrace \cup \lbrace P' \in \elam \ \vert \ P' \in E \rbrace$. Then, by definition: 
\allowdisplaybreaks
\begin{align*}
&\mathcal{P}_\oplus(P, \tau, E)=\inter{P}(\widehat{E}) \enspace \qquad  \mathcal{P}_\oplus(Q, \tau, E)= \inter{Q}(\widehat{E})\enspace .
\end{align*}
Since $(P, Q)\in  \mathcal{R}$ and $E \in (\elam \cup \dhe )/ \mathcal{R}$, it holds that $P =_{\mathrm{cxt}}Q$ and  $\widehat{E} \in \elam/_{=_{\mathrm{cxt}}}$. By applying  Lemma~\ref{lem: simeq implies same classes with same probability} we have $\inter{P}(\widehat{E})=\inter{Q}(\widehat{E})$, and hence $\mathcal{P}_\oplus(P, \tau, E)=\mathcal{P}_\oplus(Q,\tau, E)$.\\
Let us now prove  point~\ref{enum: value bisimulation}. If $l = \tau$ then $P_\oplus(\nu x. H,\tau,  E)=0=P_\oplus(\nu x. H', \tau, E)$. Otherwise, let  $l= L \in \elam$.  Since  ${=_{\mathrm{cxt}}}  $ is $  {\leq _{\mathrm{cxt}}}\cap  {(\leq_{\mathrm{cxt}})}^{op}$, by  Lemma~\ref{lem: abstraction congruence for obs and app}.\ref{eqn: application congruence for obs} we have that $\lambda x. H =_{\mathrm{cxt}} \lambda x. H'$ implies   $(\lambda x. H)L =_{\mathrm{cxt}} (\lambda x. H')L$. From Proposition~\ref{prop: the semantics is invariant under reduction}.\ref{enum: invariance beta} and Proposition~\ref{prop: same small-step implies context equivalent} we have:
\begin{equation*}
H[L/x]  =_{\mathrm{cxt}}  (\lambda x. H)L   =_{\mathrm{cxt}} (\lambda x. H')L =_{\mathrm{cxt}}  H'[L/x] \enspace .
\end{equation*}
 Therefore,  $H[L/x] \in E$ if and only if $H'[L/x] \in E$, and hence $\mathcal{P}_\oplus(\nu x.H, L, E)=\mathcal{P}_\oplus(\nu x. H',L, E)$.\\
 Now,  let $M, N \in \plam^{\lbrace x_1, \ldots, x_n\rbrace}$ be such that $M =_{\mathrm{cxt}}N$. Since ${=_{\mathrm{cxt}}} $ is   ${\leq _{\mathrm{cxt}}}\cap {(\leq_{\mathrm{cxt}})}^{op}$, by repeatedly applying  Lemma~\ref{lem: abstraction congruence for obs and app}.\ref{eqn: abstraction congruence for obs} and Lemma~\ref{lem: context lemma}.\ref{eqn: context lemma leq},   $\lambda x_1\ldots x_n. M =_{\mathrm{cxt}}\lambda x_1\ldots x_n. N$. Since these terms are closed, we obtain $\lambda x_1\ldots x_n. M \sim \lambda x_1\ldots x_n. N$. Finally, from~\eqref{eqn: open term bisimil} we conclude  $M \sim N$.
\end{proof}

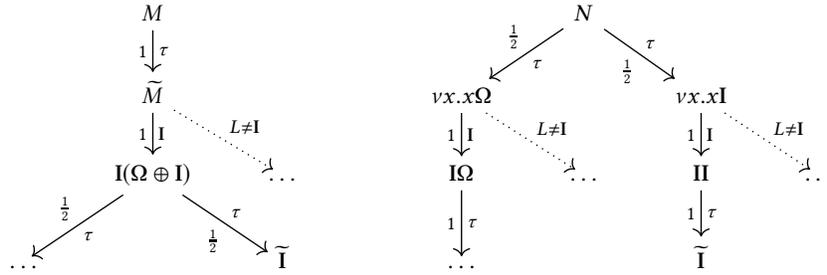
\begin{figure*}[th]
 \centering
 \begin{framed}
\[
\begin{tikzcd}[ampersand replacement = \&]
               \&  M   \arrow[d,  "1 " swap,   "\tau "]            \&    \& \&                               \&  N  \arrow[ld,  "\frac{1}{2} "swap,   "\tau " ]  \arrow[rd,  "\frac{1}{2} " swap,   "\tau "] \& \& \\
               \&  \wt{M}  \arrow[rd, dotted, "L\not =\mathbf I"]   \arrow[d,  "1" swap , "\mathbf I "]     \&    \& \&        \nu x.x\mathbf\Omega \arrow[rd,dotted, "L \not =\mathbf I"]  \arrow[d, "\mathbf I", "1" swap] \&      \& \nu x.x\mathbf I \arrow[d, "\mathbf I", "1" swap] \arrow[rd, dotted, "L\not =\mathbf I"] \&\\
               \& \mathbf I(\mathbf\Omega \oplus\mathbf I)  \arrow[ld,  "\frac{1}{2} " swap,   "\tau " ]  \arrow[rd,  "\frac{1}{2} " swap,   "\tau "] \&   \ldots    \& \&           \mathbf I\mathbf\Omega   \arrow[d, "1 " swap,"\tau"]   \&    \ldots     \&          \mathbf I\mathbf I              
             \arrow[d, "1 " swap,"\tau"]         \& \ldots  \\
\ldots \&                                  \& \wt{\mathbf I}  \& \&     \ldots                 \&         \&\wt{\mathbf I} \&
\end{tikzcd}
\]
 \vspace{-.4cm}
 \caption{Markov chain for $M=\lambda x.x (\mathbf\Omega \oplus\mathbf I)$ and $N=\lambda x.(x\mathbf\Omega \oplus x\mathbf I)$.}
 \label{fig: counterexample similarity}
 \end{framed}
\end{figure*}
\subsection{$\mathrm{PAS}$ is Not Complete}\label{chap 4 sec 5 subsec 2}
 Theorem~\ref{thm: completeness} establishes a precise correspondence between $\mathrm{PAB}$ and   contextual equivalence. But  what about $\mathrm{PAS}$ and   contextual preorder? The soundness theorem (Theorem~\ref{thm: soundness new}) states that the former implies the latter, so that it is natural to wonder whether the converse holds as well. Surprisingly enough, as in the case of the lazy reduction strategies (see~\cite{dal2014coinductive} and~\cite{crubille2014probabilistic}), the answer is negative.

A counterexample to PAS completeness is given by:
\begin{align}\label{eqn: counterexample similarity}
M&\triangleq \lambda x.x (\mathbf\Omega \oplus\mathbf I), &N&\triangleq \lambda x.(x\mathbf\Omega \oplus x\mathbf I).
\end{align}
whose Markov chain is sketched in Figure~\ref{fig: counterexample similarity}. First, observe that $M$ and $N$ are incomparable with respect to PAS:
\begin{lem}\label{lem: counterexample to similarity, M N not similar}  
Neither $M \precsim N$ nor $N \precsim M$ hold.
\end{lem}
\begin{proof}
Let $M \precsim N$. Then,   $\mathcal{P}_\oplus(M, \tau, \wt{M})\leq \mathcal{P}_\oplus(N, \tau, {\precsim}(\wt{M}))$, so that $ \nu x.x\mathbf \Omega \in  {\precsim}(\wt{M})$, and $\wt{M}\precsim \nu x.x\mathbf\Omega$.  Hence,  $ \mathcal{P}_\oplus(\wt{M}, \mathbf I, \mathbf I(\mathbf\Omega \oplus\mathbf I))\leq \mathcal{P}_\oplus(\nu  x. x\mathbf\Omega, \mathbf I , {\precsim} (\mathbf I(\mathbf\Omega \oplus\mathbf I)))$. This means that $\mathbf I\mathbf \Omega \in  {\precsim} (\mathbf I(\mathbf\Omega \oplus\mathbf I))$, so that $\mathbf I(\mathbf\Omega \oplus\mathbf I)\precsim\mathbf I \mathbf\Omega$. So  $\frac{1}{2}= \mathcal{P}_\oplus(\mathbf I(\mathbf\Omega \oplus\mathbf I), \tau, \wt{\mathbf I})\leq \mathcal{P}_\oplus(\mathbf I \mathbf\Omega, \tau, {\precsim}( \wt{\mathbf I}))=0$. A contradiction.

Now, suppose $N \precsim M$. Then we have   $\mathcal{P}_\oplus(N, \tau, \nu  x.x\mathbf I)\leq \mathcal{P}_\oplus(M, \tau, {\precsim}(\nu x.x\mathbf I))$, so that $\wt{M}\in  {\precsim} (\nu x.x\mathbf I)$, and  $\nu x.x\mathbf I \precsim \wt{M}$. Hence, $\mathcal{P}_\oplus(\nu  x.x\mathbf I, \mathbf I, \mathbf I\mathbf I)\leq \mathcal{P}_\oplus(\wt{M},\mathbf I, {\precsim}(\mathbf I\mathbf I))$. This means that $\mathbf I(\mathbf \Omega \oplus\mathbf I)\in  {\precsim} (\mathbf I\mathbf I)$, so that $\mathbf I\mathbf I \precsim \mathbf I(\mathbf\Omega \oplus\mathbf I)$. Therefore,  $1=\mathcal{P}_\oplus(\mathbf I\mathbf I, \tau,  \wt{\mathbf I})\leq \mathcal{P}_\oplus(\mathbf I(\mathbf\Omega \oplus\mathbf I), \tau, {\precsim}( \wt{\mathbf I}))= \frac{1}{2}$. A contradiction.
\end{proof}

However, the two terms can be compared through the  contextual preorder relation:
\begin{lem} \label{lem: counterexample M and N context preorder } It holds that $M \leq_{\mathrm{cxt}} N$.
\end{lem}
\begin{proof}[Proof (sketch)]
By Lemma~\ref{lem: context lemma} it is enough to show that $M \leq_{\mathrm{app}}N$.  Since $M, N \in \elam$, this amounts to check that for all finite sequences $L_1, \ldots, L_n \in \elam$, it holds that $  \sum \inter{M L_1\ldots L_n}\leq \sum \inter{NL_1\ldots L_n}$. The proof easily follows once one has:
\begin{equation}
  \inter{L[\mathbf\Omega/x]}\leq_\dist    \inter{L[\mathbf I/x]}, \label{enum: preliminary counterex 1}
\end{equation}
\begin{equation}
\textstyle \sum \inter{L[(\mathbf\Omega \oplus\mathbf I)/x]}\leq \frac{1}{2}\cdot \sum \inter{L[\mathbf\Omega/x]}+ \frac{1}{2}\cdot \sum \inter{L[\mathbf I/x]} ,\label{enum: preliminary counterex 2}
\end{equation}
for any term $L$. The first inequation is an easy consequence of Proposition \ref{prop: same small-step implies context equivalent}, while the second one can be proven by induction on an approximation of $\inter{\cdot}$.
\end{proof}
% Summing up, we obtain the following:
\begin{thm}  \label{thm: pas is not fully abstract}$\mathrm{PAS}$ is not complete (hence fully abstract) with respect to contextual preorder.
\end{thm}

\section{Conclusion}\label{sec5}
We have considered the untyped probabilistic $\lambda$-calculus $\plam$ endowed with an operational semantics based on the head spine reduction, a variant of the head reduction strategy giving rise to the same big-step semantics (Theorem~\ref{thm: equivalence head e head spine nel paper}). We have proven that probabilistic applicative bisimilarity is fully abstract with respect to contextual equivalence (Theorem~\ref{thm: completeness}). The soundness part is a consequence of a Context Lemma (Lemma~\ref{lem: context lemma}). The completeness proof relies on the Separation Theorem, introducing probabilistic Nakajima trees~\cite{leventis2018probabilistic}.  

Our result completes the picture about fully abstract descriptions of the probabilistic head reduction contextual equivalence, adding finally a coinductive characterisation. To the best of our knowledge, this picture can be resumed by the equivalences of all the following items, for $M$ and $N$ probabilistic $\lambda$-terms: 
\begin{enumerate}
\item\label{item:cxt} $M$ and $N$ are contextually equivalent,
\item\label{item:separation} $M$ and $N$ have the same probabilistic Nakajima tree \cite{leventis2018probabilistic,leventis2019strong},
\item\label{item:game} $M$ and $N$ have the same denotation in the reflexive arena $\mathcal U$  of the cartesian closed category of probabilistic concurrent game semantics \cite{ClairambaultP18},
\item\label{item:pcoh} $M$ and $N$ have the same denotation in the reflexive object $\mathcal D^\infty$ of the cartesian closed category of probabilistic coherence spaces or of the $\mathbb R^+$-weighted relations \cite{ClairambaultP18,leventis2019strong},
\item\label{item:bisimilarity} $M$ and $N$ are applicatively bisimilar (this paper),
\item\label{item:testing} $M$ and $N$ are testing equivalent according to the testing language $\mathtt{T}_0$ (a consequence of \cite{VANBREUGEL2005} and this paper).
\end{enumerate}

%Notice that none of the above equivalences is a trivial result, but requires some of the more advances proof techniques of the untyped $\lambda$-calculus, e.g.~ approximations (for \eqref{item:pcoh}), or B\"ohm's out (for \eqref{item:separation}). 

%In~\cite{leventis2019strong} it has been shown that both probabilistic Nakajima tree equivalence and  denotational equivalence in the model of probabilistic coherence spaces (see for example~\cite{danos2011probabilistic}) are fully abstract with respect to   context equivalence for the head reduction. Our result allows to widen this framework by adding a connection to bisimilarity. In particular, in~\cite{dal2014coinductive} it is shown that in the lazy call-by-name probabilistic $\lambda$-calculus there are context equivalent terms that are not bisimilar. Thus, our result highlights the crucial role  in discriminating  terms played by lazyness under a   call-by-name policy.

Last, we have shown a counterexample to the full abstraction problem for probabilistic applicative similarity (Equation~\eqref{eqn: counterexample similarity}).  We conjecture that  extending the calculus with Plotkin’s parallel disjunction \cite{plotkin1977lcf}, as done in~\cite{crubille2015applicative}, is enough to restore this property. This is left to future work.

%% Acknowledgments
\begin{acks}                            %% acks environment is optional
                                        %% contents suppressed with 'anonymous'
  %% Commands \grantsponsor{<sponsorID>}{<name>}{<url>} and
  %% \grantnum[<url>]{<sponsorID>}{<number>} should be used to
  %% acknowledge financial support and will be used by metadata
  %% extraction tools.
%  This material is based upon work supported by the
%  \grantsponsor{GS100000001}{National Science
%    Foundation}{http://dx.doi.org/10.13039/100000001} under Grant
%  No.~\grantnum{GS100000001}{nnnnnnn} and Grant
%  No.~\grantnum{GS100000001}{mmmmmmm}.  Any opinions, findings, and
%  conclusions or recommendations expressed in this material are those
%  of the author and do not necessarily reflect the views of the
%  National Science Foundation.
We would like to thank T. Ehrhard, R. Crubill\'e, V. Vignudelli and the anonymous reviewers for useful comments and discussions. 
This material is based upon work supported by the French  \grantsponsor{ANR}{{\it ANR : ``Agence National de Recherche''}}{https://anr.fr} under Grant ``PPS: Probabilistic Program Semantics'',
 No.~\grantnum{ANR-19-CE48-0014}{ANR-19-CE48-0014}.

\end{acks}

% Bibliography
\bibliography{main}

\newpage
%% Appendix
\appendix
\section{Proofs of Section~\ref{sec2}}
%\textbf{Lemma~\ref{lem: omega complete poset}}.   $(\dist{(\h)}, \leq_\dist)$ is a  dcpo with least element $\bot$.
%\begin{proof}
%It is easy to check that $\leq_\dist$  is a partial order with least element the empty distribution $\perp$. Moreover, the supremum of a directed subset $\emptyset \not = X \subseteq \dist(\h)$ is pointwise defined, for all $H \in \h$,  by $(\sup_{\mathscr{D} \in X}X)(H)\triangleq  \sup_{\mathscr{D} \in X} \mathscr{D}(H)$, that is  a head distribution:
%\begin{equation*}
%\begin{split}
%\sum_{H \in \h}(\sup_{\mathscr{D} \in X}X)(H) &=  \sum_{H \in \h}   \sup_{\mathscr{D} \in X} \mathscr{D}(H)\\
%& =   \sup_{\mathscr{D} \in X}  \sum_{H \in \h}  \mathscr{D}(H)\leq 1\enspace.
%\end{split}
%\end{equation*}
%\end{proof}
%\noindent
This  proves that the set $\lbrace \mathscr{D}\in \dist{(\h)}\ \vert \ M \Downarrow \mathscr{D} \rbrace$ is directed for all $M \in \plam$. 
\begin{lem}\label{lem: direct set} For every $M\in \Lambda_\oplus$, if $M\Downarrow \mathscr{D}$ and $M \Downarrow \mathscr{E} $ then  there exits $\mathscr{F} \in \dist(\h)$  such that $M \Downarrow \mathscr{F}$ and $\mathscr{D},\mathscr{E}  \leq_\dist \mathscr{F}$.
\end{lem}
\begin{proof}
By induction on the structure of the derivations of $M\Downarrow \mathscr{D}$ and $M \Downarrow \mathscr{E} $. If  $\mathscr{D}= \bot$ then $\mathscr{F}\triangleq\mathscr{E}$. Similarly, if  $\mathscr{E}= \bot$ then $\mathscr{F}\triangleq \mathscr{D}$. Otherwise, we consider the structure of $M$. If $M$ is a variable, say $x$, then the last rule of both  $M\Downarrow \mathscr{D}$ and $M \Downarrow \mathscr{E} $  is $s2$,  and  we set $\mathscr{F}\triangleq x$.
If $M$ is an abstraction, say $\lambda x. M'$, then  the last rule of both  $M\Downarrow \mathscr{D}$ and $M \Downarrow \mathscr{E} $ is $s3$:
\begin{center}
\AxiomC{$M' \Downarrow \mathscr{D}'$}
\RightLabel{$s3$}
\UnaryInfC{$\lambda x. M' \Downarrow \mathscr{D}$}
\DisplayProof 
\qquad  \qquad 
\AxiomC{$M' \Downarrow \mathscr{E}'$}
\RightLabel{$s3$}
\UnaryInfC{$\lambda x. M' \Downarrow \mathscr{E}$}
\DisplayProof
\end{center}
By induction hypothesis, there exists $\mathscr{F}'$ such that $M' \Downarrow \mathscr{F}'$ and $\mathscr{D}', \mathscr{E}' \leq_\dist \mathscr{F}'$, so that we set  $\mathscr{F}\triangleq \lambda x. \mathscr{F}'$. 
If $M$ is an application, say $M'N$, then the last rule of both $M\Downarrow \mathscr{D}$ and $M \Downarrow \mathscr{E} $  is $s4$:
\begin{mathpar}
\inferrule*[Right=$s4$]{M' \Downarrow \mathscr{D}'\\  \lbrace H[N/x]\Downarrow \mathscr{D}''_{H, N} \rbrace_{\lambda x. H \,  \in \,\mathrm{ supp}(\mathscr{D}')}}{M'N \Downarrow \mathscr{D}}\\
\inferrule*[Right=$s4$]{M' \Downarrow \mathscr{E}'\\  \lbrace H[N/x]\Downarrow \mathscr{E}''_{H, N} \rbrace_{\lambda x. H \,  \in \,\mathrm{ supp}(\mathscr{E}')}}{M'N \Downarrow \mathscr{E}}
\end{mathpar}
 By induction hypothesis, there exist $\mathscr{F}'$ such that $M' \Downarrow \mathscr{F}'$  and $\mathscr{D}', \mathscr{E}'\leq_\dist \mathscr{F}'$.  Moreover, for all $H \in \mathrm{ supp}(\mathscr{F}')$, if  $ H \in  \mathrm{ supp}(\mathscr{D}')\, \cap    \mathrm{ supp}(\mathscr{E}')$ then, by induction hypothesis, there exists  $ \mathscr{G}''_{H, N}$ such that $H[L/x]\Downarrow \mathscr{G}''_{H, N}$ and  $ \mathscr{D}''_{H, N} ,  \mathscr{E}''_{H, N}\leq_\dist   \mathscr{G}''_{H, N}$. Hence, we set:
\begin{equation*}
\mathscr{F}''_{H, N}\triangleq \begin{cases}
\mathscr{D}''_{H, N}&\text{if }H \in  \mathrm{ supp}(\mathscr{D}') \text{ and }H \not \in  \mathrm{ supp}(\mathscr{E}'),\\
\mathscr{E}''_{H, N}&\text{if }H \in  \mathrm{ supp}(\mathscr{E}') \text{ and }H \not \in  \mathrm{ supp}(\mathscr{D}'),\\
 \mathscr{G}''_{H, N}&\text{if } H \in  \mathrm{ supp}(\mathscr{D}')\, \cap    \mathrm{ supp}(\mathscr{E}'),\\
 \bot &\text{otherwise} .
\end{cases}
\end{equation*}
Then, we define $ \mathscr{F}$ as:
\begin{equation*}
 \sum_{\lambda x. H\,  \in\, \mathrm{ supp}(\mathscr{F}')} \mathscr{F}'(\lambda x. H)\cdot \mathscr{F}''_{H, N} + \sum_{\substack{ H  \in\, \mathrm{ supp}(\mathscr{F}')\\  \cap\, \neu}} \mathscr{F}'(H)\cdot H N .
\end{equation*}
 The last case is when $M$ is a probabilistic sum, say $M' \oplus M''$. Then the last rule of both $M\Downarrow \mathscr{D}$ and $M \Downarrow \mathscr{E} $ is $s5$:
 \begin{mathpar}
\inferrule*[Right=$s5$]{M' \Downarrow \mathscr{D}'\\  M'' \Downarrow \mathscr{D}''}{M'\oplus M'' \Downarrow \mathscr{D}}\and
\inferrule*[Right=$s5$]{M' \Downarrow \mathscr{E}'\\  M'' \Downarrow \mathscr{E}''}{M'\oplus M'' \Downarrow \mathscr{E}}
\end{mathpar}
By induction hypothesis, there exist $\mathscr{F}'$ and $\mathscr{F}''$ such that $M' \Downarrow \mathscr{F}'$ and $\mathscr{D}',\mathscr{E}' \leq_\dist \mathscr{F}'$, as well as  $ M'' \Downarrow \mathscr{F}''$ and  $\mathscr{D}'', \mathscr{E}'' \leq_\dist \mathscr{F}''$. Then, it suffices to define $\mathscr{F}\triangleq  \frac{1}{2}\cdot \mathscr{F}'+ \frac{1}{2}\cdot \mathscr{F}''$.
\end{proof}
%The following is a direct consequence of the definition of supremum:
%\begin{lem} \label{lem: inter < inter equivalent definition}  Let $M, N \in \plam$. Then, 
% $\inter{M}\leq_\dist \inter{N}$ if and only if,  for all $ \mathscr{D}\in \dist{(\h)}$ such that  $ M \Downarrow \mathscr{D}$, there exits  $ \mathscr{E}\in \dist{(\h)}$ satisfying both  $ N \Downarrow \mathscr{E}$ and $ \mathscr{D}\leq_\dist \mathscr{E} $.
%\end{lem}
\noindent
\textbf{Proposition~\ref{prop: the semantics is invariant under reduction}.}  For every $M, N \in \plam$ and $H \in \h$:
\begin{enumerate}[(1)]
\item  \label{lem: invariance beta general case}  $\inter{MN}$ is equal to the following distribution:
 \begin{equation*}
\begin{split}
& \sum_{\lambda x. H\,  \in\, \mathrm{ supp}(\inter{M})} \inter{M}(\lambda x.H)\cdot \inter{H[N/x]} \\
& + \sum_{ H \, \in\, \mathrm{ supp}(\inter{M})\,\cap\, \neu} \inter{M}(H)\cdot H N .
\end{split}
\end{equation*}
\item  \label{enum: invariance beta}$\llbracket (\lambda x. H)N \rrbracket = \llbracket H[N/x] \rrbracket$.
\item \label{enum: invariance abs} $\inter{\lambda x. M}= \lambda x. \inter{M}$.
\item \label{enum: invariance sum} $\llbracket M \oplus N \rrbracket = \frac{1}{2}\llbracket M \rrbracket + \frac{1}{2} \llbracket N \rrbracket$.
\end{enumerate}
Moreover, for every $H \in\h$,  $\inter{H}=H$.
\begin{proof}
 First, we prove point~\ref{lem: invariance beta general case}.  Let $\mathscr{D}$ be such that $MN \Downarrow \mathscr{D}$.  The case $\mathscr{D}=\bot$  is trivial, so suppose $\mathscr{D}\not =\bot$. Then,   $MN \Downarrow \mathscr{D}$ must be obtained by applying the rule $s4$ to the premises $M \Downarrow \mathscr{E}$ and $\lbrace H[N/x] \Downarrow  \mathscr{F}_{H, N} \rbrace_{\lambda x. H\, \in\, \mathrm{supp}(\mathscr{E})}$,  so that $\mathscr{D}$ is of the form:
\begin{equation}\label{eqn: structure of D if applied rules s4}
 \sum_{\lambda x. H\,  \in \, \mathrm{ supp}(\mathscr{E})} \mathscr{E}(\lambda x. H)\cdot \mathscr{F}_{H, N}\ + \sum_{\substack{ H \, \in\, \mathrm{ supp}(\mathscr{E})\\ \cap\, \neu}} \mathscr{E}(H)\cdot H N
\end{equation}
This proves the $\leq_\dist$ direction. For the converse, suppose that   $\mathscr{E}$ is a head distribution such that $M \Downarrow \mathscr{E}$  and,  for all  $\lambda x. H\,  \in \, \mathrm{ supp}(\mathscr{E})$, suppose $\mathscr{F}_{H, N}$ is a head distribution such that  $H[N/x]\Downarrow \mathscr{F}_{H, N}$. By applying rule $s4$, we get  $MN\Downarrow \mathscr{D}$, where $\mathscr{D}$ is as in~\eqref{eqn: structure of D if applied rules s4}, and the result follows.\\
Point~\ref{enum: invariance beta} is a special case of point~\ref{lem: invariance beta general case} where  $M = \lambda x. H$.  So, let us prove point~\ref{enum: invariance abs}.  As for  the $\leq_\dist$ direction, suppose $\lambda x. M \Downarrow \mathscr{D}$. The case $\mathscr{D}=\bot$ is trivial, so suppose $\mathscr{D}\not =\bot$.  Then,  $\lambda x. M \Downarrow \mathscr{D}$ must be obtained from $M \Downarrow \mathscr{D}'$ by applying rule $s3$, where $\mathscr{D}= \lambda  x. \mathscr{D}'$, so that $\inter{\lambda x. M}\leq_\dist \lambda x. \inter{M}$. For the converse, suppose $\mathscr{D}$ is a  head distribution such that $M \Downarrow \mathscr{D}$. By applying rule $s3$ we get $\lambda x. M\Downarrow \lambda x. \mathscr{D}$, so that $\lambda x. \inter{M}\leq_\dist \inter{\lambda x. M}$. Point~\ref{enum: invariance sum} is similar. \\
Finally,  for all $H \in \h$,  we prove $\inter{H}=H$  by induction on the structure of $H$. If $H$ is a variable, say $x$, then $\inter{x}=x$. If $H$ is an abstraction, say $\lambda x. H'$, then  $H'$ is a head normal form. By  induction hypothesis, $\inter{ìH}=H'$. By point~\ref{enum: invariance abs} we have $\inter{\lambda x. H'}= \lambda x. \inter{H'}= \lambda x. H'$. Last, if $H$ is an application, say $MN$, then $M$ must be of the form $x P_1\ldots P_n$. By point~\ref{lem: invariance beta general case}, we have  $\inter{MN}= \inter{x P_1\ldots P_n}(x P_1\ldots P_n)\cdot x P_1\ldots P_n N= x P_1\ldots P_n N$. 
\end{proof}
\noindent
\textbf{Lemma~\ref{lem: operational semantics monotonicity contexts}.} Let $M, N \in \Lambda_\oplus$. If $\inter{M}\leq_{\dist}\inter{N}$ then  $\forall \mathcal{C}\in \mathsf{C}\Lambda_\oplus$ $\inter{\mathcal{C}[M]}\leq_{\dist}\inter{\mathcal{C}[N]}$.
\begin{proof}
By structural induction on the context $\mathcal{C}\in \mathsf{C}\Lambda_\oplus$. The case $\mathcal{C}=[\cdot]$ is trivial. Let $\mathcal{C}=\lambda x. \mathcal{C}'$  and let  $\mathscr{D}$ be such that $\lambda x. \mathcal{C}'[M] \Downarrow  \mathscr{D}$. By Proposition~\ref{prop: the semantics is invariant under reduction}.\ref{enum: invariance abs} there exists $\mathscr{D}'$ such that $\mathcal{C}'[M]\Downarrow  \mathscr{D}'$ and $\mathscr{D}\leq_\dist \lambda x. \mathscr{D}'$. By induction  hypothesis, there exists $\mathscr{E}'$ such that $\mathcal{C}'[N]\Downarrow \mathscr{E}'$ and $\mathscr{D}'\leq_\dist \mathscr{E}'$.  We define $\mathscr{E}\triangleq  \lambda x. \mathscr{E}'$, so that  $\lambda x. \mathcal{C}'[N]\Downarrow \mathscr{E}$ and  $\mathscr{D}\leq_\dist  \lambda x. \mathscr{D}'  \leq_\dist \lambda x. \mathscr{E}' =  \mathscr{E}$. \\
 We consider the case  $\mathcal{C}= \mathcal{C}'L$ (the case $\mathcal{C}=L\mathcal{C}'$ is similar). Let $\mathscr{D}$ be such that $\mathcal{C}'[M]L \Downarrow \mathscr{D}$.   By Proposition~\ref{prop: the semantics is invariant under reduction}.\ref{lem: invariance beta general case}, there exist head  distributions $\mathscr{D}'$ and $\lbrace  \mathscr{D}_{H, L} \rbrace_{\lambda x. H\,  \in\, \mathrm{ supp} (\mathscr{D}')}$ such that $\mathcal{C}'[M]\Downarrow \mathscr{D}'$,    $\lbrace H[L/x]\Downarrow \mathscr{D}_{H, L}\rbrace_{\lambda x. H\,  \in\, \mathrm{ supp} (\mathscr{D}')}$, and:
\begin{equation*}
\mathscr{D}\leq_\dist \sum_{\lambda x. H\,  \in\, \mathrm{ supp}(\mathscr{D}')} \mathscr{D}'(\lambda x. H)\cdot \mathscr{D}_{H, L}+ \sum_{\substack{H\, \in \, \mathrm{ supp}(\mathscr{D}')\\ \ \ \ \ \  \cap \,\neu}} \mathscr{D}'(H)\cdot H L \ .
\end{equation*}
 By induction hypothesis,  there exists a head distribution $\mathscr{E}'$  such that $\mathcal{C}'[N]\Downarrow \mathscr{E}'$ and $\mathscr{D}' \leq_\dist \mathscr{E}'$. For all $\lambda x. H \,\in\, \mathrm{ supp} (\mathscr{E}')$, we set:
\allowdisplaybreaks 
\begin{align*}
&\mathscr{E}_{H, L}\triangleq  \begin{cases} \mathscr{D}_{H, L} &\text{if } \lambda x. H \in \mathrm{ supp} (\mathscr{D}') \\
\bot &\text{otherwise},
\end{cases}\\
&\mathscr{E}\triangleq  \sum_{\lambda x. H\,  \in\,  \mathrm{ supp}(\mathscr{E}')} \mathscr{E}'(\lambda x. H)\cdot \mathscr{E}_{H, L}+ \sum_{\substack{H\, \in\, \mathrm{ supp}(\mathscr{E}')\\ \ \ \ \ \ \cap \,\neu}} \mathscr{E}'(H)\cdot HL .
\end{align*}
Therefore, $\mathcal{C}'[N]L\Downarrow \mathscr{E}$ and $\mathscr{D}\leq_{\dist} \mathscr{E}$. \\
 We now consider the case $\mathcal{C}= \mathcal{C}' \oplus L$ (the case $\mathcal{C}= L \oplus  \mathcal{C}'$ is symmetric). Let  $\mathscr{D}$ be such that $\mathcal{C}'[M]\oplus L \Downarrow \mathscr{D}$. By Proposition~\ref{prop: the semantics is invariant under reduction}.\ref{enum: invariance sum}, there exist $\mathscr{D}'$ and $\mathscr{D}''$ such that  $\mathcal{C}'[M]\Downarrow \mathscr{D}'$, $L \Downarrow \mathscr{D}''$ and $\mathscr{D}\leq_\dist \frac{1}{2}\cdot \mathscr{D}'+ \frac{1}{2}\cdot \mathscr{D}''$.  By induction hypothesis, there exists $\mathscr{E}'$ such that $\mathcal{C}'[N]\Downarrow \mathscr{E}'$ and $\mathscr{D}' \leq_\dist \mathscr{E}'$.  We define $\mathscr{E}\triangleq \frac{1}{2}\cdot \mathscr{E}'+ \frac{1}{2}\cdot \mathscr{D}''$, so that $\mathcal{C}'[N]\oplus L \Downarrow \mathscr{E}$ and $\mathscr{D}\leq_\dist \frac{1}{2}\cdot \mathscr{D}'+ \frac{1}{2}\cdot \mathscr{D}''\leq_\dist \frac{1}{2}\cdot \mathscr{E}'+ \frac{1}{2}\cdot \mathscr{D}''= \mathscr{E}$.
\end{proof}

\section{The head spine reduction is equivalent to the head reduction}\label{appequiv}
\paragraph{Equivalence in a term-based setting.}
In Section~\ref{sec2} we endow the probabilistic $\lambda$-calculus with  the big-step operational semantics  $\inter{\cdot}$ introduced \textit{via} the head spine  reduction. This semantics   is often  called  \enquote{distribution-based} (see~\cite{borgstrom2016lambda}), since it involves a relation between terms and distributions, and it is opposed to the so-called  \enquote{term-based} semantics (see~\cite{di2005probabilistic}), which considers relations between terms weighted with probabilities.    In what follows,   we show that  the head and  head spine reductions have the same observational behaviour.   First, we prove this property in a \enquote{term-based} setting  (Theorem~\ref{thm: head  equal to heas spine}), in which we shall give an even stronger result:  the probability that a term converges to a given head normal form in $n$ steps is the same for both reduction strategies.  Then we  prove that  the reduction relation corresponding to the head spine evaluation  generates exactly the  distribution-based  semantics $\inter{\cdot}$   (Theorem~\ref{thm: equivalence of head reduction and head spine reduction}).  

To begin with, we define probabilistic transition relations, that is to say, relations weighted with probabilities.
\begin{defn}[Probabilistic transition relations] \label{defn: probabailistic transition relation} A \textit{probabilistic transition relation over a set $X$} is a relation $\mathcal{R} \subseteq X \times [0, 1]\times X $ such that, for all $x \in X$:
\begin{equation*}
 \sum_{\substack{p,\,  y \text{ s.t.}\\(x,p, y)\in\,  \mathcal{R}}} p\leq 1 .
\end{equation*}
If $\mathcal{R} \subseteq X \times [0, 1]\times X$ is a relation, we shall write  $x \ \mathcal{R}_p \ y$ in place of  $(x,p, y)\in \mathcal{R}$. \\
Given  $\mathcal{R}$ a probabilistic transition relation over $X$, we define the  relation $\mathcal{R}^n\subseteq X \times [0, 1]\times X $ by induction on $n\in \mathbb{N}$:
\begin{align*}
x\ \mathcal{R}^0 _p  \ y &\Leftrightarrow  x=y \, \wedge \, p=1\\
x\ \mathcal{R}^{n+1} _p  \ y &\Leftrightarrow   \exists y'\,\exists p', p''\,  (x \  \mathcal{R}^{n}_{p'}\ y' \, \wedge \, y'\  \mathcal{R}_{p''} \  y \, \wedge\,  p=p' p'') .  
\end{align*}
\end{defn}
\begin{prop} \label{prop: prob trans} Let   $\mathcal{R}$ be a probabilistic transition relation over $X$. For all $n \in \mathbb{N}$, $\mathcal{R}^n$ is a probabilistic transition relation.
\end{prop}
\begin{proof}
By induction on $n\in \mathbb{N}$. The case $n=0$ is trivial, so let us consider $n>0$. By using the induction hypothesis, we have:
\begin{equation*}
\begin{split}
 \sum_{\substack{p,\, y \text{ s.t.}\\ x\   \mathcal{R}^n_p \ y}} p &=  \sum_{\substack{p,\, y \text{ s.t.}\\
  \exists y'\, \exists p', p''\\  (x   \ \mathcal{R}^{n-1}_{p'}  y'\\  \wedge  \,  y' \  \mathcal{R}_{p''} \, y   \\ \wedge\  p\, =\, p' p'')}} p \leq  \sum_{\substack{p',\,p'',\, y, \, y'  \text{ s.t.}\\    x   \ \mathcal{R}^{n-1}_{p'}  y'  \\  \wedge\ y' \  \mathcal{R}_{p''} \, y  }} p '  p'' \leq   \sum_{\substack{p', \,y' \text{ s.t.}\\ x\   \mathcal{R}^{n-1}_{p'} \, y'}} p'\leq 1 .
 \end{split}
\end{equation*}
\end{proof}
Both  the head and  head spine reduction strategies can be introduced as probabilistic transition relations.
\begin{defn} [Head and head spine reductions] A  \textit{head context} is a context of the form $\lambda x_1\ldots x_n.[\cdot]L_1\ldots L_m$, also written $\contrat{[\cdot]}$, where $n, m \geq 0$ and $L_i \in \Lambda_\oplus$. Head contexts are ranged over by $\mathcal{E}$. \\
The probabilistic transition relations   $\rightarrow$ (\textit{head reduction}) and $\dashrightarrow$ (\textit{head spine reduction})  over $\Lambda_\oplus$ are defined as follows:
\allowdisplaybreaks
\begin{align*}
M \rightarrow_p N&\triangleq  \begin{cases}M=\cont{(\lambda y. P)Q}, \, N=\cont{P[Q/y]},  \, p=1,\\
\hspace{3cm }\text{or}\\
M=\cont{P_1\oplus P_2}, \, P_1 \not = P_2,\, N=\cont{P_i}, \, p=\frac{1}{2},\\
\hspace{3cm }\text{or}\\
M=\cont{P\oplus P},\, N=\cont{P}, \, p=1.
 \end{cases}\\ \\
M \dashrightarrow_p N&\triangleq \begin{cases} M=\cont{(\lambda y. H)Q}, \,   N= \cont{H[Q/y]},\, p=1,\\ 
\hspace{3cm }\text{or}\\
M=\cont{(\lambda y. P)Q}, \,  P \dashrightarrow_p P', \, N= \cont{(\lambda y. P')Q},\\
\hspace{3cm }\text{or}\\
M=\cont{P_1\oplus P_2}, \, P_1\not =P_2, \, N= \cont{P_i}, \, p=\frac{1}{2},\\
\hspace{3cm }\text{or}\\
M=\cont{P\oplus P},\, N=\cont{P}, \, p=1.
 \end{cases} 
 \end{align*}
\end{defn}
For all $n \in \mathbb{N}$,  the relations $\rightarrow^n$ and $\dashrightarrow^n$  can be constructed  using Definition~\ref{defn: probabailistic transition relation}, and they are  probabilistic transition relations by Proposition~\ref{prop: prob trans}. 

Let us state some remarkable properties concerning both the head and  head spine reductions:
\begin{lem}[Reduction properties]\label{lem: reduction properties head spine} Let $M, N, L \in \Lambda_\oplus$.  The following statements hold:
\begin{enumerate}[(1)]
\item \emph{Application}: If $M \dashrightarrow_p N$ then $ML \dashrightarrow_p NL$.\label{enum: application head spine}
\item \emph{Substitution}:  If $M \rightarrow_p N$ then $M[L/x]\rightarrow_p N[L/x]$.  \label{enum: substitution head spine}
\item \emph{Abstraction}: If $M\  \mathcal{R}_p\ N$ then $\lambda x. M\  \mathcal{R}_p \ \lambda x. N$, where $\mathcal{R}\in \lbrace \rightarrow, \dashrightarrow \rbrace$. \label{enum: abstraction head spine}
\end{enumerate}
\end{lem}
\begin{proof}
Straightforward.
\end{proof}
Observe  that the application property does not hold for the head reduction. For example, $\lambda x.\mathbf I \mathbf I\rightarrow_p \lambda x.\mathbf I$, but $(\lambda x. \mathbf I\mathbf I)\mathbf I\rightarrow_p\mathbf I \mathbf I \not = (\lambda x.\mathbf I)\mathbf I$. Also,  the substitution property does not hold for the head spine reduction. For example, if $M\triangleq (\lambda x. y)\mathbf I$ then $M \dashrightarrow_p y$ but $M[\mathbf \Omega/y]\dashrightarrow_p M[\mathbf \Omega/y]\not = y[\mathbf \Omega/y]$. 

The following definition introduces the probability of convergence  for both reduction strategies.
\begin{defn}[$\mathcal{H}^\infty$ and $\mathcal{S}^\infty$] \label{defn: Hn and Sn} Let $M\in \Lambda_\oplus$, $H \in \h$ and $n \in \mathbb{N}$.  We define the probability $\mathcal{H}^n(M, H)$ (resp.~$\mathcal{S}^n(M, H)$) that $M$ converges to  $H$  in exactly $n$ steps of  head reduction (resp.~of head spine reduction) as follows:
\allowdisplaybreaks
\begin{align*}
&\mathcal{H}^n(M, H)\triangleq \sum_{\substack{(M_0, \ldots, M_n)\text{ s.t. } M_0=M,  \\ M_n=H, \, \forall i<n\, M_i \rightarrow_{p_{i+1}}M_{i+1} }} \prod^n_{i=1}p_i\\
&\mathcal{S}^n(M, H)\triangleq  \sum_{\substack{(M_0, \ldots, M_n)\text{ s.t. } M_0=M,  \\ M_n=H, \, \forall i<n\, M_i \dashrightarrow_{p_{i+1}}M_{i+1} }}  \prod^n_{i=1}p_i  \, .
\end{align*}
The probability $\mathcal{H}^\infty(M, H)$ (resp.~$\mathcal{S}^\infty(M, H)$) that $M$ converges to   $H$ in an arbitrary number of steps of  head reduction (resp.~of head spine reduction) is defined as follows:
\allowdisplaybreaks
\begin{align*}
&\mathcal{H}^\infty(M, H)\triangleq \sum_{n=0}^\infty \mathcal{H}^n(M, H) \quad \mathcal{S}^\infty(M, H)\triangleq \sum_{n=0}^\infty \mathcal{S}^n(M, H)  .
\end{align*}
\end{defn}
We now state and prove some basic properties about $\mathcal{H}^n$ and $\mathcal{S}^n$.
\begin{lem} \label{lem: properties of S infty for one direction} Let $M, N \in \plam$ and  $H \in \h$. The following statements hold:
\begin{enumerate}[(1)]
\item \label{enum: S infty sum} If either $\mathcal{X}= \mathcal{H}$ and $\mathcal{R}=\, \rightarrow$, or $\mathcal{X}= \mathcal{S}$ and $\mathcal{R}=\,  \dashrightarrow$, then:
\begin{itemize}
\item If $n=0$ and $M=H$ then $\mathcal{X}^n(M , H)=1$. 
\item If $n>0$ and $M\ \mathcal{R}_1 \ M'$ then $\mathcal{X}^n(M , H)=\mathcal{X}^{n-1}(M', H)$.
\item If $n>0$,  $M\  \mathcal{R}_{\frac{1}{2}}\, M'$,  and $M\ \mathcal{R}_{\frac{1}{2}}\, M''$, then $\mathcal{X}^n(M , H)= \frac{1}{2}\cdot \mathcal{X}^{n-1}(M', H)+ \frac{1}{2}\cdot \mathcal{X}^{n-1}(M'', H)$.
\end{itemize}
\item \label{enum: S infty abst}   For all $ n\in \mathbb{N}$:
\begin{equation*}
\begin{split}
  \mathcal{H}^n(\lambda x. M, \lambda x. H)&= \mathcal{H}^n(M, H)\\
   \mathcal{S}^n(\lambda x. M, \lambda x. H)&= \mathcal{S}^n(M, H).
   \end{split}
\end{equation*}
\item \label{enum: H infty app} For all $ n\in \mathbb{N}$:
\begin{multline*}
 \mathcal{H}^n (M[N/x], H)=\\
  \sum_{l+l'=n}  \sum_{H' \in \h} \mathcal{H}^{l}(M, H')\cdot \mathcal{H}^{l'}(H'[N/x], H) .
\end{multline*}
\item \label{enum: S infty app} For all $ n\in \mathbb{N}$:
\begin{multline*}
\mathcal{S}^n (MN, H)= \sum_{l+l'=n}\sum_{H' \in \h} \mathcal{S}^l(M, H')\cdot \mathcal{S}^{l'}(H'N, H).
\end{multline*}
\end{enumerate}
\end{lem}
\begin{proof}
Concerning point~\ref{enum: S infty sum}, we just prove the case where $n>0$, $M \rightarrow_{\frac{1}{2}}M'$ and $ M \rightarrow_{\frac{1}{2}}M'' $:
\allowdisplaybreaks
\begin{align*}
   \mathcal{H}^n (M, H) &=   \sum_{\substack{(M_0, \ldots, M_n)\text{ s.t. } M_0=M,  \\ M_n=H, \, \forall i<n\, M_i \rightarrow_{p_{i+1}}M_{i+1} }} \prod^n_{i=1}p_i \\
  &=  \frac{1}{2}\cdot \Bigg(  \sum_{\substack{(M_0, \ldots, M_{n-1})
  \text{ s.t. } M_0=M',  \\ M_{n-1}=H,\, \forall i<n-1\, M_i \rightarrow_{p_{i+1}}M_{i+1} }} \prod^{n-1}_{i=1}p_i \Bigg) \\
  &\phantom{=\ }+   \frac{1}{2}\cdot \Bigg(   \sum_{\substack{(M_0, \ldots, M_{n-1}) \text{ s.t. } M_0=M'', \\  \,M_{n-1}=H, \, \forall i<n-1\, M_i \rightarrow_{p_{i+1}}M_{i+1} }} \prod^{n-1}_{i=1}p_i \Bigg) \\
&  =\frac{1}{2}\cdot     \mathcal{H}^{n-1} (M',H) + \frac{1}{2}\cdot  \mathcal{H}^{n-1}  (M'', H)   .
\end{align*}
 Concerning point~\ref{enum: S infty abst}, for all $n \in \mathbb{N}$ we have:
\allowdisplaybreaks
\begin{align*}
  \mathcal{H}^n (M, H)& =   \sum_{\substack{(M_0, \ldots, M_n)\text{ s.t. } M_0=M,  \\ M_n=H, \, \forall i<n\, M_i \rightarrow_{p_{i+1}}M_{i+1} }} \prod^n_{i=1}p_i \\
  &=    \sum_{\substack{(\lambda x. M_0, \ldots, \lambda x. M_n)\text{ s.t. } \lambda x.M_0=\lambda x. M,  \\ \lambda x.M_n=\lambda x.H, \, \forall i<n\, \lambda x.M_i \rightarrow_{p_{i+1}}\lambda x.M_{i+1} }} \prod^n_{i=1}p_i \\
  &=      \mathcal{H}^n (\lambda x. M,\lambda x. H)   .
\end{align*}
We prove the equation $\mathcal{S}^n (M, H)=  \mathcal{S}^n (\lambda x. M,\lambda x. H) $ in a similar way. \\
Let us now prove point~\ref{enum: H infty app} by induction on $n \in \mathbb{N}$.   We have three cases:
\begin{enumerate}[(a)]
\item If $M$ is a head normal  form,  then $\mathcal{H}^{l}(M,H')\neq 0$ just when $l=0$ and $H'=M$. In all cases, the equation holds.
\item Suppose $M \rightarrow_\frac{1}{2} M_1$ and $M \rightarrow_\frac{1}{2}M_2$. If $n=0$ then the equation trivially holds. Otherwise, by Lemma~\ref{lem: reduction properties head spine}.\ref{enum: substitution head spine} we have  $M[N/x] \rightarrow_\frac{1}{2} M_1[N/x]$ and $M[N/x] \rightarrow_\frac{1}{2}M_2[N/x]$. Therefore, by using point~\ref{enum: S infty sum} and the induction hypothesis:
\allowdisplaybreaks
\begin{align*}
\ \ \ \ \ \ \ \ &\mathcal{H}^n (M[N/x], H)=\\
&= \frac{1}{2}\cdot \mathcal{H}^{n-1}(M_1[N/x], H)+ \frac{1}{2}\cdot \mathcal{H}^{n-1}(M_2[N/x], H)  \\
&=\frac{1}{2}  \sum_{l+l'=n-1} \sum_{H' \in \h}\mathcal{H}^l(M_1,H')\cdot \mathcal{H}^{l'}(H'[N/x],H ) \\
&\phantom{= \ }+ \frac{1}{2}  \sum_{l+l'=n-1} \sum_{H' \in \h}\mathcal{H}^l(M_2,H')\cdot \mathcal{H}^{l'}(H'[N/x],H ) \\
%&= \sum_{l+l'=n-1} \sum_{H' \in \h} \Big(\frac{1}{2}\cdot \mathcal{H}^l(M',H') +\frac{1}{2} \cdot  \mathcal{H}^l(M'',H') \Big) \cdot \mathcal{H}^{l'}(H'[N/x],H )\\
&=\sum_{l+l'=n-1} \sum_{H' \in \h} \mathcal{H}^{l+1}(M,H')  \cdot \mathcal{H}^{l'}(H'[N/x],H ) \\
&= \sum_{l+l'=n} \sum_{H' \in \h} \mathcal{H}^{l}(M,H')  \cdot \mathcal{H}^{l'}(H'[N/x],H ) .
\end{align*} 
\item If $M \rightarrow_1 M'$ then we proceed  similarly.
\end{enumerate}
Finally we prove point~\ref{enum: S infty app}  by induction  on $n \in \mathbb{N}$. We have three cases:
\begin{enumerate}[(a)]
\item If $M$ is a head normal form, then $\mathcal{S}^l (M, H')\neq 0$ whenever $l=0$ and $H'=M$. In all cases,  the equation holds.
\item Suppose $M \dashrightarrow_{\frac{1}{2}} M_1$ and $M \dashrightarrow_\frac{1}{2} M_2$. If $n=0$ then the equation trivially holds.
Otherwise, by  Lemma~\ref{lem: reduction properties head spine}.\ref{enum: application head spine} we have $MN \dashrightarrow_{\frac{1}{2}} M_1N$ and $MN \dashrightarrow_\frac{1}{2} M_2N$. Therefore,  by using point~\ref{enum: S infty sum} and the induction hypothesis:
\allowdisplaybreaks
\begin{align*}
&\mathcal{S}^n (MN, H)=\\
&= \frac{1}{2}\cdot \mathcal{S}^{n-1}(M_1N, H)+ \frac{1}{2}\cdot \mathcal{S}^{n-1}(M_2N, H)   \\
&=\frac{1}{2}  \sum_{l+l'=n-1} \sum_{H' \in \h}\mathcal{S}^l(M_1,H')\cdot \mathcal{S}^{l'}(H'N,H ) \\
&\phantom{= \ }+ \frac{1}{2}  \sum_{l+l'=n-1} \sum_{H' \in \h}\mathcal{S}^l(M_2,H')\cdot \mathcal{S}^{l'}(H'N,H )  \\
%&= \sum_{l+l'=n-1} \sum_{H' \in \h} (\frac{1}{2}\cdot \mathcal{S}^l(M_1,H') +\frac{1}{2}  \cdot \mathcal{S}^l(M_2,H') ) \cdot \mathcal{S}^{l'}(H'N,H )\\
&=\sum_{l+l'=n-1} \sum_{H' \in \h} \mathcal{S}^{l+1}(M,H')  \cdot \mathcal{S}^{l'}(H'N,H ) \\
&= \sum_{l+l'=n} \sum_{H' \in \h} \mathcal{S}^{l}(M,H')  \cdot \mathcal{S}^{l'}(H'N,H ) .
\end{align*} 
\item If $M \dashrightarrow_1 M'$, we proceed similarly.
\end{enumerate}
\end{proof} 
Before stating the main theorem, relating the head and  head spine reduction strategies, we need a further technical lemma.
\begin{lem} 	\label{lem: commutation diagram head spine} If $M \dashrightarrow_p M'$ then there exists $n_0\in \mathbb{N}$ and $M_0\in \Lambda_\oplus$ such that $M \rightarrow^{n_0+1}_p M_0$ and $M' \rightarrow_1^{n_0}M_0$. Diagrammatically:
\begin{center}
\begin{tikzcd}[ampersand replacement = \&]
M \arrow[rr, dashed,  "p"swap , pos=1]\arrow[rrdd,  "p"swap,  pos=1, "n_0+1" pos=0.5]\&  \& M' \arrow[dd, "1"  swap, pos=0.9 , "n_0" pos=0.5]\\ \\
\& \& M_0
\end{tikzcd}
\end{center}
\end{lem}
\begin{proof}
By induction on the structure of $M$. $M$ cannot be a head normal form, so that we have three cases:
\begin{enumerate}[(1)]
\item$M= \mathcal{E}[(\lambda y. H)Q]$, where $\mathcal{E}=\contrat{[\cdot]}$ and $H\in\h$. Then, $M'= \cont{H[Q/y]}$, and we set $n_0\triangleq 0$ and $M_0\triangleq M'$.
\item $M= \mathcal{E}[(\lambda y. P)Q]$, where $\mathcal{E}=\contrat{[\cdot]}$ and $P \dashrightarrow_p P'$.  Then, $M'= \cont{(\lambda y. P')Q}$.  By  applying Lemma~\ref{lem: reduction properties head spine}.\ref{enum: application head spine},  $P\vec{L}\dashrightarrow_p P'\vec{L}$. By induction hypothesis, there exists $n'_0$ and $P_0$ such that $P\vec{L} \rightarrow^{n'_0+1}_p P_0$ and $P'\vec{L}\rightarrow^{n'_0}_1P_0$. By repeatedly applying Lemma~\ref{lem: reduction properties head spine}.\ref{enum: substitution head spine}, we have that $P[Q/y]\vec{L} \rightarrow^{n'_0+1}_p   \! P_0[Q/y] $ and $P'[Q/y]\vec{L}\rightarrow^{n'_0}_1  P_0[Q/y]$, since $y$ is not free in $\vec{L}$.  Moreover, by repeatedly applying Lemma~\ref{lem: reduction properties head spine}.\ref{enum: abstraction head spine}, $\cont{P[Q/y]} \rightarrow^{n'_0+1}_p  \lambda \vec{x}.P_0[Q/y] $ and $\cont{P'[Q/y]}\rightarrow^{n'_0}_1  \lambda \vec{x}. P_0[Q/y]$. We set $n_0\triangleq n_0'+1$ and $M_0\triangleq \lambda \vec{x}.P_0[Q/y]$. On the one hand,    $\cont{(\lambda y.P)Q}\rightarrow_1 \cont{P[Q/y]} \rightarrow^{n'_0+1}_p \lambda \vec{x}.P_0[Q/y]$ and, on the other hand, $\cont{(\lambda y. P')Q}\rightarrow_1 \cont{P'[Q/y]}\rightarrow^{n'_0}_1\lambda \vec{x}.P_0[Q/y]$.
\item $M= \mathcal{E}[P_1\oplus P_2]$, where $\mathcal{E}=\contrat{[\cdot]}$. Then, $M'=\cont{P_i}$. We set $n_0\triangleq 0$ and $M_0\triangleq M'$. \qedhere
\end{enumerate}
\end{proof}
\begin{thm}[$\mathcal{H}^n= \mathcal{S}^n$]\label{thm: head  equal to heas spine} Let $M \in \Lambda_\oplus$ and  $H \in \h$. Then, for all $n \in \mathbb{N}$:
\begin{equation*}
\mathcal{S}^n(M, H)=\mathcal{H}^n(M, H) .
\end{equation*}
\end{thm}
\begin{proof}
By induction on $n$. If $n=0$ then $\mathcal{S}^0(M, H)=\mathcal{H}^0(M, H)$ by definition. Suppose $n >0$. If $M$ is a head normal form, then $\mathcal{S}^n(M, H)=0 =  \mathcal{H}^n(M, H)$. Otherwise, we can apply a head spine reduction step to $M$. If $M \dashrightarrow_1 M'$ then, by Lemma~\ref{lem: commutation diagram head spine}, there exist $n_0 $ and $M_0$ such that:
\begin{equation*}
M \rightarrow^{n_0+1}_1 M_0\qquad \text{and}\qquad M'\rightarrow_1^{n_0}M_0  . 
\end{equation*}
Moreover, by induction hypothesis and by Lemma~\ref{lem: properties of S infty for one direction}.\ref{enum: S infty sum} we have $\mathcal{S}^n(M, H)= \mathcal{S}^{n-1}(M', H)=\mathcal{H}^{n-1}(M', H)$. If $n_0\leq n-1$ then $\mathcal{H}^{n-1}(M', H)= \mathcal{H}^{n-1-n_0}(M_0, H)= \mathcal{H}^n(M, H)$. Otherwise, $n-1<n_0$ and  $\mathcal{H}^{n-1}(M', H)=0= \mathcal{H}^{n}(M, H)$.\\
If $M \dashrightarrow_\frac{1}{2} M'$ and $M \dashrightarrow _\frac{1}{2}M''$ then, by  Lemma~\ref{lem: commutation diagram head spine}, there exist $n'_0, n''_0$ and $M'_0, M''_0$ such that:
\begin{equation*}
\begin{array}{ll}
M \rightarrow^{n'_0+1}_\frac{1}{2} M'_0,  &M'\rightarrow_1^{n'_0}M'_0\\ \\ 
  M \rightarrow^{n''_0+1}_\frac{1}{2} M''_0, & M''\rightarrow_1^{n''_0}M''_0  . 
\end{array}
\end{equation*}
Then, there exist $N$, $ N'$ and $N''$ such that:
\begin{equation*}
M \rightarrow^{t}_1 N \qquad N \rightarrow_\frac{1}{2} N' \rightarrow^{t'}_1 M'_0 \qquad N \rightarrow_\frac{1}{2} N'' \rightarrow^{t''}_1 M''_0 ,
\end{equation*}
where $n'_0= t+t'$ and  $n''_0= t+t''$. By induction hypothesis and by Lemma~\ref{lem: properties of S infty for one direction}.\ref{enum: S infty sum}: 
\begin{equation*}
\begin{split}
\mathcal{S}^n(M, H)&=\frac{1}{2}\cdot \mathcal{S}^{n-1}(M', H)+ \frac{1}{2}\cdot \mathcal{S}^{n-1}(M'', H)\\
&= \frac{1}{2}\cdot \mathcal{H}^{n-1}(M', H)+ \frac{1}{2}\cdot \mathcal{H}^{n-1}(M'', H)
   .
 \end{split}
\end{equation*}  
We have four cases:
\begin{enumerate}[(a)]
\item If $n'_0 , n''_0\leq n-1$ then, by using Lemma~\ref{lem: properties of S infty for one direction}.\ref{enum: S infty sum}:
\begin{equation*}
\begin{split}
&\mathcal{H}^n(M, H)= \\
&=\mathcal{H}^{n-t}(N, H)\\
&= \frac{1}{2}\cdot \mathcal{H}^{n-(t+1)}(N', H)+ \frac{1}{2}\cdot \mathcal{H}^{n-(t+1)}(N'', H)  \\
&=  \frac{1}{2}\cdot \mathcal{H}^{n-(n'_0+1)}(M'_0, H)+ \frac{1}{2}\cdot \mathcal{H}^{n-(n''_0+1)}(M''_0, H) \\
&=  \frac{1}{2}\cdot \mathcal{H}^{n-1}(M', H)+ \frac{1}{2}\cdot \mathcal{H}^{n-1}(M'', H)  .
\end{split}
\end{equation*}
\item If $n'_0\leq n-1$ and $n-1<n''_0$    then, by using Lemma~\ref{lem: properties of S infty for one direction}.\ref{enum: S infty sum}:
\begin{equation*}
\begin{split}
&=\mathcal{H}^n(M, H)=\\
&=\mathcal{H}^{n-t}(N, H)\\
&= \frac{1}{2}\cdot \mathcal{H}^{n-(t+1)}(N', H)+ \frac{1}{2}\cdot \mathcal{H}^{n-(t+1)}(N'', H)  \\
&= \frac{1}{2}\cdot \mathcal{H}^{n-(n'_0+1)}(M'_0, H)=  \frac{1}{2}\cdot \mathcal{H}^{n-1}(M', H)\\
&=  \frac{1}{2}\cdot \mathcal{H}^{n-1}(M', H)+ \frac{1}{2}\cdot \mathcal{H}^{n-1}(M'', H)  .
\end{split}
\end{equation*}
\item The case  where $n-1<n'_0$ and   $n''_0\leq n-1$ is similar to the previous one.
\item If $n-1<n'_0, n''_0$ then $\mathcal{H}^n(M, H)=0= \frac{1}{2}\cdot \mathcal{H}^{n-1}(M', H)+ \frac{1}{2}\cdot \mathcal{H}^{n-1}(M'', H)$.
\end{enumerate}
\end{proof}

\paragraph{The term-based and the distribution-based  semantics coincide.} What we have established so far is an equivalence between the head and  head spine reductions in a \enquote{term-based} operational semantics introduced through the notion of probabilistic transition relation. We are going to show  that the term-based and the distribution-based semantics for the head spine reduction coincide. This allows us to show that the big-step  semantics introduced in~\eqref{eq: big-step semantics} is invariant with respect to the usual head reduction steps $(\lambda x. M)N\rightarrow M[N/x]$, where $M$ is not necessarily a head normal form.

%To begin with, we generalize Definition~\ref{defn: Hn and Sn} by introducing  a notation for the probability of convergence in an \textit{arbitrary} number of steps.
%\begin{defn}[$\mathcal{H}^\infty$ and $\mathcal{S}^\infty$]  Let $M, N \in \Lambda_\oplus$ and let $H \in \h$.  
%We define the probability $\mathcal{H}^\infty(M, H)$ (resp.~$\mathcal{S}^\infty(M, N)$) that $M$ converges to  the head normal form $H$ in an arbitrary number of steps of  head reduction (resp.~of head spine reduction) as follows:
%\allowdisplaybreaks
%\begin{align*}
%&\mathcal{H}^\infty(M, H)\triangleq \sum_{n=0}^\infty \mathcal{H}^n(M, H)&  & & &\mathcal{S}^\infty(M, H)\triangleq \sum_{n=0}^\infty \mathcal{S}^n(M, H) \enspace .
%\end{align*}
%The probability that $M$ converges to an arbitrary head normal form in an arbitrary number of steps of  head reduction (resp.~of head spine reduction) is defined as follows:
%\allowdisplaybreaks
%\begin{align*}
%&\mathcal{H}^\infty(M, X)\triangleq \sum_{H \in X} \mathcal{H}^\infty(M, H)&  & & &\mathcal{S}^\infty (M, X)\triangleq \sum_{H \in X} \mathcal{S}^\infty (M, H) \enspace .
%\end{align*}
%When $X= \h$ we simply write  $\mathcal{H}^\infty(M)$ (resp.~$\mathcal{S}^\infty(M)$) in place of $\mathcal{H}^\infty(M, X)$ (resp.~$\mathcal{S}^\infty(M, X)$).
%\end{defn}
\begin{lem} \label{lem: equivalence part 1}Let $M \in \plam$. For all $H \in \h$, $\inter{M}(H) \leq \mathcal{S}^\infty(M, H)$.
\end{lem}
\begin{proof}
We show that, for all $\mathscr{D}$ such that $M \Downarrow \mathscr{D}$ and for all $H \in \h$, it holds that $\mathscr{D}(H)\leq \mathcal{S}^\infty(M, H)$. The proof is by induction on the derivation of $M \Downarrow \mathscr{D}$ by considering the structure of $M$.  Since the case $\mathscr{D}= \bot$ is trivial, we shall assume that the last rule of $M \Downarrow \mathscr{D}$ is not $s1$. \\
If $M=x$ then $\mathscr{D}=x$ and the last rule of $M \Downarrow \mathscr{D}$ is $s2$. If  $H \neq x$ then $\mathscr{D}(H)=0$.  Otherwise,  $\mathscr{D}(x)= 1 = \mathcal{S}^\infty(x, x)$.\\
If $M= \lambda x. M'$ then $\mathscr{D}= \lambda x. \mathscr{D}'$ and the last rule of $M \Downarrow \mathscr{D}$ is the following:
\begin{prooftree}
\AxiomC{$M' \Downarrow \mathscr{D}'$}
\RightLabel{$s3$}
\UnaryInfC{$\lambda x.M' \Downarrow \lambda x. \mathscr{D}'$}
\end{prooftree}
 If  $H\in \neu$ then $\mathscr{D}(H)=0$. Otherwise,   $H= \lambda x. H'$ and, by using the induction hypothesis and Lemma~\ref{lem: properties of S infty for one direction}.\ref{enum: S infty abst}, we have:
\begin{equation*}
\begin{split}
\mathscr{D}(H)&= \lambda x. \mathscr{D}'(\lambda x. H')= \mathscr{D}'(H')\leq \mathcal{S}^\infty (M', H')\\
&= \mathcal{S}^\infty (\lambda x. M', \lambda x. H') .
\end{split}
\end{equation*}
If $M= PQ$ then $\mathscr{D}= \sum_{\lambda x. H'\,  \in \, \mathrm{ supp}(\mathscr{E})} \mathscr{E}(\lambda x. H')\cdot \mathscr{F}_{H', Q} + \sum_{  H' \, \in \,\mathrm{ supp}(\mathscr{E})\, \cap\,  \neu} \mathscr{E}(H')\cdot H' Q$, and the last rule of $M \Downarrow \mathscr{D}$ is $s4$ with premises $P \Downarrow \mathscr{E}$ and $\lbrace H'[Q/x]\Downarrow \mathscr{F}_{H', Q} \rbrace_{\lambda x. H'\, \in\,  \mathrm{ supp}(\mathscr{E})}$.
By  induction hypothesis, Lemma~\ref{lem: properties of S infty for one direction}.\ref{enum: S infty sum} and Lemma~\ref{lem: properties of S infty for one direction}.\ref{enum: S infty app}, we have:
\allowdisplaybreaks
\begin{align*}
 \mathscr{D}(H)&=\sum_{\lambda x. H' \,  \in \, \mathrm{ supp}(\mathscr{E})} \mathscr{E}(\lambda x. H')\cdot \mathscr{F}_{H', Q}(H) \\
 &\phantom{= \ }+ \sum_{  H' \, \in \,\mathrm{ supp}(\mathscr{E})\,\cap\, \neu} \mathscr{E}(H')\cdot H' Q(H)\\
&\leq \sum_{\lambda x. H' \,  \in \, \h}  \mathcal{S}^\infty(P , \lambda x. H')\cdot \mathcal{S}^\infty(H'[Q/x], H)\\
&\phantom{= \ }+\sum_{  H' \, \in \, \neu}  \mathcal{S}^\infty(P, H')\cdot  \mathcal{S}^\infty(H'Q, H')\\
&= \sum_{\lambda x. H' \,  \in \, \h}  \mathcal{S}^\infty(P , \lambda x. H')\cdot \mathcal{S}^\infty((\lambda x. H')Q, H)\\
&\phantom{= \ } +\sum_{  H' \, \in \, \neu}  \mathcal{S}^\infty(P, H')\cdot  \mathcal{S}^\infty(H'Q, H')\\
&= \sum_{ H' \,  \in \, \h}  \mathcal{S}^\infty(P ,H')\cdot \mathcal{S}^\infty(H'Q, H)=   \mathcal{S}^\infty(PQ ,H)  . 
\end{align*}
  If $M= P \oplus Q$ then $\mathscr{D}= \frac{1}{2}\cdot \mathscr{D}_1 + \frac{1}{2}\cdot \mathscr{D}_2$ and the last rule of $M \Downarrow \mathscr{D}$ is as follows:
\begin{prooftree}
\AxiomC{$P \Downarrow \mathscr{D}_1$}
\AxiomC{$Q\Downarrow \mathscr{D}_2$}
\RightLabel{$s5$}
\BinaryInfC{$P \oplus Q\Downarrow \frac{1}{2}\cdot \mathscr{D}_1+ \frac{1}{2}\cdot \mathscr{D}_2$}
\end{prooftree}
By using the induction hypothesis and by Lemma~\ref{lem: properties of S infty for one direction}.\ref{enum: S infty sum}, we have:
\begin{equation*}
\begin{split}
\mathscr{D}(H)&=\frac{1}{2}\cdot \mathscr{D}_1(H)+ \frac{1}{2}\cdot \mathscr{D}_2(H) \\
&\leq \frac{1}{2}\cdot \mathcal{S}^\infty (P, H)+ \frac{1}{2}\cdot \mathcal{S}^\infty(Q, H)= \mathcal{S}^\infty(P \oplus Q, H) .
\end{split}
\end{equation*}
\end{proof}
\begin{lem} \label{lem: to prove lemma about head contexts } Let $M \in \plam$. Then:
\begin{enumerate}[(1)]
\item \label{enum: lemma about head context 1} If $M \dashrightarrow_1 M'$ and $M'\Downarrow \mathscr{D}$, then $M\Downarrow \mathscr{D}$.
\item \label{enum: lemma about head context 2} If $M \dashrightarrow_\frac{1}{2} M_1$,  $M \dashrightarrow_\frac{1}{2} M_2$, $M_1 \Downarrow \mathscr{D}_1$ and $M_2\Downarrow \mathscr{D}_2$, then there exists $\mathscr{D}$ such that  $ \frac{1}{2}\cdot\mathscr{D}_1 + \frac{1}{2}\cdot \mathscr{D}_2\leq_\dist \mathscr{D}$ and $M\Downarrow \mathscr{D}$.
\end{enumerate}
\end{lem}
\begin{proof}
We prove both points simultaneously by induction on the structure of $M$. If $M$ is not a head normal form, then there exists a head context $\mathcal{E}$ such that $M=\mathcal{E}[P]$ and, either $P\dashrightarrow_1 P'$, or both $P \dashrightarrow_\frac{1}{2}P_1$ and $P \dashrightarrow_\frac{1}{2}P_2$. By looking at the structure of $M$ we have several cases. \\
$\bullet$ If  $\mathcal{E}=[\cdot]$ then we have three subcases:
\begin{enumerate}[(a)]
\item  If $M=(\lambda x.H)N$, then it must be that $M \dashrightarrow_1 M'=H[N/x]$. From $M' \Downarrow \mathscr{D}$ we can construct:
\begin{prooftree}
\AxiomC{}
\RightLabel{$s2$}
\UnaryInfC{$\lambda x. H \Downarrow \lambda x. H$}
\AxiomC{$H[N/x]\Downarrow \mathscr{D}$}
\RightLabel{$s4$}
\BinaryInfC{$(\lambda x. H)N\Downarrow \mathscr{D}$}
\end{prooftree}
\item\label{enum: case distribution for one side equivalence} Suppose $M= (\lambda x. Q)N$ with $Q \not \in \h$. We consider the case  $Q \dashrightarrow_\frac{1}{2} Q_1$ and $Q \dashrightarrow_\frac{1}{2} Q_2$.  W.l.o.g.~we  assume  that, for $i\in \lbrace 1,2\rbrace$,  $\mathscr{D}_i= \sum_{\lambda x. H\,  \in \, \mathrm{ supp}(\lambda x. \mathscr{E}_i)} (\lambda x.\mathscr{E}_i)(\lambda x. H)\cdot \mathscr{F}^i_{H, N} $, and the last rule of the  derivation of $(\lambda x. Q_i)N\Downarrow \mathscr{D}_i$ is $s4$ with premises $\lambda x. Q_i \Downarrow \lambda x. \mathscr{E}_i$ and $\lbrace H[N/x]\Downarrow \mathscr{F}^i_{H, N} \rbrace_{\lambda x. H\, \in\,  \mathrm{ supp}(\lambda x. \mathscr{E}_i)}$. Moreover, we can assume that the last rule of $\lambda x. Q_i \Downarrow \lambda x. \mathscr{E}_i$  is $s3$ with premise $Q_i  \Downarrow \mathscr{E}_i$.
By applying the induction hypothesis, there exists $\mathscr{E}$ such that  $ Q \Downarrow \mathscr{E}$ and  $ \frac{1}{2}\cdot \mathscr{E}_1+ \frac{1}{2}\cdot \mathscr{E}_2\leq_\dist \mathscr{E}$.   Since $\lbrace \mathscr{F} \in \dist(\h) \ \vert \ H[N/x] \Downarrow \mathscr{F} \rbrace$ is a directed set by Lemma~\ref{lem: direct set},   for all  $ H \in  \mathrm{ supp}(\mathscr{E}_1)\cap    \mathrm{ supp}(\mathscr{E}_2)$  there exists  $ \mathscr{G}_{H, N}$ such that $H[N/x]\Downarrow \mathscr{G}_{H, N}$ and  $ \mathscr{F}^1_{H, N},  \mathscr{F}^2_{H, N}\leq_\dist  \mathscr{G}_{H, N} $. We define:
\begin{equation*}
\ \ \  \mathscr{F}_{H, N}\triangleq \begin{cases}
\mathscr{F}^i_{H, N}&\text{if }H \in  \mathrm{ supp}(\mathscr{E}_i) \text{ and}\\
& H \not \in  \mathrm{ supp}(\mathscr{E}_{3-i}), \text{for }i \in \lbrace 1,2 \rbrace,\\
 \mathscr{G}_{H, N}&\text{if } H \in  \mathrm{ supp}(\mathscr{E}_1)\cap    \mathrm{ supp}(\mathscr{E}_2),\\
 \bot &\text{otherwise} .
\end{cases}
\end{equation*}
For all $H \in  \mathrm{ supp}(\mathscr{E})$, we have $H[N/x]\Downarrow \mathscr{F}_{H, N}$. Moreover, for all $i \in \lbrace 1,2\rbrace$ and $H \in  \mathrm{ supp}(\mathscr{E}_i)$, $\mathscr{F}^i_{H, N}\leq_\dist \mathscr{F}_{H, N}$. We define $\mathscr{D}\triangleq  \sum_{\lambda x. H\,  \in \, \mathrm{ supp}(\lambda x. \mathscr{E})} (\lambda x. \mathscr{E})(\lambda x. H)\cdot \mathscr{F}_{H,N}$, so that $(\lambda x. Q)N \Downarrow  \mathscr{D}$. Then:
\allowdisplaybreaks
\begin{align*}
&\frac{1}{2}\cdot  \mathscr{D}_1  + \frac{1}{2}\cdot \mathscr{D}_2=\\
&= \frac{1}{2}   \sum_{\lambda x.H\,  \in \, \mathrm{ supp}(\lambda x. \mathscr{E}_1)} (\lambda x. \mathscr{E}_1)(\lambda x. H)\cdot \mathscr{F}^1_{H, N} \\
&\phantom{= \ } +\frac{1}{2}   \sum_{\lambda x.H\,  \in \, \mathrm{ supp}(\lambda x. \mathscr{E}_2)} (\lambda x. \mathscr{E}_2)(\lambda x. H)\cdot \mathscr{F}^2_{H, N}\\
&= \frac{1}{2}   \sum_{ H\,  \in \, \mathrm{ supp}( \mathscr{E}_1)} \mathscr{E}_1(H)\cdot \mathscr{F}^1_{H, N} \\
&\phantom{= \ }+\frac{1}{2}   \sum_{ H\,  \in \, \mathrm{ supp}( \mathscr{E}_2)}  \mathscr{E}_2(H)\cdot \mathscr{F}^2_{H, N}\\
&\leq_\dist  \frac{1}{2}   \sum_{H\,  \in \, \mathrm{ supp}(\mathscr{E}_1)} \mathscr{E}_1( H)\cdot \mathscr{F}_{H, N} \\
& \phantom{= \ } +\frac{1}{2}   \sum_{ H\,  \in \, \mathrm{ supp}( \mathscr{E}_2)} \mathscr{E}_2(H)\cdot \mathscr{F}_{H, N}\\
&=   \sum_{H\,  \in \, \mathrm{ supp}(\mathscr{E})} \bigg( \frac{1}{2}\cdot  \mathscr{E}_1+ \frac{1}{2}\cdot \mathscr{E}_2 \bigg)(H)\cdot \mathscr{F}_{H, N} \\
& \leq_\dist  \sum_{H\,  \in \, \mathrm{ supp}(\mathscr{E})} \mathscr{E} (H)\cdot \mathscr{F}_{H, N} \\
&=  \sum_{\lambda x.H\,  \in \, \mathrm{ supp}(\lambda x.\mathscr{E})} (\lambda x. \mathscr{E}) (\lambda x.H)\cdot \mathscr{F}_{H, N}= \mathscr{D}  .
\end{align*}
\item Suppose $ M=P_1\oplus P_2$ then it must be that $M \dashrightarrow_\frac{1}{2} M_1=P_1$ and  $M \dashrightarrow_\frac{1}{2} M_2=P_2$, with $M_1 \Downarrow \mathscr{D}_1$ and $M_2 \Downarrow \mathscr{D}_2$. Then, it suffices to define $\mathscr{D}\triangleq  \frac{1}{2}\cdot \mathscr{D}_1+\frac{1}{2}\cdot \mathscr{D}_2$.
\end{enumerate}
$\bullet$ Suppose $\mathcal{E}= \lambda x. \mathcal{E}'$ and let us consider the case     $P \dashrightarrow_\frac{1}{2}P_1$ and $P \dashrightarrow_\frac{1}{2}P_2$. Then, for $i\in \lbrace 1,2\rbrace$,   the last rule in the derivation of $\mathcal{E}[P_i]\Downarrow \mathscr{D}_i$ is as follows:
\begin{prooftree}
\AxiomC{$\mathcal{E}'[P_i]\Downarrow \mathscr{D}'_i$}
\RightLabel{$s3$}
\UnaryInfC{$\lambda x. \mathcal{E}'[P_i] \Downarrow \lambda x. \mathscr{D}'_i$}
\end{prooftree}
By applying the induction hypothesis,  there exists $\mathscr{D'}$ such that $\mathcal{E}'[P]\Downarrow \mathscr{D'}$ and $ \frac{1}{2}\cdot \mathscr{D}'_1+ \frac{1}{2}\cdot \mathscr{D}'_2\leq_\dist \mathscr{D}'$. We define $\mathscr{D}\triangleq \lambda x. \mathscr{D}'$. Then, we have both   $\lambda x. \mathcal{E}'[P]\Downarrow \mathscr{D}$ and $ \frac{1}{2}\cdot \mathscr{D}_1+ \frac{1}{2}\cdot \mathscr{D}_2 \leq_\dist \mathscr{D}$.\\
$\bullet$ Suppose $\mathcal{E}= \mathcal{E}'L$ and let us consider the case $P \dashrightarrow_\frac{1}{2}P_1$ and $P \dashrightarrow_\frac{1}{2}P_2$.  So, for $i\in \lbrace 1,2\rbrace$, $\mathscr{D}_i =  \sum_{\lambda x. H'\,  \in \, \mathrm{ supp}(\mathscr{E}_i)} \mathscr{E}_i(\lambda x. H')\cdot \mathscr{F}^i_{H', L} + \sum_{  H' \, \in \,\mathrm{ supp}(\mathscr{E}_i)\cap \neu} \mathscr{E}_i(H')\cdot H' L $, and the last rule of the  derivation of  $\mathcal{E}[P_i]\Downarrow \mathscr{D}_i$ is $s4$ with premises $\mathcal{E}'[P_i]  \Downarrow \mathscr{E}_i$ and $\lbrace H'[L/x]\Downarrow \mathscr{F}^i_{H, L} \rbrace_{\lambda x. H'\, \in\,  \mathrm{ supp}(\mathscr{E}_i)}$.
The proof is similar to point~\ref{enum: case distribution for one side equivalence}.
%By applying the induction hypothesis, we have that there exists $\mathcal{E}'[P]\Downarrow \mathscr{E}$, where $\mathscr{E}\triangleq \frac{1}{2}\cdot \mathscr{E}_1+ \frac{1}{2}\cdot \mathscr{E}_2$. Hence, we can construct:
%\begin{prooftree}
%\AxiomC{$\mathcal{E}'[P]  \Downarrow \mathscr{E}$}
%\AxiomC{$\lbrace H'[L/x]\Downarrow \mathscr{F}_{H, L} \rbrace_{\lambda x. H'\, \in\,  \mathrm{ supp}(\mathscr{E})}$}
%\RightLabel{$s4$}
%\BinaryInfC{$\mathcal{E}'[P]L \Downarrow  \sum_{\lambda x. H'\,  \in \, \mathrm{ supp}(\mathscr{E})} \mathscr{E}(\lambda x. H')\cdot \mathscr{F}_{H', L} + \sum_{  H' \, \in \,\mathrm{ supp}(\mathscr{E})\cap \neu} \mathscr{E}(H')\cdot H' L$}
%\end{prooftree}
\end{proof}
\begin{lem} \label{lem: equivalence part 2} Let $M \in \plam$. For all $H \in \h$, $ \mathcal{S}^\infty(M, H)\leq \inter{M}(H)$. 
\end{lem}
\begin{proof}
We prove by induction on $n \in \mathbb{N}$ that  there exists $\mathscr{D}$ such that $M \Downarrow \mathscr{D}$ and,  $\forall H \in \h$,    $ \mathcal{S}^n(M, H)\leq \mathscr{D}(H)$.  The case $n=0$ is trivial, so let $n>0$. If $M$ is a head normal form, then $\mathcal{S}^n(M, H)=0$ and we take $\mathscr{D}\triangleq \bot$. Otherwise, we have two cases:
\begin{enumerate}[(a)]
\item If $M\dashrightarrow_1 M'$ then we have $\mathcal{S}^n( M, H)= \mathcal{S}^{n-1}(M', H)$, by Lemma~\ref{lem: properties of S infty for one direction}.\ref{enum: S infty sum}. By induction hypothesis there exists $\mathscr{D}$ such that $M' \Downarrow \mathscr{D}$ and $\mathcal{S}^{n-1}(M', H)\leq \mathscr{D}(H)$, for all $H \in \h$. By applying Lemma~\ref{lem: to prove lemma about head contexts }.\ref{enum: lemma about head context 1}, $M \Downarrow \mathscr{D}$.  
\item If $M \dashrightarrow_\frac{1}{2}M'$ and $M \dashrightarrow_\frac{1}{2}M''$ then,   
 by Lemma~\ref{lem: properties of S infty for one direction}.\ref{enum: S infty sum}, we have $\mathcal{S}^n( M, H)=\frac{1}{2}\cdot  \mathcal{S}^{n-1}(M', H)+ \frac{1}{2}\cdot  \mathcal{S}^{n-1}(M'', H)$. By induction hypothesis there exist $\mathscr{D}'$ and $\mathscr{D}''$ such that $M' \Downarrow \mathscr{D}'$, $M''\Downarrow \mathscr{D''}$,   $\mathcal{S}^{n-1}(M', H)\leq \mathscr{D}'(H)$, and $\mathcal{S}^{n-1}(M'', H)\leq \mathscr{D}''(H)$, for all $H \in \h$. By applying Lemma~\ref{lem: to prove lemma about head contexts }.\ref{enum: lemma about head context 2}, there exists $\mathscr{D}$ such that $M \Downarrow \mathscr{D}$ and $\frac{1}{2}\cdot \mathscr{D}'+ \frac{1}{2}\cdot \mathscr{D}''  \leq_\dist \mathscr{D}$.
\end{enumerate}
\end{proof}
We are now able to prove that $\mathcal{H}^\infty$,  $\mathcal{S}^\infty$ and $\inter{\cdot}$ are all equivalent operational semantics:
\begin{thm}[Equivalence] \label{thm: equivalence of head reduction and head spine reduction}Let $M \in \plam$. For all $H \in \h$, $ \mathcal{H}^\infty(M, H)=\mathcal{S}^\infty(M, H)= \inter{M}(H)$.
\end{thm}
\begin{proof}
Let $H \in \h$. By Theorem~\ref{thm: head  equal to heas spine}, we have $ \mathcal{H}^\infty(M, H)= \mathcal{S}^\infty(M, H)$. By Lemma~\ref{lem: equivalence part 1} and Lemma~\ref{lem: equivalence part 2}, we have $\mathcal{S}^\infty(M, H)= \inter{M}(H)$.
\end{proof}
As expected, Proposition~\ref{prop: the semantics is invariant under reduction}.\ref{enum: invariance beta} says that the operational semantics $\inter{\cdot}$ in~\eqref{eq: big-step semantics}  is invariant under the head \emph{spine} reduction step rewriting  $(\lambda x. H)N$ into $H[N/x]$, where $H \in \h$.  A consequence of Theorem~\ref{thm: equivalence of head reduction and head spine reduction}  is that  $\inter{\cdot}$ is also invariant under the usual head reduction step rewriting  $(\lambda x. M)N$ into $M[N/x]$:
\begin{cor} \label{cor: operational semantics invariant under head reduction step} Let $M, N \in \Lambda_\oplus$. Then $\inter{(\lambda x. M)N}= \inter{M[N/x]}$.
\end{cor}
\begin{proof}
By Lemma~\ref{lem: properties of S infty for one direction}.\ref{enum: S infty sum}, for all $n \in \mathbb{N}$ and $H \in \h$,  we have   $\mathcal{H}^n((\lambda x. M)N, H)=\mathcal{H}^{n-1}(M[N/x], H)$. This means that  $\mathcal{H}^\infty((\lambda x. M)N, H)=\mathcal{H}^{\infty}(M[N/x], H)$. We conclude by Theorem~\ref{thm: equivalence of head reduction and head spine reduction}.
\end{proof}
\section{Proofs of Section~\ref{sec3}}
\textbf{Lemma~\ref{lem: abstraction congruence for obs and app}.}  Let $M, N \in \plam^{\Gamma \cup \lbrace x \rbrace}$. Then:
\begin{enumerate}[(1)]
\item \label{eqn: abstraction congruence for obs} If $M \leq _{\mathrm{app}}N$ then  $\lambda x. M \leq_{\mathrm{app}}\lambda x. N$.
\item  \label{eqn: abstraction congruence for app} If $\lambda x.M \leq _{\mathrm{cxt}}\lambda x.N$ then $M \leq_{\mathrm{cxt}}N$.
\item \label{eqn: application congruence for obs} If $M \leq _{\mathrm{cxt}}N$ then, for all $L \in \plam$,  $ML\leq _{\mathrm{cxt}}NL$. 
\end{enumerate}
\begin{proof}
Concerning point~\ref{eqn: abstraction congruence for obs}, let us suppose  that $\lambda x.M \leq_{\mathrm{app}}\lambda x. N$ does not hold. Then, there exists an applicative context  $\mathcal{C}=(\lambda x_1\ldots x_n.[\cdot])P_1\ldots P_m$ such that  $\sum \inter{\mathcal{C}[\lambda x.N]}< \sum \inter{\mathcal{C}[\lambda x. M]}$. We consider the applicative context $\mathcal{C}'\triangleq \mathcal{C}[\lambda x.[\cdot]]$. Then $\textstyle \sum \inter{\mathcal{C}'[N]}=\textstyle \sum \inter{\mathcal{C}[\lambda x. N]}< \textstyle\sum\inter{\mathcal{C}[\lambda x. M]} =\textstyle \sum \inter{\mathcal{C}'[M]}$.
Therefore,  $M \leq _{\mathrm{app}}N$ does not hold.\\
Let us now  prove point~\ref{eqn: abstraction congruence for app}. 
Suppose that  $M \leq_{\mathrm{cxt}}N$ does not hold. Then, there exists $\mathcal{C}\in \clam$  such that $\sum \inter{\mathcal{C}[N]}< \sum \inter{\mathcal{C}[M]}$. We consider the context $\mathcal{C}'\triangleq \mathcal{C}[ [\cdot ]x]$. By applying Corollary~\ref{cor: operational semantics invariant under head reduction step} twice and Lemma~\ref{lem: operational semantics monotonicity contexts}, we can conclude  $\sum \inter{\mathcal{C}'[\lambda x. N]}= \sum \inter{\mathcal{C}[(\lambda x. N)x]}= \sum \inter{\mathcal{C}[N]}< \sum\inter{\mathcal{C}[M]} =\sum \inter{\mathcal{C}[(\lambda x. M)x]}=\sum \inter{\mathcal{C}'[\lambda x. M]}$.
Hence,  $\lambda x. M \leq _{\mathrm{cxt}}\lambda x.N$ does not hold.\\
Last, we prove point~\ref{eqn: application congruence for obs}. Suppose $M \leq_{\mathrm{cxt}}N$ and let $\mathcal{C}\in \clam$. By defining $\mathcal{C}'\triangleq \mathcal{C}[[\cdot]L]$ we have $\sum \inter{\mathcal{C}[ML]}= \sum \inter{\mathcal{C}'[M]}\leq \sum \inter{\mathcal{C}'[N]}= \sum \inter{\mathcal{C}[NL]}$. 
Therefore, $ML \leq_{\mathrm{cxt}}NL$.
\end{proof}
\noindent
\textbf{Lemma~\ref{lem: fundamental step toward context lemma}.} Let $M ,N \in \elam$ be such that  $M \leq_{\mathrm{app}} N$. Then  $\sum \inter{\mathcal{C}[M]}\leq \sum \inter{\mathcal{C}[N]}$, for all $ \mathcal{C} \in \genlam$. 
\begin{proof}
By  Theorem~\ref{thm: equivalence of head reduction and head spine reduction},  it is enough to show that, for all $n \in \mathbb{N}$ and for all  contexts $\mathcal{C}\in \genlam$:
\begin{equation}\label{eqn: equantion to prove for context lemma}
 \sum_{H \in \h} \mathcal{H}^n (\mathcal{C}[M], H)\leq   \sum_{H \in \h}\mathcal{H}^\infty(\mathcal{C}[N], H)  .
\end{equation}
Henceforth, we  write $\sum \mathcal{H}^n (\mathcal{C}[M])$ (resp.~$\sum \mathcal{H}^\infty (\mathcal{C}[M])$) in place of $\sum_{H \in \h} \mathcal{H}^n (\mathcal{C}[M], H)$ (resp.~$\sum_{H \in \h} \mathcal{H}^\infty (\mathcal{C}[M], H)$). 
The proof is by induction on  $(n, \vert \mathcal{C}\vert )$, where $n \in \mathbb{N}$ and $\vert \mathcal{C}\vert$ is the size  of $\mathcal{C}\in \genlam$, i.e.~the number of nodes in the syntax tree of $\mathcal{C}$. First, note that $\mathcal{C}$ must be of the form  $\mathcal{C}_0\mathcal{C}_1\ldots \mathcal{C}_k$, for some $k \in \mathbb{N}$.  We have several cases depending on the structure of $\mathcal{C}_0$:
\begin{enumerate}[(a)]
\item \label{enum: cases context lemma a} $\mathcal{C}_0=x$ then both $\mathcal{C}[M]$ and $\mathcal{C}[N]$ are head normal forms, and the inequation in~\eqref{eqn: equantion to prove for context lemma} is straightforward.
\item  \label{enum: cases context lemma b}  If $\mathcal{C}_0= \lambda x. \mathcal{C}'$ then we have two cases:
\begin{enumerate}[(i)]
\item If $k=0$ then, by Lemma~\ref{lem: properties of S infty for one direction}.\ref{enum: S infty abst} and by  induction hypothesis, $\sum \mathcal{H}^n( \lambda x.\mathcal{C}'[M])=\sum \mathcal{H}^n( \mathcal{C}'[M]) \leq \sum \mathcal{H}^\infty(\mathcal{C}'[N])=\sum \mathcal{H}^\infty(\lambda x. \mathcal{C}'[N]) $.  
\item For $k>0$  we have two cases depending on $n \in \mathbb{N}$.  If $n=0$ then    $\sum \mathcal{H}^0 ((\lambda x. \mathcal{C}'[M])\mathcal{C}_1[M]\ldots \mathcal{C}_k[M])=0$ by Lemma~\ref{lem: properties of S infty for one direction}.\ref{enum: S infty sum}.  Otherwise, by  Lemma~\ref{lem: properties of S infty for one direction}.\ref{enum: S infty sum} and by using the induction hypothesis, we  have:
\allowdisplaybreaks
\begin{align*}
\ \ \ \ \ \ \ \ \ \ &\sum \mathcal{H}^n ((\lambda x. \mathcal{C}'[M])\mathcal{C}_1[M]\ldots \mathcal{C}_k[M])=\\
&= \sum \mathcal{H}^{n-1} ((( \mathcal{C}'[M])[\mathcal{C}_1[M]/x] ) \mathcal{C}_2[M]\ldots \mathcal{C}_k[M])\\
&\leq  \sum \mathcal{H}^{\infty} ((( \mathcal{C}'[N])[\mathcal{C}_1[N]/x] ) \mathcal{C}_2[N]\ldots \mathcal{C}_k[N])\\
&=\sum \mathcal{H}^{\infty} ((\lambda x. \mathcal{C}'[N])\mathcal{C}_1[N]\ldots \mathcal{C}_k[N]) .
\end{align*}
\end{enumerate}
\item \label{enum: cases context lemma c} If  $\mathcal{C}_0= \mathcal{C}' \oplus \mathcal{C}''$, then we have two cases depending on $n \in \mathbb{N}$.  If $n=0$,   Lemma~\ref{lem: properties of S infty for one direction}.\ref{enum: S infty sum} implies  $\sum \mathcal{H}^n( (\mathcal{C}'[M] \oplus \mathcal{C}''[M])\mathcal{C}_1[M]\ldots \mathcal{C}_k[M])=0$.  Otherwise,  by using the induction hypothesis and by  Lemma~\ref{lem: properties of S infty for one direction}.\ref{enum: S infty sum}, we  have:
\allowdisplaybreaks
\begin{align*}
&\sum \mathcal{H}^n( (\mathcal{C}'[M] \oplus \mathcal{C}''[M])\mathcal{C}_1[M]\ldots \mathcal{C}_k[M])=\\
&= \frac{1}{2} \sum \mathcal{H}^{n-1} (\mathcal{C}'[M] \mathcal{C}_1[M]\ldots \mathcal{C}_k[M])\\
&\phantom{= \ }+ \frac{1}{2} \sum \mathcal{H}^{n-1} ( \mathcal{C}''[M] \mathcal{C}_1[M]\ldots \mathcal{C}_k[M])\\
&\leq  \frac{1}{2}\sum  \mathcal{H}^{\infty} (\mathcal{C}'[N] \mathcal{C}_1[N]\ldots \mathcal{C}_k[N])\\
&\phantom{= \ }+ \frac{1}{2} \sum \mathcal{H}^{\infty} ( \mathcal{C}''[N] \mathcal{C}_1[N]\ldots \mathcal{C}_k[N])\\
&=\sum\mathcal{H}^{\infty} ((\mathcal{C}'[N]\oplus \mathcal{C}''[N])\mathcal{C}_1[N]\ldots \mathcal{C}_k[N]) .
\end{align*}
\item \label{enum: cases context lemma d} The last case is when $\mathcal{C}_0=[\cdot]$. First, note that $M=M_0\ldots M_h$ for some $h \in \mathbb{N}$. Since $M$ is closed, we can assume that $M_0= \lambda x. M'_0$ is an abstraction. We apply Case~\ref{enum: cases context lemma b} to the context $( \lambda x.M'_0)M_1 \ldots M_{h} \mathcal{C}_1 \ldots \mathcal{C}_{k}$, and we have  
$\sum \mathcal{H}^n (M_0M_1 \ldots M_{h}\mathcal{C}_1[M]\ldots \mathcal{C}_k[M])\leq \sum \mathcal{H}^\infty (M_0M_1 \ldots M_{h}\mathcal{C}_1[N]\ldots \mathcal{C}_k[N])$.  Since  it holds that $M \leq_{\mathrm{app}}N$,  we  obtain $\sum \mathcal{H}^\infty (M\mathcal{C}_1[N]\ldots \mathcal{C}_k[N]) \leq \sum  \mathcal{H}^\infty(N \mathcal{C}_1[N]\ldots \mathcal{C}_k[N])$.
\end{enumerate}
\end{proof}	
\noindent
\textbf{Lemma~\ref{lem: context lemma 2}.} Let $H, H'\in \h^{\lbrace x \rbrace}$. Then, the following are equivalent statements:
 \begin{enumerate}[(1)]
 \item \label{enum: context lemma 2 1}$ \lambda x. H \precsim \lambda x. H' ,$
 \item \label{enum: context lemma 2 2}$\nu x. H \precsim \nu x. H',$
 \item \label{enum: context lemma 2 3}$\forall P \in \elam, \ H[P/x]\precsim H'[P/x] .$
 \end{enumerate}
 \begin{proof} Let us first show that point~\ref{enum: context lemma 2 1} implies point~\ref{enum: context lemma 2 2}. By Proposition~\ref{prop: properties dal lago bisimilarity}, if $ \lambda x. H \precsim \lambda x. H' $ then:
\begin{equation*}
 1=\mathcal{P}_\oplus(\lambda x. H, \tau, \lbrace \nu x.H \rbrace)\leq \M{\lambda x. H', \tau, {\precsim}(\nu x. H)}.
\end{equation*} 
 Hence,  $\M{\lambda x. H',\tau, {\precsim}(\nu x. H)}=1$, so that $\nu x. H\precsim \nu x. H'$. To prove that point~\ref{enum: context lemma 2 2} implies point~\ref{enum: context lemma 2 3},  if $\nu x. H \precsim \nu x. H'$ then,   by Proposition~\ref{prop: properties dal lago bisimilarity}, we have:
\begin{equation*}
1=\M{\nu x. H, P, \lbrace H[P/x]\rbrace}\leq \M{\nu x. H', P, {\precsim} (H[P/x])} ,
\end{equation*} 
 for all $P \in \elam$. Hence,  $\M{\nu x. H', P,{\precsim} ( H[P/x])}=1$, so that $H[P/x]\precsim H'[P/x]$.  \\
 We now prove that point~\ref{enum: context lemma 2 3} implies point~\ref{enum: context lemma 2 2}. Let us consider the  relation $\mathcal{R}$ defined by:
 \begin{equation*}
  \lbrace (\nu x. H, \nu x. H' )\in \dhe^2  \ \vert \ \forall P \in \elam, \, H[P/x]\precsim H'[P/x] \rbrace  \cup  {\precsim} 
 \end{equation*}
 where $\dhe^2= \dhe \times \dhe$. Clearly, $\mathcal{R}$ is a preorder because $\precsim$ is. Now, if we show that $\mathcal{R}$ is a simulation  then  $\mathcal{R}\subseteq\  \precsim$, so that $\nu x.H \precsim \nu x. H'$ holds whenever $ H[P/x]\precsim H'[P/x]$ for all $ P \in \elam$. The only interesting case is $\R{\nu x. H}{\nu x. H'}$. Let $P\in \elam$. By definition, we have $H[P/x]\precsim H'[P/x]$, so that:
\begin{equation*}
\begin{split}
\M{\nu x. H, P, \lbrace H[P/x]\rbrace }&\leq \M{\nu x. H', P, {\precsim} (\lbrace H[P/x] \rbrace)}\\
&\leq \M{\nu x. H', P, \mathcal{R} (\lbrace H[P/x] \rbrace)} . 
\end{split}
\end{equation*} 
Finally, we prove that point~\ref{enum: context lemma 2 2} implies point~\ref{enum: context lemma 2 1}. Let us consider the following relation:
 \begin{equation*}
 \mathcal{R}\triangleq \lbrace (\lambda  x. H, \lambda x. H' )\in \h \times 	\h \ \vert \ \nu x.H \precsim \nu x. H' \rbrace  \cup  {\precsim}  .
 \end{equation*}
It is a preorder because $\precsim$ is. Now, if we show that $\mathcal{R}$ is a simulation  then  $\mathcal{R}\subseteq\  \precsim$, so that $\lambda x.H \precsim \lambda x. H'$ whenever $\nu x. H \precsim \nu x. H'$.  The only interesting case is $\R{\lambda  x. H}{\lambda x. H'}$. By definition, we have  $\nu x. H \precsim \nu x. H'$, so that  $\M{\lambda x. H, \tau, \lbrace \nu x. H  \rbrace }\leq \M{\lambda x. H', \tau,{\precsim} (\lbrace \nu x.H \rbrace)}\leq \M{\lambda  x. H', \tau, \mathcal{R} (\lbrace \nu x. H  \rbrace)}$.
  \end{proof}
\noindent
\textbf{Lemma~\ref{lem: commutation abstraction precsim}.}    Let $X \subseteq \h^{\lbrace x \rbrace}$. We have:
 \begin{align*}
{\precsim}(\lambda x.X)\cap\he&=\lambda x. {\precsim} (X)\cap\he,\\
{\precsim}(\nu x.X)&= \nu x. {\precsim} (X)\enspace.
 \end{align*}
 \begin{proof}
 Let us  prove the first equation. For all $\lambda x. H \in \he$ we have:
\allowdisplaybreaks
\begin{align*}
\lambda x. H \in  {\precsim} (\lambda x. X)&\Leftrightarrow \exists H' \in X, \ \lambda x. H' \precsim \lambda x. H\\
%&\Leftrightarrow \exists H' \in X, \  \forall P \in \elam, \  H'[P/x] \precsim H[P/x] &&\text{Lemma}~\ref{lem: context lemma 2} \\
&\Leftrightarrow   \exists H' \in X,  \  H' \precsim H &&\text{by}~\eqref{eqn: open term simil} \\
&\Leftrightarrow \lambda x. H \in \lambda x. {\precsim} (X) .
\end{align*}
Concerning the second equation, first note that ${\precsim} (\nu x. X)$ contains only distinguished head normal forms. Indeed, suppose $M \in {\precsim} (\nu x.X)$ for some term $M\in \elam$. Then, there exists $H \in X$ such that $\nu x.H  \precsim M$.  By Proposition~\ref{prop: properties dal lago bisimilarity}, we would have $1=\M{\nu x. H, P, \lbrace H[P/x]\rbrace }\leq \M{M, P, {\precsim}(\lbrace H[P/x]\rbrace)}=0$.  Then, for all   $\nu x.H \in \dhe$, we have:
\allowdisplaybreaks
\begin{align*}
\nu x. H \in  {\precsim} (\nu x. X)&\Leftrightarrow \exists H' \in X, \ \nu x. H' \precsim \nu x. H\\
&\Leftrightarrow  \exists H' \in X, \ \lambda x. H' \precsim \lambda x. H &&\text{Lemma}~\ref{lem: context lemma 2} \\
&\Leftrightarrow   \exists H' \in X,  \  H' \precsim H&&\text{by}~\eqref{eqn: open term simil} \\
&\Leftrightarrow \nu x. H \in \nu x. {\precsim} (X) .
\end{align*}
 \end{proof}
 \noindent
 \textbf{Lemma~\ref{lem: context lemmna 1}.}  Let $M, N \in \elam$. For all $X \subseteq \he$, $\inter{M}(X) \leq \inter{N}({\precsim} (X))$ if and only if  $M \precsim N$.
 \begin{proof} 
The right-to-left direction follows from Proposition~\ref{prop: properties dal lago bisimilarity}.
Concerning the converse, we define $\mathcal{R}$ as:
\begin{equation*}
 \lbrace (P, Q)\in \elam \times \elam \ \vert \ \forall X \subseteq \he, \ \inter{P}(X) \leq \inter{Q}({\precsim}(X))\rbrace \cup  {\precsim} 
\end{equation*} 
 If we prove that $\mathcal{R}$ is a probabilistic simulation, then $\mathcal{R}\subseteq \ \precsim$, so that $M \precsim N$ whenever $\inter{M}(X) \leq \inter{N}({\precsim} (X))$, for all $X \subseteq \he$.  So, let us first prove that $\mathcal{R}$ is a preorder. On the one hand, $\mathcal{R} $ is clearly reflexive. On the other hand, let $P, Q, L \in \elam$ be   such that $\R{P}{L}$ and $\R{L}{Q}$.  By Proposition~\ref{prop: properties dal lago bisimilarity}, $\precsim$ is transitive. It follows that, for all $ X \subseteq \he$: 
 \begin{equation*}
  \inter{P}(X)\leq \inter{L}({\precsim} (X))\leq \inter{Q}({\precsim}({\precsim} (X)))\leq \inter{Q}({\precsim} (X)) ,
\end{equation*}  
Now, let $P, Q \in \elam$ be such that $\R{P}{Q}$, and let $X \subseteq  \h^{\lbrace x\rbrace}$. We have:
\allowdisplaybreaks
\begin{align*}
\mathcal{P}_{\oplus}(P, \tau, \nu x. X)&= \inter{P}(\lambda x. X)\\
& \leq \inter{Q}({\precsim}(\lambda x. X))\\
& \leq \inter{Q}({\precsim}(\lambda x. X)\cap \he)&&Q\in \elam\\
&= \inter{Q}(\lambda x.{\precsim}(X)\cap \he)&&\text{Lemma}~\ref{lem: commutation abstraction precsim}\\
&= \inter{Q}(\lambda x.{\precsim}(X))&&Q \in \elam\\
&= \M{ Q, \tau, \nu x. {\precsim} (X)}\\
&=\M{Q, \tau, {\precsim}(\nu x.X}) &&\text{Lemma}~\ref{lem: commutation abstraction precsim} \\
&\leq  \M{Q, \tau, \mathcal{R}(\nu x.X})    .
\end{align*}
Therefore, $\mathcal{R}$ is a probabilistic simulation.
 \end{proof}
\noindent
\textbf{Lemma~\ref{lem: context lemmna 3}} (Key Lemma)\textbf{.}  Let $M, N \in \elam$. If $M \precsim N$ then, for all $P \in \elam$, $MP \precsim NP$.
\begin{proof} By Lemma~\ref{lem: context lemmna 1} it suffices to  prove that, for all $X \subseteq \he$, $\inter{MP}(X)\leq \inter{NP}({\precsim} (X))$. This amounts to show that, for all $\mathscr{D}$ such that $MP \Downarrow \mathscr{D}$, it holds  $\mathscr{D}(X) \leq \inter{NP}({\precsim}(X))$. This is trivial when $\mathscr{D}= \bot$, so that we can assume that  the last rule in the derivation of $MP \Downarrow \mathscr{D}$ is the following:
\begin{prooftree}
\AxiomC{$M \Downarrow \mathscr{E}$}
\AxiomC{$\lbrace H[P/x] \Downarrow \mathscr{F}_{H, P}\rbrace_{\lambda x. H \in\, \mathrm{supp}(\mathscr{E})} $}
\RightLabel{$s4$}
\BinaryInfC{$MP \Downarrow \sum_{\lambda x. H\,  \in\, \mathrm{ supp}(\mathscr{E})} \mathscr{E}(\lambda x.H)\cdot \mathscr{F}_{H, P} $}
\end{prooftree}
 Since $\mathscr{E}$ is a finite distribution, $\mathscr{D}$ is a sum of finitely many summands. Let $\mathrm{supp}(\mathscr{E})$ be $\lbrace \lambda z. H_1, \ldots, \lambda z. H_n\rbrace\subseteq \he$.  We  define the pair $(\lbrace p_i \rbrace_{1 \leq i \leq n}, \lbrace r_I \rbrace_{I \subseteq \lbrace 1, \ldots, n \rbrace})$ as follows:
\begin{enumerate}[(a)]
\item  For all $i \leq n$, $p_i \triangleq \mathscr{E}(\lambda  z.H_i) $.
\item For all $I \subseteq \lbrace 1, \ldots,n \rbrace$:
\begin{equation*}
 r_{I}  \triangleq  \sum_{\substack{\lambda z.H'  \text{ s.t.}\\ 
\lbrace i \leq n \ \vert \ \lambda z.H' \in  {\precsim}( \lambda z. H_i)\rbrace=I}} \inter{N}(\lambda z.H')  .
\end{equation*}
\end{enumerate}
Let us show that $(\lbrace p_i \rbrace_{1 \leq i \leq n}, \lbrace r_I \rbrace_{I \subseteq \lbrace 1, \ldots, n \rbrace})$ is a probabilistic assignment by proving that  
Condition~\eqref{eqn: condition probabilistic assignment} holds. First, from $M \precsim N$ and by Lemma~\ref{lem: context lemmna 1}, we have that $\mathscr{E}(\bigcup_{i \in I} \lbrace \lambda z. H_i\rbrace   ) \leq \inter{N}(\bigcup_{i \in I} {\precsim}(\lambda z. H_i)  )$.  Then, for all $I \subseteq\, \lbrace 1, \ldots, n \rbrace$, we have:
\allowdisplaybreaks
\begin{align*}
\sum_{i \in I} p_i&= \sum_{i \in I}\mathscr{E}(\lambda z. H_i)\\
&= \mathscr{E}(\bigcup_{i \in I} \lbrace \lambda z. H_i \rbrace)\\
& \leq \inter{N}(\bigcup_{i \in I} {\precsim}(\lambda z. H_i)   )\\
&=\inter{N}(\bigcup_{i \in I} {\precsim}(\lambda z. H_i) \, \cap \, \he  ) &&\text{since } N \in \elam\\
&= \sum_{\substack{\lambda z.H' \in\\ \,  \bigcup_{i \in I} {\precsim}(\lambda z. H_i)  }} \inter{N}(\lambda z.H')  \\
&\leq \sum_{ \substack{I'\subseteq \lbrace 1, \ldots, n \rbrace  \\ \text{ s.t. }  I'\cap I \not =\emptyset}} r_{I'} \, .
\end{align*}
By applying Lemma~\ref{lem:  probabilistic assignment entanglement}, for all $I = \lbrace 1, \ldots, n \rbrace$ and for every  $k \in I $ there exists $h_{k, I }\in [0, 1]$ such that:
\allowdisplaybreaks
\begin{align}
&\forall j\leq n:& & p_j \leq \sum_{\substack{J\subseteq \lbrace 1, \ldots, n \rbrace \\ \text{s.t. } j \in J}} h_{j, J} \cdot r_{J } \label{eqn: n, h}\\
&\forall  J \subseteq \lbrace 1, \ldots, n \rbrace:&&  1 \geq \sum_{\substack{j \in \lbrace 1, \ldots, n \rbrace \\ \text{s.t. }j \in J}} h_{j, J }\label{eqn: h leq 1} \,  .
\end{align}
We now show that, for all $\lambda z.H' \in  \bigcup_{i \in I} {\precsim}(\lambda z. H_i)  $, there exist $n$ real numbers $s^{H'}_{1}$, \ldots, $s_{n}^{H'}$ such that:
\allowdisplaybreaks
\begin{align}
&\forall i\leq n:&&  \mathscr{E}(\lambda z. H_i)\leq \sum_{\substack{\lambda z.H' \in\\ \,  {\precsim}( \lambda z.H _i)}} s_{i}^{H'} \label{eqn: n, s}\\
&\forall \lambda z.H' \in   \bigcup_{i \in I} {\precsim}(\lambda z. H_i) : && \inter{N}(\lambda z.H')\geq \sum_{i=1}^n s_{i}^{H'}\label{eqn: s t leq u} \, .
\end{align}
For all $i \leq n$ and for all $\lambda z.H' \in\,  {\precsim}(\lambda z.H_i)$, we set:
\allowdisplaybreaks
\begin{equation*}
 s^{H'}_i \triangleq h_{i, \lbrace k\leq n \ \vert \ \lambda z.H' \in {\precsim}(\lambda z. H_k)\rbrace}\cdot \inter{N}(\lambda z.H') .
\end{equation*}
Concerning the inequation in~\eqref{eqn: n, s}, by using the inequation in~\eqref{eqn: n, h}  we have, for all $i \leq n$:
\allowdisplaybreaks
\begin{align*}
& \mathscr{E}(\lambda z. H_i)\leq \\
&\leq \sum_{\substack{  I \subseteq \lbrace 1, \ldots n \rbrace \\ \text{s.t. }i \in I }}h_{i, I }\cdot r_{I }\\
&= \sum_{\substack{  I \subseteq \lbrace 1, \ldots n \rbrace \\ \text{s.t. }i \in I }} h _{i, I }\cdot \Bigg(    \sum_{\substack{\lambda z.H'  \text{ s.t.}\\ 
\lbrace k \leq n \ \vert \ \lambda z.H' \in   {\precsim}( \lambda z. H_k)\rbrace=I}} \inter{N}(\lambda z.H')  \Bigg)\\
 &= \sum_{ \lambda z.H' \in {\precsim}( \lambda z.H_i)} h_{i, \lbrace k\leq n \ \vert \ \lambda z.H' \in {\precsim}( \lambda z. H_k)\rbrace}\cdot \inter{N}(\lambda z.H')\\
 &=  \sum_{\lambda z.H' \in  {\precsim}( \lambda z.H _i)} s_{i}^{H'}  .
\end{align*}
As for the inequation in~\eqref{eqn: s t leq u}, by using the inequation in~\eqref{eqn: h leq 1} we have, for all $\lambda z.H' \in  \bigcup_{i \in I} {\precsim}(\lambda z. H_i)$:
\allowdisplaybreaks
\begin{align*}
  \sum_{i=1}^n s^{H'}_i  &=  \sum^n_{i=1}h_{i,\lbrace k\leq n \ \vert \ \lambda z.H' \in\, {\precsim}(\lambda z. H_k)\rbrace}\cdot \inter{N}(\lambda z.H')\\
  &\leq \inter{N}(\lambda z.H') .
\end{align*}
We are now able to prove that $\mathscr{D}(X)\leq \inter{NP}({\precsim}(X))$.   First, by applying  Lemma~\ref{lem: context lemma 2} and Lemma~\ref{lem: context lemmna 1}, for all $i \leq n$, for all $ \lambda z. H' \in  {\precsim}( \lambda z. H_i)$, for all $P \in \elam$,  and for all $ X \subseteq \he$:
\begin{equation}\label{eqn: second inequation context lemma 3}
\mathscr{F}_{H_i, P}(X)\leq \inter{H_i[P/x]}(X)\leq \inter{H'[P/x]}({\precsim}(X))  .
\end{equation}
Therefore, for all $X \subseteq \he$:
\allowdisplaybreaks
\begin{align*}
&\mathscr{D}(X)\leq\\
 &\leq \sum^n_{i=1}\Bigg( \sum_{\lambda z.H' \in {\precsim}(\lambda z. H_i)} s_{i}^{H'} \Bigg) \cdot \mathscr{F}_{H_i, P} (X)= &&\text{by}~\eqref{eqn: n, s}  \\
& =  \sum^n_{i=1}\sum_{\lambda z.H' \in {\precsim}(\lambda z. H_i)} s_{i}^{H'}\cdot \mathscr{F}_{H_i, P} (X) \\
&\leq  \sum^n_{i=1}\sum_{\substack{\lambda z. H' \in   {\precsim}(\lambda z. H_i)}} s_{i}^{H'}\cdot \inter{H'[P/z]}({\precsim}  (X)) &&\text{by}~\eqref{eqn: second inequation context lemma 3}\\
&\leq   \sum^n_{i=1}\sum_{\substack{\lambda z.H' \in \\  \, \bigcup^n_{i=1}{\precsim}(\lambda z.H_i)}} s_{i}^{H'}\cdot \inter{H'[P/z]} ({\precsim}(X)) \\
&\leq  \sum_{\substack{ \lambda z.H' \in\\ \,\bigcup^n_{i=1}  {\precsim}(\lambda z.H_i)}} \bigg(\sum^n_{i=1} s_{i}^{H'}\bigg)\cdot   \inter{H'[P/z]} ({\precsim}(X)) \\
&\leq  \sum_{\substack{ \lambda z.H' \in\\ \,\bigcup^n_{i=1}  {\precsim}(\lambda z.H_i)}}    \inter{N}(\lambda z.H') \cdot   \inter{H'[P/z]} ({\precsim}(X))&&\text{by}~\eqref{eqn: s t leq u} \\
&\leq  \sum_{ \substack{\lambda z.H'\in\\ \, \mathrm{supp}(\inter{N})}}   \inter{N}(\lambda z.H') \cdot   \inter{H'[P/z]} ({\precsim}(X))\\
&=\inter{NP}({\precsim}(X))&& \text{Prop.}~\ref{prop: the semantics is invariant under reduction}.\ref{lem: invariance beta general case}
\end{align*}
and hence $\mathscr{D}(X)\leq \inter{NP}({\precsim}(X))$.
\end{proof}
\section{Proofs of Section~\ref{sec4}}\label{app4}
\textbf{Lemma~\ref{fact: descending chain}.}  Let $\lbrace A_n \rbrace_{n \in \mathbb{N}}$ be a  descending chain of countable sets of positive real numbers satisfying $\sum_{r \in A_n}r< \infty$, for all $n \in \mathbb{N}$.  Then:
\begin{equation}\label{eqn: infinite descending chain}
\sum_{r\,  \in \, \bigcap_{n\in \mathbb{N}} A_n }r = \inf_{n \in \mathbb{N}} \bigg( \sum_{r \, \in\, A_n} r \bigg) \, .
\end{equation}
\begin{proof}
Henceforth,  if $A$ is a subset of real numbers, we let   $\mes{A}$ denote $\sum_{r \in A}r$. First, notice that it suffices to prove the following particular situation:
\begin{equation}\label{eqn: particular case infinite descending chain}
\text{if } \bigcap_{n \in \mathbb{N}} A_n= \emptyset \text{ then } \inf_{m \in \mathbb{N}}  \mes{A_m} =0  .
\end{equation}
Let us  show that the implication in~\eqref{eqn: particular case infinite descending chain} gives us the equation in~\eqref{eqn: infinite descending chain}.  So, consider the chain $\lbrace B_n \rbrace_{n \in \mathbb{N}}$ defined by $B_n \triangleq A_n \setminus \bigcap_{m \in \mathbb{N}} A_m$. Since $\bigcap_{n \in \mathbb{N}} B_n= \emptyset$, then  $\inf_{m \in \mathbb{N}}\mes{ B_m }=0 $ by~\eqref{eqn: particular case infinite descending chain}. We have:
\allowdisplaybreaks
\begin{align*}
\mes{ \bigcap_{n \in \mathbb{N}}A_n}& =  \mes{ \bigcap_{n \in \mathbb{N}}A_n } +\inf_{m \in \mathbb{N}}\mes{ B_m}  \\
&=\inf_{m \in \mathbb{N}} ( \mes{ \bigcap_{n \in \mathbb{N}}A_n} +\mes{ B_m}  )\\
&=\inf_{m \in \mathbb{N}} ( \mes{ \bigcap_{n \in \mathbb{N}} A_n \cup  B_m}  )\\
&= \inf_{m \in \mathbb{N}} \mes{ A_m}   .
\end{align*}
 So, let us prove~\eqref{eqn: particular case infinite descending chain} and suppose $\bigcap_{n \in \mathbb{N}}A_n= \emptyset$. Since $\lbrace A_n \rbrace_{n \in \mathbb{N}}$ is a descending chain such that $\forall n\in \mathbb{N}$ $\mes{A_n}<\infty$, we have that $\mes{A_n}_{n \in \mathbb{N}}$ is a monotone  decreasing sequence of positive real numbers. This means that  $\lim_{n \to \infty} \mes{A_n}= \inf_{n \in \mathbb{N}}\mes{A_n}$. Thus, to prove the statement, it suffices to show that for all $\epsilon >0$ there exists $k \in \mathbb{N}$ such that  for all $m \geq k$  it holds that $\mes{A_m}<\epsilon$. Now, given a $A_n$ and $\epsilon>0$, there always exists a finite subset of $A_n$, let us call it $A^*_n$,  such that $\mes{A^*_n}\approx_{\epsilon}  \mes{A_n}$. Moreover, since $\bigcap_{n \in \mathbb{N}} A_n= \emptyset$,  for all $r \in A^*_n$ there exists a $n_r\in \mathbb{N}$ such that $r \not \in A_{n_r}$. By considering $A_k$ such that $k \triangleq \max_{r \in A^*_{n}}n_r$ we have $A_k \subseteq  A_n \setminus A^*_n$. Hence, $\mes{A_k}\leq \mes{A_n \setminus A^*_n}= \mes{A_n}- \mes{A^*_n}<\epsilon$.
\end{proof}
\noindent
\textbf{Inequation~\eqref{enum: preliminary counterex 1} of Lemma~\ref{lem: counterexample M and N context preorder }.} Let $M \in \plam$. Then: 
\begin{equation*}
  \inter{M[\mathbf{\Omega}/x]}\leq_\dist   \inter{M[\mathbf{I}/x]}.
\end{equation*}
\begin{proof}
Let us consider the context $(\lambda x. M)[\cdot]\in \clam$. Since $\inter{\mathbf \Omega}\leq_\dist \inter{\mathbf I}$, by applying Lemma~\ref{lem: operational semantics monotonicity contexts}  we obtain $\inter{(\lambda x.M)\mathbf \Omega}\leq_{\dist} \inter{(\lambda x. M)\mathbf I}$. From Corollary~\ref{cor: operational semantics invariant under head reduction step}, we conclude $\inter{M[\mathbf \Omega /x]}\leq _{\dist}\inter{M[\mathbf I/x]}$.
\end{proof}
\noindent
\textbf{Inequation~\eqref{enum: preliminary counterex 2} of Lemma~\ref{lem: counterexample M and N context preorder }.} Let $M \in \plam$. Then:
\begin{equation*}
 \sum \inter{M[(\mathbf{\Omega} \oplus \mathbf{I})/x]}\leq \frac{1}{2}\cdot \sum \inter{M[\mathbf{\Omega}/x]}+ \frac{1}{2}\cdot \sum \inter{M[\mathbf{I}/x]}.
\end{equation*}
\begin{proof}
By Theorem~\ref{thm: equivalence of head reduction and head spine reduction} it is enough to  prove  the following inequation for all $n \in \mathbb{N}$:
\begin{multline}\label{eqn: last equation to prove}
\sum_{H \in \h}  \mathcal{H}^{n}(M[(\mathbf{\Omega} \oplus \mathbf{I})/x], H) \\
\leq  \sum_{H \in \h} \frac{1}{2}\cdot   \mathcal{H}^{\infty}(M[\mathbf{\Omega}/x], H)
 + \frac{1}{2}\cdot  \mathcal{H}^{\infty}(M[\mathbf{I}/x], H).
\end{multline}
The proof is  by induction on $(n, \vert M \vert)$, where $n \in \mathbb{N}$ and $\vert M \vert$ is the size of $M$, i.e.~the number of nodes in the syntax tree of $M$.  We have several cases:\\
If $M= \lambda x. M'$ then, by using the induction hypothesis and Lemma~\ref{lem: properties of S infty for one direction}.\ref{enum: S infty abst}:
\allowdisplaybreaks
\begin{align*}
&\sum_{H \in \h}  \mathcal{H}^{n}(M[(\mathbf{\Omega} \oplus \mathbf{I})/x], H)=\\
&=\sum_{\lambda x. H \in \h}  \mathcal{H}^{n}(\lambda x. (M'[(\mathbf{\Omega} \oplus \mathbf{I})/x]), \lambda x.H)\\ 
&=\sum_{ H \in \h}  \mathcal{H}^{n}(M'[(\mathbf{\Omega} \oplus \mathbf{I})/x],H)\\
&\leq \sum_{H \in \h} \frac{1}{2}\cdot   \mathcal{H}^{\infty}(M'[\mathbf{\Omega}/x], H)+ \frac{1}{2}\cdot  \mathcal{H}^{\infty}(M'[\mathbf{I}/x], H)\\
&=\sum_{H \in \h} \frac{1}{2}\cdot   \mathcal{H}^{\infty}(M[\mathbf{\Omega}/x], H)+ \frac{1}{2}\cdot  \mathcal{H}^{\infty}(M[\mathbf{I}/x], H) .
\end{align*}
Suppose now that  $M$ is a head normal form. From the previous case we can assume w.l.o.g.~that $M$ is a neutral term of the form $ y\vec{P}$, where   $\vec{P}=P_1\ldots P_m$ for some $m \in \mathbb{N}$ and $P_1,\ldots, P_m\in \plam$.  If  $y \not = x$ then $y\vec{P}[(\mathbf{\Omega} \oplus \mathbf{I})/x]$, $y\vec{P}[\mathbf{\Omega}/x]$, and $y\vec{P}[\mathbf{I}/x]$ are  head normal forms, and the inequation in~\eqref{eqn: last equation to prove}  is straightforward. Otherwise, $y=x$.  If $n\geq 2$ then, by using the induction hypothesis,  Lemma~\ref{lem: properties of S infty for one direction}.\ref{enum: S infty sum}, and Equation~\eqref{enum: preliminary counterex 1}, we have:
\begin{align*}
&\sum_{H \in \h} \mathcal{H}^{n} (M[(\mathbf{\Omega} \oplus \mathbf{I})/x], H)=\\
&=\sum_{H \in \h} \mathcal{H}^{n} ( (\mathbf{\Omega} \oplus \mathbf{I})\vec{P}[(\mathbf{\Omega} \oplus \mathbf{I})/x] , H)\\
&=\sum_{H \in \h}\frac{1}{2}\cdot \mathcal{H}^{n-1} ( \mathbf{\Omega} \vec{P}[(\mathbf{\Omega} \oplus \mathbf{I})/x] , H)\\
&\phantom{=\ }+\frac{1}{2}\cdot \mathcal{H}^{n-1} ( \mathbf{I} \vec{P}[(\mathbf{\Omega} \oplus \mathbf{I})/x] , H) \\
 &=\frac{1}{2}  \sum_{H \in \h}\mathcal{H}^{n-2} (  \vec{P}[(\mathbf{\Omega} \oplus \mathbf{I})/x] , H)\\
 &\leq  \frac{1}{2}  \sum_{H \in \h}  \frac{1}{2}\cdot   \mathcal{H}^{\infty} ( \vec{P}[\mathbf{\Omega} /x], H)+ \frac{1}{2}\cdot \mathcal{H}^{\infty} ( \vec{P}[\mathbf{I} /x], H)\\
 &\leq  \frac{1}{2}  \sum_{H \in \h}  \frac{1}{2}\cdot   \mathcal{H}^{\infty} ( \vec{P}[\mathbf{I} /x], H)+ \frac{1}{2}\cdot \mathcal{H}^{\infty} ( \vec{P}[\mathbf{I} /x], H)\\
 &=  \sum_{H \in \h}  \frac{1}{2}\cdot   \mathcal{H}^{\infty} ( \vec{P}[\mathbf{I} /x], H)\\
 &=\sum_{H \in \h}   \frac{1}{2}\cdot  \mathcal{H}^{\infty} ( \mathbf{\Omega}\vec{P}[\mathbf{\Omega} /x] , H)+ \frac{1}{2}\cdot  \mathcal{H}^{\infty} ( \mathbf{I} \vec{P}[\mathbf{I}/x] , H) \\
 &=\sum_{H \in \h}  \frac{1}{2}\cdot \mathcal{H}^{\infty} (M[\mathbf{\Omega}/x], H)+ \frac{1}{2}\cdot  \mathcal{H}^{\infty} (M[\mathbf{I}/x], H) .
\end{align*}
If $n< 2$ then $\mathcal{H}^{n} (M[(\mathbf{\Omega} \oplus \mathbf{I})/x])=0 $. \\
Last, suppose  that $M$ is not a head normal form. By using the induction hypothesis,   Lemma~\ref{lem: properties of S infty for one direction}.\ref{enum: S infty sum} and Lemma~\ref{lem: properties of S infty for one direction}.\ref{enum: H infty app}, we have:
\allowdisplaybreaks
\begin{align*}
& \sum_{H \in \h} \mathcal{H}^n(M[(\mathbf{\Omega} \oplus \mathbf{I})/x], H)=\\
 &= \sum_{H \in \h}    \sum_{l+l'=n} \sum_{H' \in \h}\mathcal{H}^{l}(M, H')\cdot \mathcal{H}^{l'}(H'[(\mathbf{\Omega} \oplus \mathbf{I})/x], H) \\
 &= \sum_{l+l'=n}   \sum_{H' \in \h}  \mathcal{H}^{l}(M, H')\cdot \Bigg( \sum_{H \in \h}  \mathcal{H}^{l'}(H'[(\mathbf{\Omega} \oplus \mathbf{I})/x], H)  \Bigg) \\
  &=    \sum_{\substack{l+l'=n\\  l'<n}} \sum_{H' \in \h} \mathcal{H}^{l}(M, H')\cdot \Bigg( \sum_{H \in \h}  \mathcal{H}^{l'}(H'[(\mathbf{\Omega} \oplus \mathbf{I})/x], H)  \Bigg) \\
 &\leq  \sum_{H' \in \h}  \mathcal{H}^{\infty}(M, H')\cdot \Bigg(  \sum_{H \in \h} \frac{1}{2}\cdot   \mathcal{H}^{\infty}(H'[\mathbf{\Omega}/x], H)\\
 &\phantom{\leq  \sum_{H' \in \h}  \mathcal{H}^{\infty}(M, H')\cdot \Bigg(  \sum_{H \in \h} \ }+ \frac{1}{2}\cdot  \mathcal{H}^{\infty}(H'[\mathbf{I}/x], H)\Bigg)\\
  &= \frac{1}{2}\cdot \sum_{H \in \h}  \mathcal{H}^{\infty}(M[\mathbf{\Omega}/x], H) + \frac{1}{2}\cdot \sum_{H \in \h} \mathcal{H}^{\infty}(M[\mathbf{I}/x], H)  .
\end{align*}
\end{proof}
\noindent
\textbf{Lemma~\ref{lem: counterexample M and N context preorder }.}   It holds that $M \leq_{\mathrm{cxt}} N$.
\begin{proof} By Lemma~\ref{lem: context lemma} it is enough to show that $M \leq_{\mathrm{app}}N$.  Since $M, N \in \elam$, this amounts to check that, for all $n \in \mathbb{N}$ and for all $L_1, \ldots, L_n \in \elam$, it holds that $  \sum \inter{M L_1\ldots L_n}\leq \sum \inter{NL_1\ldots L_n}$.
The proof is by induction on $n \in \mathbb{N}$.  \\
 If $n=0$ then, by Proposition~\ref{prop: the semantics is invariant under reduction}.\ref{enum: invariance abs} and Proposition~\ref{prop: the semantics is invariant under reduction}.\ref{enum: invariance sum}, we have: $ \sum \inter{M} =1 =  \frac{1}{2}\cdot \sum \inter{x\mathbf{\Omega}}+ \frac{1}{2}\cdot \sum \inter{x\mathbf{I}} 
 =\sum \inter{x \mathbf{\Omega} \oplus x \mathbf{I}}=\sum \inter{N}$.\\
Suppose $n=1$, and let us define $H^M \triangleq H[M /x]$, for all $M \in \plam$ and $H\in \h$. We have:
\allowdisplaybreaks
\begin{align*}
&\sum \inter{ML}=  \sum \inter{(\lambda x.x (\mathbf{\Omega} \oplus \mathbf{I}))L}\\
&= \sum \inter{L(\mathbf{\Omega} \oplus \mathbf{I})}&&\text{Prop.}~\ref{prop: the semantics is invariant under reduction}.\ref{enum: invariance beta}\\
&=\sum_{\lambda x.H \in\, \text{supp}(\inter{L})}  \inter{L}(\lambda x.H)\cdot \sum \inter{H^{\mathbf{\Omega} \oplus \mathbf{I}}}&&\text{Prop.}~\ref{prop: the semantics is invariant under reduction}.\ref{lem: invariance beta general case}\\
&\leq \frac{1}{2}\cdot \sum_{\lambda x.H \in\, \text{supp}(\inter{L})}  \inter{L}(\lambda x.H)\cdot  \sum  \inter{H^{\mathbf{\Omega}}} \\
& \phantom{\leq \ }+\frac{1}{2}\cdot \sum_{\lambda x.H \in\, \text{supp}(\inter{L})}  \inter{L}(\lambda x.H)\cdot \sum \inter{H^\mathbf{I}}&&\text{Eq.}~\eqref{enum: preliminary counterex 2}\\
%&= \frac{1}{2}\cdot \sum_{\lambda x.H \in\, \text{supp}(\inter{L})}  \inter{L}(\lambda x.H)\cdot \sum \inter{H[\mathbf{\Omega}/x]} \\
%&\phantom{=\ } +  \frac{1}{2}\cdot \sum_{\lambda x.H \in\, \text{supp}(\inter{L})}  \inter{L}(\lambda x.H)\cdot \sum\inter{H[ \mathbf{I}/x]} \\
&=\frac{1}{2}\cdot \sum\inter{L\mathbf{\Omega}}+ \frac{1}{2}\cdot\sum \inter{L\mathbf{I}}&&\text{Prop.}~\ref{prop: the semantics is invariant under reduction}.\ref{lem: invariance beta general case}\\
&=\frac{1}{2}\cdot \sum\inter{(x\mathbf{\Omega})^L}+ \frac{1}{2}\cdot \sum\inter{(x\mathbf{I})^L}\\
&= \sum_{\substack{H \in\, \text{supp}(\inter{x \mathbf{\Omega}}) \\  \cup\,   \text{supp}(\inter{ x \mathbf{I}})}}  \frac{1}{2} \cdot \Big(  \inter{x\mathbf{\Omega}} + \inter{x\mathbf{I}}\Big)(H)\cdot \sum \inter{H^L}\\
&= \sum_{H \in\,  \text{supp}\inter{x \mathbf{\Omega} \oplus x \mathbf{I}})} \inter{x\mathbf{\Omega} \oplus x\mathbf{I}}(H)\cdot \sum \inter{H^L} &&\text{Prop.}~\ref{prop: the semantics is invariant under reduction}.\ref{enum: invariance sum}\\
&= \sum_{\lambda x.H \in\, \text{supp}(\inter{N})}\inter{N}(\lambda x.H)\cdot\sum \inter{H^L}&&\text{Prop.}~\ref{prop: the semantics is invariant under reduction}.\ref{enum: invariance abs}\\
&=\sum\inter{NL} &&\text{Prop.}~\ref{prop: the semantics is invariant under reduction}.\ref{lem: invariance beta general case}  .
\end{align*}  
Finally, suppose $n>1$. We define:
\begin{align*}
P&\triangleq ML_1\ldots L_{n-1}\\
Q&\triangleq NL_1\ldots L_{n-1}\\
r& \triangleq \sum_{\lambda x. H \in \, \text{supp}(\inter{Q})}   \inter{Q}(\lambda x. H)\cdot \sum \inter{H[L/x]}\\
r'&=\sum_{\lambda x. H \in \, \text{supp}(\inter{P})}  \inter{P}(\lambda x. H)\cdot \sum \inter{H[L/x]} .
\end{align*}
Since by induction hypothesis $0 \leq \sum  \inter{Q} - \sum \inter{P}$, we have that $r-r'$ is positive. By  Proposition~\ref{prop: the semantics is invariant under reduction}.\ref{lem: invariance beta general case} this quantity is $\sum\inter{QL_n}- \sum\inter{PL_n}$. Therefore, $\sum \inter{PL_n}\leq \sum \inter{QL_n}$.
\end{proof}

\end{document}